\documentclass[12pt,reqno]{amsart}
\usepackage{amsmath,amssymb,amsfonts,amsthm}
\usepackage[mathscr]{eucal}
\usepackage[all]{xy}
\usepackage{hyperref}

\author{S. L. Lyakhovich, E. A. Mosman and A. A. Sharapov}

\address{Department of Quantum Field Theory, Tomsk State University, Lenin ave. 36, Tomsk 634050, Russia.}

\email{sll@phys.tsu.ru} \email{mosman@phys.tsu.ru}
\email{sharapov@phys.tsu.ru}

\title[Characteristic classes of $Q$-manifolds]{Characteristic classes of $Q$-manifolds:\\ classification  and applications}

\newtheorem{thm}{Theorem}[section]

\newtheorem{cor}{Corollary}
\newtheorem{lem}{Lemma}[section]
\newtheorem{prop}{Proposition}[section]
\theoremstyle{definition}
\newtheorem{defn}{Definition}[section]
\theoremstyle{remark}
\newtheorem{ex}{Example}[section]

\newtheorem{rem}{Remark}[section]

\renewcommand{\geq}{\geqslant}

\keywords{Q-manifolds, characteristic classes, gauge theories}

\begin{document}
\maketitle

\begin{abstract}
A $Q$-manifold $M$ is a supermanifold endowed with an odd vector
field $Q$ squaring to zero. The Lie derivative $L_Q$ along $Q$
makes the algebra of smooth tensor fields on $M$ into a
differential algebra. In this paper, we define and study the
invariants  of $Q$-manifolds called characteristic classes. These
take values in the cohomology of the operator $L_Q$ and, given an
affine symmetric connection with curvature $R$, can be represented
by universal tensor polynomials in the repeated covariant
derivatives of $Q$ and $R$ up to some finite order. As usual, the
characteristic classes are proved to be independent of the choice
of the affine connection used to define them. The main result of
the paper is a complete classification of the intrinsic
characteristic classes, which, by definition, do not vanish
identically on flat $Q$-manifolds. As an illustration of the
general theory  we interpret some of the intrinsic characteristic
classes as anomalies in the BV and BFV-BRST quantization methods
of gauge theories. An application to the theory of  (singular)
foliations is also discussed.
\end{abstract}

\section{Introduction}

By definition, a \textit{ $Q$-manifold } is a pair $(M,Q)$, where
$M$ is a smooth supermanifold equipped  with an odd vector field
$Q$ satisfying the integrability condition $[Q,Q]=0$. Every such
$Q$ is called a \textit{homological vector field}. Equivalently,
one can think of a $Q$-manifold as a smooth supermanifold whose
structure sheaf of supercommutative algebras of functions is
endowed with the differential $Q$. The action of $Q$ is naturally
extended from $C^\infty(M)$ to the whole tensor algebra of $M$.
Hereafter by the tensor algebra $\mathcal{T}(M)$ of a
supermanifold $M$ we mean the space of smooth tensor fields on $M$
endowed with the usual tensor operations: the tensor product,
contraction and permutation of indices. The Lie derivative $L_Q$
along $Q$ respects the tensor operations and makes
$\mathcal{T}(M)$ into a differential tensor algebra. We define the
group of $Q$-cohomologies (with tensor coefficients) as the quotient
$H_Q(M)=\mathrm{Ker}(L_Q)/\mathrm{Im}(L_Q)$. The tensor operations
in $\mathcal{T}(M)$ induce those in $H_Q(M)$; hence, we can speak
of the tensor algebra of $Q$-cohomology.

Let us mention the examples of $Q$-manifolds, which are important
in mathematics and physics.

\begin{ex}
 An odd tangent bundle $ \Pi TN$ of an ordinary manifold $N$.
This is obtained by applying the parity reversing functor $\Pi$ to
the tangent space of $N$. The algebra of smooth functions
$C^\infty(\Pi TN)$ is naturally isomorphic to the exterior algebra
of differential forms on $N$ and the role of $Q$ is played by the
exterior differential.
\end{ex}

\begin{ex}\label{LA}
 Replacing the tangent bundle of $N$ by a general Lie algebroid
$E \rightarrow N$, we come to the $Q$-manifold $\Pi E$.  The
differential algebra of smooth functions on $\Pi E$ is modeled on
$\Gamma(\wedge^\bullet E^\ast)$, the exterior  algebra of
$E$-forms, with the differential $Q$ being  the Lie algebroid
differential $d_E: \Gamma (\wedge^n E^\ast)\rightarrow \Gamma
(\wedge^{n+1} E^\ast)$. For a more detailed  discussion of the
relationship between Lie algebroids and homological vector fields
see \cite{Vaintrob}.
\end{ex}

\begin{ex}
Any $L_\infty$-algebra can be thought of as a formal $Q$-manifold
with the homological vector field $Q$ vanishing at the origin
\cite{AKSZ}, \cite{Kon1}. In a particular case of Lie algebras we
have a $Q$-manifold $\Pi \mathcal{G}$, where $\mathcal{G}$ is a
Lie the algebra and $Q$ is given by the Chevalley-Eilenberg
differential on $\wedge^\bullet \mathcal{G}^\ast \simeq
C^\infty(\Pi\mathcal{G})$. One can also view $\Pi \mathcal{G}$ as
a linear $Q$-manifold coming from the Lie algebroid
$\mathcal{G}\rightarrow \{\ast\}$ over a single point set.
\end{ex}

\begin{ex}
In the theory of gauge systems, the homological vector field
generates the BRST symmetry \cite{HT}, \cite{S}. To any
variational gauge system one associates either a symplectic or an
anti-symplectic manifold (according to which formalism, Lagrangian
BV or Hamiltonian BFV, is used) together with an odd,
selfcommuting, (anti-)Hamiltonian vector field $Q$. The
corresponding (anti-)symplectic two-form, being $Q$-invariant,
defines an element of the $Q$-cohomology group. More generally,
the quantization problem for non-variational gauge dynamics leads
naturally to flat $S_\infty$- or $P_\infty$-algebras whose first
structure map is given by  a classical BRST differential
\cite{CaFe}, \cite{KLS}, \cite{LS1}, \cite{LS2}.

\end{ex}

It should be noted that in many interesting cases, including the
examples above, a $Q$-manifold  carries an  additional
$\mathbb{Z}$-grading relative to which  $Q$ has degree 1. In this
case, $H_Q(M)$ is not just a differential group but a cochain
complex. Though important in particular applications, the
$\mathbb{Z}$-grading is of little significance  for our subsequent
considerations and we do not address  it here.

A morphism of $Q$-manifolds is a smooth map $\varphi:
M_1\rightarrow M_2$ inducing a homomorphism $\varphi^\ast :
C^\infty(M_2)\rightarrow C^\infty(M_1)$ of differential algebras,
i.e.,  the usual chain property holds: $Q_1\circ
\varphi^\ast=\varphi^\ast \circ Q_2$. In this case, the
homological vector fields $Q_1$ and $Q_2$ are said to be
$\varphi$-related. As a composition of two morphisms is apparently
a $Q$-morphism, we have a well-defined category of $Q$-manifolds.

In this paper, we define and study the invariants of $Q$-manifolds
called \textit{characteristic classes}\cite{LS}, \cite{LMS}. These take values in the
$Q$-cohomology and can be obtained by pulling back the
$Q$-cohomology classes of some formal $Q$-manifold of infinite
dimension (a classifying $Q$-space). The constructions of the
classifying space and the corresponding characteristic  map will
be presented in Secs. 3 and 4; for now, we would like to provide a
less formal introduction to the notion of the characteristic
classes in terms of natural covariants associated to a
$Q$-manifold with connection. Consider a $Q$-manifold $M$ endowed
with a symmetric affine connection $\nabla$ and let $R$ be the
curvature of $\nabla$. The tensor algebra $\mathcal{A}\subset
{\mathcal{T}}(M)$ of all local covariants associated to $Q$ and
$\nabla$ is generated by the repeated covariant derivatives of $Q$
and $R$.\footnote{In the context of classical differential
geometry, this statement is known as the second reduction theorem
\cite[p.165]{Schouten}, and the elements of $\mathcal{A}$ are
called \textit{differential concomitants} of $\nabla$ and $Q$, see
also \cite{JM}.} The algebra $\mathcal{A}$ is obviously invariant
under the action of $L_Q$ and thus it is a differential subalgebra
of $\mathcal{T}(M)$. We say that a $Q$-closed local covariant
$\mathcal{C}\in \mathcal{A}$ is a \textit{universal cocycle} if
the closedness equation $L_Q \mathcal{C}=0$ follows from the
integrability condition $[Q,Q]=0$ regardless of any specificity of
$Q$, $\nabla$ and $M$. In other words, the universal cocycles are
$Q$-invariant tensor polynomials in $\nabla^nQ$ and $\nabla^m R$
that can be attributed to \textit{any} $Q$-manifold with
connection.

The \textit{characteristic classes} of $Q$-manifolds can now be
defined as the elements of $H_Q(M)$ represented by universal
cocycles. It can be shown (Theorem \ref{independing thm}) that the
$Q$-cohomology classes of universal cocycles do not depend on the
choice of symmetric  connection and hence they are invariants of
the $Q$-manifold as such.

The simplest examples of the universal cocycles are the tensor
powers of the homological vector field $Q^{\otimes n}$. Obviously,
these cocycles exhaust all the universal cocycles that do not
involve the connection. A less trivial example of a universal
cocycle is obtained by taking the complete contraction of the
$2n$-form representing the Pontryagin characters of the tangent
bundle $TM$ with the contravariant tensor $Q^{\otimes 2n}$.  The
result is the sequence of  $Q$-invariant functions
$P_n=\mathrm{Str}(R^n_{QQ})$, where $R_{QQ}=[\nabla_Q,\nabla_Q]$
is a (1,1)-tensor defining an endomorphism of $TM$ and the
exponent $n$ means  the $n$th power of the endomorphism.

The cocycles $Q^{\otimes n}$ and $P_n$ are members  of two
complementary sets of universal cocycles: intrinsic and vanishing.
A universal cocycle is called \textit{vanishing} if it vanishes
identically upon setting the curvature $R$ to zero. From the
viewpoint of the $Q$-structure the most interesting are intrinsic
cocycles. The intrinsic cocycles survive on flat $Q$-manifolds,
that is why their cohomology classes are closely related to the
structure of the homological vector field rather than the topology
of $M$. This motivates us to introduce the notion of
\textit{intrinsic characteristic classes}. To give their formal
definition we recall that the tensor algebra $\mathcal{A}$ of
local covariants associated to $Q$ and $\nabla$ is constructed
from two infinite sequences of tensors $\{\nabla^nQ\}$ and
$\{\nabla^m R\}$. Let $\mathcal{R}$ denote the ideal of
$\mathcal{A}$ generated by $\{\nabla^m R\}$. Since $L_Q
\mathcal{R}\subset \mathcal{R}$, we have the short exact sequence
of complexes
\begin{equation*}
\xymatrix{0
\ar[r]&{\mathcal{R}}\ar[r]^-{i}&{\mathcal{A}}\ar[r]^-{p} &
{\mathcal{A}}/{\mathcal{R}}\ar[r]&0}
\end{equation*}
giving rise to the exact triangle in cohomology
\begin{equation}\label{ExTr}
\xymatrix{H(\mathcal{R})\ar[rr]^{i_\ast}& & H(\mathcal{A}) \ar[dl]^{p_\ast}\\
&H(\mathcal{A}/\mathcal{R})\ar[ul]^{\partial}& }
\end{equation}

Here $H(\mathcal{A})$ is the tensor algebra of characteristic
classes and $\partial$ is the connecting homomorphism. By
definition, the space of intrinsic characteristic classes is
identified with the subspace $\mathrm{Im}\,p_\ast =\mathrm{Ker}\,
\partial\subset H(\mathcal{A}/\mathcal{R})$. Geometrically, one can view the group
$H(\mathcal{A}/\mathcal{R})$  as the space of characteristic
classes of a flat $Q$-manifold. The universal cocycles of a flat
$Q$-manifold are built out of the $(n,1)$-tensors $\nabla^n Q$,
symmetric in lower indices. If $\partial$ is nonzero, not any
characteristic class can be extended from flat to arbitrary
$Q$-manifolds and the obstruction to extendability is controlled
by the elements of
$\mathrm{Im}\,\partial=\mathrm{Ker}\,i_\ast\subset
H(\mathcal{R})$. In Sec. 6, we show that the space
$\mathrm{Im}\,\partial$ is nonempty and spanned by the functions
$P_n$.

The main result of this paper is a complete classification of the
intrinsic characteristic classes. In the case of flat
$Q$-manifolds this is done in Sec. 5, where we construct a
multiplicative (w.r.t. the tensor product) basis in
$H(\mathcal{A})$. The universal cocycles representing the basis
elements are assembled into three infinite series $A$, $B$, and
$C$. The elements of $A$-series are represented by odd functions
on $M$, while the elements of $B$- and $C$-series are represented
by tensor fields of types $(n,1)$ and $(n,0)$, respectively, one
for each $n\in \mathbb{N}$.  The surprising thing is that one can
always  choose the basis cocycle to be tensor polynomials in
$\nabla Q$ and $\nabla^2 Q$ alone, i.e., no derivatives higher
than two of the homological vector field are needed to define all
the characteristic classes. This resembles the situation with the
characteristic classes of framed foliations or, more generally,
with the Gelfand-Fuks cohomology \cite{Fu1}, and this is more than
just an analogy. In fact, with the concept of classifying space,
we show that the enumeration problem for all independent
characteristic classes of flat $Q$-manifolds amounts to
computation of the stable cohomology of the Lie algebra of formal
vector fields with tensor coefficients, the problem that was
actually solved by D.B. Fuks \cite{Fu2}. The stable cohomology
groups under consideration are also known as {\textit{graph
cohomology}} groups, because they can be calculated via certain
complex of finite graphs. The graph complexes were introduced by
M. Kontsevich in the beginning of the nineties \cite{K1},
\cite{K2} and since then they have appeared in various contexts of
differential geometry and topology as well as in the topological
quantum field theory. It was stressed in \cite{Kon} that various
invariants (characteristic classes) of differential geometric
structures can be defined as the image of the graph cohomology in
the $Q$-cohomology of an appropriate $Q$-manifold (perhaps, with
an additional  structure). So the $Q$-manifolds provide a general
framework for most of the known constructions of characteristic
classes (Lie algebroids, vector bundles, foliations, complex
structures, knots, strongly homotopical algebras, rational
homotopy types, etc.). The further development of this framework
was one of our motives for writing this paper.

In Sec. 6 we extend the $A$-, $B$-, and $C$-series of
characteristic classes from flat to arbitrary $Q$-manifolds. The
extension is quite straightforward for the characteristic classes
of series $B$ and $C$, but it is up against some topological
obstructions for $A$-series. Roughly, only half of $A$-series'
characteristic classes can be defined for general $Q$-manifolds,
and the definition involves a special choice of the symmetric
affine connection. The obstructions for extendability of the other
half are explicitly identified with the aforementioned set of
$Q$-invariant functions $\{P_n\}$.

In Sec. 7 we focus on the exterior algebra of local covariants
with values in forms. The algebra has the structure of bicomplex
with respect to the exterior differential $d$ and the Lie
derivative $L_Q$. Under assumption that the Pontryagin classes of
$TM$ vanish, we prove that for any characteristic class
represented by a differential form there exists a $d$-exact
representative.

In Sec. 8 the general  construction is illustrated with some
examples from quantum field theory and theory of foliations.
Namely, we interpret the first term of $A$-series (the modular
class) and the second term of $C$-series as lower-order  anomalies
of gauge symmetries in the BRST quantization approach. We briefly
discuss a relationship between characteristic classes of
$Q$-manifolds, Lie algebroids, and (singular) foliations and give
an example of a regular foliation with nontrivial modular class.

Some auxiliary results are proved in Appendices A, B, and C.

\subsection*{Terminology and notation.} We use the standard language
and notation of supermanifold theory \cite{L}, \cite{L1}. Our
tendency, however, is to omit the prefix ``super'' whenever
possible. So the terms like manifolds, vector bundles, smooth
functions etc., will usually mean the corresponding notions of
supergeometry. Given a smooth supermanifold $M$ with $\dim M=p|q$,
we set $|\dim M|=p+q$. We let $\mathcal{T}(M)$ denote the algebra
of tensor fields, $\mathfrak{X}(M)$ the Lie algebra of vector
fields, $\Omega(M)$ the exterior algebra of differential forms,
and $\mathfrak{A}(M)$ the associative algebra of endomorphisms of
$TM$. The algebra $\mathfrak{A}(M)$, being a $\mathbb{Z}_2$-graded
algebra over $C^\infty(M)$, possesses a natural trace
$\mathrm{Str}: \frak{A}(M)\rightarrow C^\infty(M)$. Sometimes it
will be convenient to treat the  space $\frak{X}(M)$ as a right
$\mathfrak{A}(M)$-module and the space $C^{\infty}(M)$ as a left
$\mathfrak{X}(M)$-module. All the partial or covariant derivatives
are assumed to act from the left.

\subsection*{Acknowledgments} We wish to thank  R. Fernandes,  S. Merkulov and D. Roytenberg for useful discussions.
The work was partially supported by the RFBR grant no
09-02-00723-a, by the grant from Russian Federation President
Programme of Support for Leading Scientific Schools no
871.2008.02, by the State Contract no 02.740.11.0238 from Russian
Federal Agency for Science and Innovation, and also by Russian
Federal Agency of Education under the State Contract no P1337. EAM appreciates  the financial support from Dynasty Foundation.

\section{$Q$-vector bundles}

We begin with a collection of definitions and simple facts
concerning the concept of equivariant vector bundles \cite{Segal}.

Let $G$ be a Lie group (possibly zero or infinite dimensional). A
\textit{$G$-manifold } is a smooth manifold $M$ together with a
group homomorphism $G\rightarrow \mathrm{Diff}(M)$. In other
words, the Lie group $G$ acts smoothly on $M$. A
$G$-\textit{equivariant map} between $G$-manifold $M'$ and $M''$
is a smooth map $f: M'\rightarrow M''$ such that $f(gx)=gf(x) $
for all $g\in G$ and $x\in M'$. In what follows we will refer to
the $G$-equivariant maps as $G$-maps.

Analogously, consider a vector bundle $E$ over a smooth manifold
$M$ and let $\mathrm{Aut}(E)$ denote the group of bundle
automorphisms, i.e., the group of fiberwise linear diffeomorphisms
$f: E\rightarrow E$ mapping fibers to fibers\footnote{It is
usually assumed that $f|_M=id_M$. In our definition of the group
$\mathrm{Aut}(E)$ we admit a nontrivial action of $f$ on the base
$M$.}. A $G$-structure on $E$ is determined by a homomorphism
$G\rightarrow \mathrm{Aut}(E)$. Given a $G$-structure, we refer to
$E$ as a \textit{$G$-equivariant vector bundle} or just a
$G$-bundle for short. Note that both the total space $E$ and the
base $M$ of a $G$-bundle are $G$-manifolds and the canonical
projection $p: E\rightarrow M$ is a $G$-map.

The $G$-bundles  form a category (for a given $G$) whose morphisms
are $G$-equivariant bundle homomorphisms. Upon restricting to the
$G$-bundles over a fixed base manifold, we get a subcategory for
which all the usual operations on vector bundles are naturally
defined: the direct sum, tensor product, and dualization. The
parity reversion  functor $E\mapsto \Pi E$, being compatible with
bundle automorphisms, is one more natural operation in the
category of the $G$-vector bundles.

An even (odd) section $s:M\rightarrow (\Pi)E$ of a $G$-vector
bundle is called equivariant if it is a $G$-map. The equivariant
sections form a vector subspace $\Gamma^G(E)\subset \Gamma(E)$ in
the space of all smooth sections. They can also be viewed as fixed
points of the natural action of $G$ on $\Gamma(E)$.

$G$-vector bundles are of frequent occurrence. Here are some
examples.

\begin{ex}
Any tangent bundle $TM$ can be regarded as a
$\mathrm{Diff}(M)$-equivariant vector bundle and so are all the
associated tensor bundles. If $M$ is a $G$-manifold, then $TM$
carries the canonical $G$-structure induced by the homomorphisms
$G\rightarrow \mathrm{Diff}(M)\rightarrow \mathrm{Aut}(TM)$.
\end{ex}

\begin{ex}
Let $M$ be a $G$-manifold and  $\sigma: G\rightarrow GL(V)$ be a
representation of $G$ in a vector space $V$. Then $E=M\times V$ is
a trivial $G$-bundle; here $G$ acts on $E$ by
$g(x,v)=(gx,\sigma(g)v)$.
\end{ex}

\begin{ex}
If $E\rightarrow M$ is a vector bundle, then the $n$-fold tensor
product $E^{\otimes n}$ is an $S_n$-bundle, where $S_n$ is the
symmetric group permuting the factors of the product and $M$ is
regarded as a trivial $S_n$-manifold.
\end{ex}

In this paper, we mostly consider equivariant vector bundles
associated to the Lie group $\mathbb{R}^{0|1}$. If $\pi:
E\rightarrow M$ is such a vector bundle, then the infinitesimal
action of $\mathbb{R}^{0|1}$ on $E$ and $M$ is generated by some
homological vector fields $Q_E$ and $Q_M$, respectively; in so
doing, $Q_M$ appears to be a unique vector field $\pi$-related to
$Q_E$. For this reason, we will refer in sequel to
$\mathbb{R}^{0|1}$-equivariant vector bundles as $Q$-bundles or
vector bundles endowed with a $Q$-structure. Let
$\mathcal{T}(E)=\oplus_{n,m} \mathcal{T}^{n,m}(E)$ denote the
algebra of the $E$-tensor fields on $M$; by definition,
$\mathcal{T}^{n,m}(E)$ is the $C^{\infty}(M)$-module of sections
of $(E^\ast)^{\otimes n}\otimes E^{\otimes m}$. Then the action of
$Q_E$ on $E$, being fiberwise linear, endows  $\mathcal{T}(E)$
with the structure of a differential tensor algebra. The
differential $\delta: \mathcal{T}(E)\rightarrow \mathcal{T}(E)$ is
completely determined by its action on local coordinate functions
$\{x^i\}$ on $U\subset M$ and a frame of sections $\{s_a\}$ in
$E|_U$:
\begin{equation*}\label{}
    \delta x^i=Q_M^i(x)\,,\qquad \delta s_a=-(-1)^{\epsilon_a}\Lambda_a^b(x ) s_b\,.
\end{equation*}
Here $Q_M=\pi_\ast(Q_E)$ is the aforementioned  homological vector
field on the base $M$, and the odd matrix $\Lambda$ is called the
\textit{twisting element}. It follows from the relation
$\delta^2=0$ that in each trivializing chart the twisting element
obeys the Maurer-Cartan equation
\begin{equation}\label{MC}
    Q_M\Lambda =\Lambda^2 \,.
\end{equation}
If  $\{y^a\}$ are the fiber coordinates on $E$ dual to the frame
sections $\{s_a\}$, then the homological vector field $Q_E$ can be
locally written as
\begin{equation}\label{QE}
Q_E=Q_M^i(x)\frac{\partial}{\partial
x^i}+y^a\Lambda_a^b(x)\frac{\partial}{\partial y^b}\,.
\end{equation}

Writing $C^\infty_{lin}(E)$ for the space of smooth functions on
$E$ that are linear in the fiber coordinates, we can invariantly
characterize $Q_E$ as a homological vector field whose action
preserves the subspace $C_{lin}^\infty(E)\subset C^{\infty}(E)$
(and, as a consequence, the subalgebra $C^\infty(M)\subset
C^\infty(E)$).

Since $Q_E$ is odd, its action is always integrable to the action
of $\mathbb{R}^{0|1}$ and so there is a one-to-one correspondence
between the $\mathbb{R}^{0|1}$-equivariant vector bundles and the
vector bundles endowed with the action of the linear homological
vector field (\ref{QE}). From this perspective, a morphism of
$Q$-vector bundles is just a fiberwise linear map $\phi:
E_1\rightarrow E_2$ such that $Q_2=\phi_\ast(Q_1)$. In what
follows, it will be  convenient to refer to a $Q$-vector bundle as
a pair $(E,\delta)$, where $E$ is a vector bundle and $\delta$ is
a differential on the algebra of $E$-tensor fields. $Q$-invariant
$E$-tensors (i.e., $\mathbb{R}^{0|1}$-equivariant sections of the
associated tensor bundle) are by definition cocycles of the
differential $\delta$. Clearly, the corresponding group of
$\delta$-cohomology, defined by
$H(E,\delta)=\mathrm{Ker}\delta/\mathrm{Im}\delta$, inherits the
structure of tensor algebra, $H(E,\delta)=\bigoplus
H^{n,m}(E,\delta)$.

Let us give some examples of $Q$-vector bundles.

\begin{ex}
If $E$ is a vector bundle, then $(E,0)$ is a $Q$-vector bundle
with trivial differential.
\end{ex}

\begin{ex}
The tangent bundle $TM$ of a $Q$-manifold has a canonical
$Q$-structure defined by the Lie derivative $L_Q$ along $Q$.
\end{ex}

\begin{ex}
Let $E$ be a vector bundle over a $Q$-manifold and suppose $E$ to
admit a flat connection $\nabla$. Then $\nabla_Q^2=0$, and we have
the $Q$-vector bundle $(E,\nabla_Q)$.
\end{ex}

Any morphism $\phi: E_1\rightarrow E_2$ of $Q$-vector bundles
induces a homomorphism on sections\footnote{We simply identify
$\Gamma(E^\ast)$ with $C_{lin}^{\infty}(E)$, then $\phi_\ast$ is
given by  the pullback of $\phi$.}
\begin{equation}\label{hom}
    \phi_\ast: \Gamma(E^\ast_2)\rightarrow \Gamma(E^\ast_1)\,,
\end{equation}
where $E^\ast_{1,2}$ are the $Q$-vector bundles dual to $E_{1,2}$.
The homomorphism $\phi_\ast$ in its turn gives rise to a
homomorphism of the cohomology groups
$\mathrm{Ker}\delta/\mathrm{Im}\delta$ of the differential
$\delta$ on the space of sections $\Gamma(E^\ast_{1,2})$.
Generally there is no natural way to extend (\ref{hom}) to the
full algebras of $E$-tensor fields except when $\phi$ is a
fiberwise isomorphism. In this last case, we have a unique
homomorphism $\widetilde{\phi}_\ast: \Gamma(E_2)\rightarrow
\Gamma(E_1)$ such that
\begin{equation*}\label{}
    \langle \phi_\ast(u),\widetilde{\phi}_\ast(v)\rangle_1=\langle
    u,v\rangle_2\circ \phi|_{M_1}\,, \qquad \forall  u\in \Gamma(E^{\ast}_2)\,,\quad
    \forall v\in \Gamma(E_2)\,.
\end{equation*}
Here the triangle brackets $\langle\cdot , \cdot \rangle $ stand
for pairing between  the spaces $ \Gamma(E)$ and $\Gamma(E^\ast)$.
The pair $\varphi=(\phi_\ast,\widetilde{\phi}_\ast)$ defines a
homomorphism $\varphi: \mathcal{T}(E_2)\rightarrow
\mathcal{T}(E_1)$ of differential tensor algebras. Thus, we get a
homomorphism of $\delta$-cohomology groups
\begin{equation*}\label{}
    \varphi_\ast: H(E_2,\delta_2)\rightarrow H(E_1,\delta_1)\,.
\end{equation*}
We emphasize  that the last homomorphism can be defined only when
$\mathrm{rank} E_1=\mathrm{rank} E_2$ and $\phi: E_1\rightarrow
E_2$ is a fiberwise isomorphism.

Let us now specify the constructions above to the tangent bundle
of a $Q$-manifold. In this case, we have a canonical $Q$-structure
on $TM$ identified with the operator of Lie derivative
$\delta=L_Q$. The differential tensor algebra of $M$ was denoted
by $\mathcal{T}(M)$ and the $Q$-cohomology group was denoted by
$H_Q(M)$ (see the previous section). Upon choosing a symmetric
affine connection $\nabla$ on $M$, we can define the algebra of
local covariants $\mathcal{A}\subset \mathcal{T}(M)$. As a tensor
algebra $\mathcal{A}$ is generated by  two sequences of tensor
fields $\{\nabla^n Q\}$ and $\{\nabla^n R\}$,  $R$ being the
curvature of $\nabla$. Since $\delta \mathcal{A}\subset
\mathcal{A}$, $\mathcal{A}$ is a differential subalgebra of
$\mathcal{T}(M)$ with the cohomology group $H(\mathcal{A})$. The
natural inclusion
$$
i: \mathcal{A}\rightarrow \mathcal{T}(M)
$$
induces the homomorphism in cohomology
$$
i_\ast: H(\mathcal{A})\rightarrow H_Q(M)\,.
$$

In the previous section, we gave a preliminary definition of the
characteristic classes of $Q$-manifolds as elements of
$\mathrm{Im}\, i_\ast$ that are represented by the so-called
universal cocycles. The adjective ``universal'' implies  that the
$\delta$-closedness condition  is satisfied by virtue of the
integrability condition $[Q,Q]=0$ alone. With this interim
definition, we can readily show the independence of the
characteristic classes of the choice of $\nabla$.

\begin{thm}\label{independing thm} The characteristic classes of a
$Q$-manifold  do not depend on the choice of symmetric connection
and hence they are invariants of the $Q$-manifold itself.
\end{thm}
\begin{proof} Let $\mathcal{C}_{\nabla_0}[Q]$ and
$\mathcal{C}_{\nabla_1}[Q]$ be two universal cocycles that differ
only by the choice of the connection. Consider the direct product
of $M$ and the linear superspace $\mathbb{R}^{1|1}$ with one even
coordinate $t$ and one odd coordinate $\theta$. The product
structure of $\widetilde{M}=M\times \mathbb{R}^{1|1}$ induces the
decomposition of the linear space of tensor fields:
\begin{equation*}
\mathcal{T}(\widetilde{M})=\mathcal{T}^{\ '}(\widetilde{M})
\oplus\mathcal{T}^{\ ''}(\widetilde{M})\,,
\end{equation*}
where $\mathcal{T}^{\ '}(\widetilde{M})$ is the space of sections
of the vector bundle
\begin{equation*}
T^{\bullet, \bullet}M\times\mathbb{R}^{1|1}\longrightarrow
\widetilde{M}\,.
\end{equation*}
Simply stated, the elements of $\mathcal{T}^{\ '}(\widetilde{M})$
are the smooth families of tensor fields on $M$ parameterized by
``points'' of $\mathbb{R}^{1|1}$. So, we have two natural
projections
\begin{equation}\label{projectors}
\pi{'}:\mathcal{T}(\widetilde{M})\longrightarrow\mathcal{T}^{\
'}(\widetilde{M}), \qquad
\pi{''}:\mathcal{T}(\widetilde{M})\longrightarrow\mathcal{T}^{\
''}(\widetilde{M}).
\end{equation}

Equip the supermanifold $\widetilde{M}=M\times \mathbb{R}^{1|1}$
with the homological vector field  $ \widetilde{Q}=Q +\theta
\partial_t $ and an adapted connection $\widetilde{\nabla}$. The latter is completely specified
by the covariant derivatives:
\begin{equation*}
\widetilde{\nabla}_{\frac{\partial}{\partial t}}V = \partial_t
V\,,\quad \widetilde{\nabla}_{\frac{\partial}{\partial\theta}}V =
{\partial_\theta V}\,,\quad
\widetilde{\nabla}_XV=\nabla^t_X(\pi'V)+X(\pi{''}V)\,,
\end{equation*}
where $V\in \mathfrak{X}(\widetilde{M})$,  $X\in \mathfrak{X}(M)$
and  $\nabla^t=t\nabla_1+(1-t)\nabla_0$ is the one-parameter
family of connections on $M$. Clearly, the operator
$\widetilde{\delta}=\mathcal{L}_{\widetilde{Q}}$ commutes with the
projectors (\ref{projectors}):
\begin{equation}\label{proj-delta}
\widetilde{\delta} \ \pi ' = \pi ' \widetilde{\delta}\,, \qquad
\widetilde{\delta} \ \pi'' = \pi'' \widetilde{\delta}\,.
\end{equation}

Consider now the universal cocycle
$\mathcal{C}_{\widetilde{\nabla}}[\widetilde{Q}]$. Due to the
specific structure of $\widetilde{Q}$ and $\widetilde{\nabla}$ we
have
\begin{equation}\label{proj-C}
\pi'(\mathcal{C}_{\widetilde{\nabla}}[\widetilde{Q}]) =
\mathcal{C}_{{\nabla^t}}[Q]+\theta \Psi
\end{equation}
for some $\Psi \in \mathcal{T}^{\ '}(\widetilde{M})$ obeying
$\partial_\theta \Psi=0$. Since
$\widetilde{\delta}\mathcal{C}_{\widetilde{\nabla}}[\widetilde{Q}]=0$,
it follows from (\ref{proj-delta}) and (\ref{proj-C}) that
\begin{equation*}
\widetilde{\delta} \ \pi'
\mathcal{C}_{\widetilde{\nabla}}[\widetilde{Q}] =
\widetilde{\delta} (\mathcal{C}_{{\nabla^t}}[Q]+\theta \Psi)= 0\,.
\end{equation*}
The last equation is equivalent to the following ones:
\begin{equation*}\label{}
    \delta (\mathcal{C}_{{\nabla^t}}[Q])=0\,,\qquad \partial_t
    \mathcal{C}_{\nabla^t}[Q]=\delta \Psi\,.
\end{equation*}
Integrating the second equation by $t$ from $0$ to $1$, we get
\begin{equation*}\label{}
\mathcal{C}_{\nabla_1}[Q] - \mathcal{C}_{\nabla_0}[Q]=\delta
\int_0^1 dt \Psi\,.
\end{equation*}
Thus, the $\delta$-cohomology class of the universal cocycle
$\mathcal{C}_{\nabla}[Q]$ does not depend on the choice of
symmetric connection.
\end{proof}

\section{The classifying $Q$-space}

Let $V$ be a finite-dimensional superspace with coordinates
${y^i}$. Denote by $L_0(V)$ the Lie algebra of formal vector
fields on $V$ vanishing at the origin. The generic element of
$L_0(V)$ reads
\begin{equation*}\label{v}
v=\sum_{n=1}^\infty y^{i_n}\cdots y^{i_1}v_{i_1\cdots
i_n}^j\frac{\partial}{\partial y^j}\,.
\end{equation*}
One can regard the expansion coefficients $v_{i_1\cdots i_n}^j\in
\mathbb{R}$ as coordinates in the infinite-dimensional superspace
$L_0(V)=L^0(V)\oplus L^1(V)$.

As usual, we can associate to  $L_0(V)$ the linear manifold
$\mathbb{M}=\Pi L_0(V)$ with coordinates $\{c_{i_1\cdots
i_n}^j\}$. By definition, $\epsilon(c_{i_1\cdots
i_n}^j)=\epsilon(v_{i_1\cdots i_n}^j)+1$. The Lie algebra
structure on ${L}_0(V)$ is then encoded by the homological vector
field
\begin{equation}\label{QQ}
\mathbb{Q}= \sum_{n=1}^{\infty}
\sum_{l=1}^{n}\binom{n}{l}(-1)^{\epsilon_{i_1}+\ldots+\epsilon_{i_l}}c_{i_{1}\ldots
i_{l}}^{m} c_{m i_{l+1}\ldots i_{n}}^{j}\frac{\partial}{\partial
c_{i_{1}\ldots i_{n}}^{j}}
\end{equation}
on $\mathbb{M}$. Besides, $\mathbb{M}$ is provided with the
natural action of $GL(V)$. Since $\mathbb{Q}$ is obviously
invariant under the $GL(V)$-transformations, one can think of it
as an equivariant section of $GL(V)$-vector bundle $T\mathbb{M}$.

Taking ${V}$ as typical fiber, consider the trivial vector bundle
$\mathbb{E}=\mathbb{M}\times V$ with the diagonal action of
$GL(V)$. We can assign $\mathbb{E}$ with a $Q$-structure starting
with the homological vector field $\mathbb{Q}$ on the base. To
this end, we need to specify a twisting element $\Lambda\in
\mathcal{T}^{1,1}(\mathbb{E})$ satisfying the Maurer-Cartan
equation (\ref{MC}). Using the natural frame
$v_i=\partial/\partial y^i$ in $V$, we set
\begin{equation}\label{twist}
\delta v_i= -(-1)^{\epsilon_i}c^j_iv_j\,.
\end{equation}
Since the action of the homological vector field commutes with the
general linear transformations, we can regard $\mathbb{E}$ as an
$\mathbb{R}^{0|1}\times GL(V)$-equivariant vector bundle. For
reasons clarified below we refer to $(\mathbb{E}, \delta)$ as a
\textit{classifying $Q$-space}.

Associated to $\mathbb{E}$ is the differential algebra
$\mathcal{T}(\mathbb{E})=\bigoplus \mathcal{T}^{n,m}(\mathbb{E})$
of $\mathbb{E}$-tensor fields. Denote by
$\mathcal{T}(\mathbb{E})^{\mathrm{inv}}\subset
\mathcal{T}(\mathbb{E})$ the tensor subalgebra of
$GL(V)$-equivariant sections of $T^{\bullet,\bullet}\mathbb{E}$
or, what is the same, $GL(V)$-invariant $\mathbb{E}$-tensors. The
subalgebra $\mathcal{T}(\mathbb{E})^{\mathrm{inv}}$ is also
invariant under the action of the differential $\delta$. Therefore
we can speak about the differential tensor algebra
$(\mathcal{T}(\mathbb{E})^{\mathrm{inv}}, \delta)$  and the
corresponding algebra of $\delta$-cohomology
$H(\mathbb{E})^{\mathrm{inv}}$. It should be emphasized that we
treat $\mathbb{M}$ as a formal manifold, so the components of
$\mathbb{E}$-tensors are given by formal power series in the
coordinates $\{c_{i_1\cdots i_n}^j\}$. As an example, consider the
following sequence of $GL(V)$-invariant tensor fields with linear
dependence of coordinates:
\begin{equation}\label{basis}
    C_n= dy^{i_n}\otimes \cdots \otimes dy^{i_1}c_{i_1\cdots i_n}^j\frac{\partial}{\partial
    y^j} \;\in\; \mathcal{T}^{n,1}(\mathbb{E})^{\mathrm{inv}}\,.
\end{equation}

It follows from the first main theorem of invariant theory
\cite{Weyl}, \cite{Fu1} that $\mathbb{E}$-tensors (\ref{basis})
constitute a multiplicative basis of
$\mathcal{T}(\mathbb{E})^{\mathrm{inv}}$ so that any
$GL(V)$-invariant tensor is made up algebraically of $\{C_n\}$.
The space $\mathcal{T}(\mathbb{E})^{\mathrm{inv}}$ is naturally
graded by the subspaces $\mathcal{T}_r(\mathbb{E})^{\mathrm{inv}}$
consisting of the homogeneous tensor polynomials of degree $r$ in
$\{C_n\}$. This grading makes the space of $GL(V)$-invariant
tensor fields into the cochain complex: $$\delta:
\mathcal{T}_r(\mathbb{E})^{\mathrm{inv}}\rightarrow
\mathcal{T}_{r+1}(\mathbb{E})^{\mathrm{inv}}\,.$$

Besides, we have an increasing filtration of
$\mathcal{T}(\mathbb{E})^{\mathrm{inv}}$ by the sequence of
subcomplexes
\begin{equation}\label{Tfilt}
   0\subset  F_1\mathcal{T}(\mathbb{E})^{\mathrm{inv}}\subset
    F_2\mathcal{T}(\mathbb{E})^{\mathrm{inv}}\subset\cdots\subset
    F_\infty \mathcal{T}(\mathbb{E})^{\mathrm{inv}}=\mathcal{T}(\mathbb{E})^{\mathrm{inv}} \,,
\end{equation}
where the $\mathbb{E}$-tensors from  $F_n
\mathcal{T}(\mathbb{E})^{\mathrm{inv}}$ are generated by $C_1,...,
C_n$.

Observe that the tensor algebra
$\mathcal{T}(\mathbb{E})^{\mathrm{inv}}$ is not freely  generated
by $\{C_n\}$ and it is not hard to write some tensor polynomials
in $C$'s that vanish identically for some values of  $\dim V$.
Denote the space of all such polynomials  by $I$. Then, the second
main theorem of invariant theory \cite{Weyl}, \cite{Fu1} ensures
that $I\cap F_k \mathcal{T}_r(\mathbb{E})^{\mathrm{inv}}=0$ for
$|\dim V| \gg k,r$. Informally speaking, the algebra
$\mathcal{T}(\mathbb{E})^{\mathrm{inv}}$ becomes free as the
dimension of $V$ goes to infinity. We will return to this point in
Sec. 5.

\section{The characteristic map}

Let $M$ be a flat $Q$-manifold, i.e., a smooth manifold endowed
with a homological vector field $Q$ and a flat symmetric
connection $\partial$. We can assume without loss of generality
that $M$ is simply connected (otherwise replace $(M,Q)$ by its
universal covering $Q$-manifold $(\widetilde{M}, \widetilde{Q})$,
where $\widetilde{Q}$ is the lift of $Q$ with respect to the
covering map\footnote{To define the universal covering of a
supermanifold $M$, we identify $M$ with the total space of an odd
vector bundle $E\rightarrow M_B$, where $M_B$ is the  body of $M$.
If $\widetilde{M}_B$ is a universal covering of $M_B$ and $p:
\widetilde{M}_B\rightarrow M_B$ is the corresponding projection,
then the universal covering of $M$ is, by definition, the
supermanifold $\widetilde{M}$ associated to the total space of the
pullback bundle $p_* (E)$.} $p: \widetilde{M}\rightarrow M$).
Under these assumptions $M$ is a parallelizable manifold and the
Lie derivative $\delta = L_Q$ makes $TM$ into a trivial $Q$-vector
bundle endowed with a flat connection.

In this section, we use the data above to construct a
characteristic map of $TM$ to the classifying $Q$-space
$\mathbb{E}=\mathbb{M}\times V$, where $V$ is the typical fiber of
$TM$. The construction  is not completely canonical and depends on
the choice of a trivialization $\varphi: TM\rightarrow M\times V$.
Nonetheless, we show that the pullback of the characteristic map
gives rise to a well-defined homomorphism in the cohomology
$H(\mathbb{E})^{\mathrm{inv}}\rightarrow H_Q(M)$ that depends
neither on the flat connection nor on the choice of
trivialization. The construction goes as follows.

Let  $\varphi: TM \rightarrow M\times V$ be a  trivialization and
$\varphi_\ast :  \Gamma(M\times V)\rightarrow  \mathfrak{X}(M)$ is
the induced homomorphism on sections. Given a flat symmetric
connection $\partial$, we say that the trivialization $\varphi$ is
\textit{compatible} with $\partial$, if the pullback of any
constant section  $v\in \Gamma(M\times V)$ is a covariantly
constant section of $TM$, i.e.,  $\partial_X \varphi_\ast(v)=0$
for all $ X\in \mathfrak{X}(M)$. Clearly, the pullback of constant
sections defines a global frame in $TM$, which is completely
determined by its value at any single point $p\in M$. (The frame
in $T_qM$ results from the parallel transport of a frame in $T_pM$
along any path joining $p$ to $q$.) Since any two frames in $T_pM$
are related to each other by an invertible linear transformation,
we can identify the set of all the compatible trivializations with
the group $GL(V)$.

The isomorphism $\varphi$ is naturally prolonged to the
isomorphism of associated tensor bundles
\begin{equation*}
\varphi: T^{n,m}M\rightarrow M\times V^{n,m}\,.
\end{equation*}
Here $V^{n,m}=V^{\otimes n}\otimes (V^\ast)^{\otimes m}$ is the
standard $GL(V)$-module.   Given an even (odd) section $s:
M\rightarrow (\Pi)T^{n,m}M$, one can define the so-called Gauss
map $M\rightarrow (\Pi) V^{n,m}$ through the composition of maps
\begin{equation*}
\xymatrix@1{{M}\ar[r]^-s&{(\Pi)T^{n,m}M}\ar[r]^-{\varphi}&
{M\times (\Pi)V^{n,m}}\ar[r]^-{p}& {(\Pi)V^{n,m}}}\,,
\end{equation*}
with $p$ being the projection onto the second factor.   Given the
homological vector field $Q$ and  the flat symmetric connection
$\partial$, we can build  the sequence of tensors $\partial^n Q\in
{\mathcal{T}}^{n,1}(M)$. Taken together the tensor fields
$\{\partial^nQ\}_{n=1}^\infty$ define the Gauss map
\begin{equation}\label{Gauss}
{\mathbb{G}}: M \rightarrow  \mathbb{M}
\end{equation}
of the flat $Q$-manifold $M$ to the base $\mathbb{M}$ of the
classifying $Q$-space $\mathbb{E}$. Note that the Gauss map
$\mathbb{G}$ depends on the trivialization $\varphi$, even though
we do not indicate this explicitly. Let $\{x^i\}$ be  a local
coordinate system on $M$ adapted to the flat connection $\partial$
in the sense that the Christoffel symbols of $\partial$ vanish and
let $V$ spans the coordinate vector fields $\{\partial/\partial
x^i\}$. Then the coordinate expression of the map (\ref{Gauss}) is
\begin{equation}\label{Gauss-in-coord}
c^j_{i_1\cdots i_n}=\partial_{i_1}\cdots
\partial_{i_n}Q^j(x)\,,\qquad n=1,2,...\,.
\end{equation}
Further, we can trivially extend the Gauss map ${\mathbb{G}}$  to
the bundle map
\begin{equation}\label{Charmap}
 \widehat{{\mathbb{G}}}: TM \rightarrow \mathbb{E}
\end{equation}
by setting
\begin{equation}\label{Charmap1}
\widehat{{\mathbb{G}}}: \xymatrix@1{TM\ar[r]^-{\varphi}&M\times
V\ar[r]^-{{\mathbb{G}}\times \mathrm{id}}& \mathbb{M}\times
V}=\mathbb{E}\,.
\end{equation}
Again, the map $\widehat{{\mathbb{G}}}$ depends on the chosen
trivialization $\varphi$. Since $\widehat{{\mathbb{G}}}$ is a
fiberwise isomorphism, it gives rise to the pullback map on
sections, which then extends to the  homomorphism of the full
tensor algebras
\begin{equation}\label{Charmap2}
\widehat{{\mathbb{G}}}_\ast: \mathcal{T}(\mathbb{E})\rightarrow
\mathcal{T}(M)\,.
\end{equation}

\begin{thm}\label{Th2}
The map $\widehat{{\mathbb{G}}}$ defined by Eq.(\ref{Charmap1}) is
a morphism of $Q$-vector bundles.
\end{thm}

\begin{proof} We only need to show that the homomorphism (\ref{Charmap2})
obeys the chain property
\begin{equation}\label{ChainPr}
\delta \circ\widehat{{\mathbb{G}}}_\ast =
\widehat{{\mathbb{G}}}_\ast \circ\delta\,.
\end{equation}
In view of the Leibniz rule, it is enough to check the last
operatorial identity on the coordinate functions on $\mathbb{M}$
and the  frame sections  of $\mathbb{E}$.  This can be done
directly by making use of the  coordinate description of
$\widehat{{\mathbb{G}}}$ and $\delta$ given by Eqs. (\ref{QQ}),
(\ref{twist}), (\ref{Gauss-in-coord}). Applying (\ref{ChainPr}) to
the coordinate functions $c_{i_1\cdots i_n}^j\in
C^{\infty}(\mathbb{M})$, we get
\begin{equation}\label{binom}
Q^m\partial_m \partial_{i_1}\cdots\partial_{i_n}Q^j=
\sum_{l=1}^{n} \binom{n}{l}
(-1)^{\epsilon_{i_1}+\cdots+\epsilon_{i_l}+1}
\partial_{(i_{1}}\cdots \partial_{i_{l}}Q^{m} \partial_{m}\partial_{i_{l+1}}\cdots
\partial_{i_{n})}Q^{j},
\end{equation}
where the parentheses around indices denote symmetrization in the
graded sense. But the last equality is just the differential
consequence  of the integrability condition for the homological
vector field:
\begin{equation*}\label{}
    \partial_{i_1}\cdots\partial_{i_n}Q^2=0\,.
\end{equation*}

If now $\{y^i\}$ are linear coordinates in $V$ such that
$\widehat{{\mathbb{G}}}_\ast(\partial/\partial
y^i)=\partial/\partial x^i$, then
$$
    (\delta \circ \widehat{\mathbb{G}}_\ast)\left(\frac{\partial}{\partial
    y^i}\right)=-(-1)^{\epsilon_i}\frac{\partial Q^j}{\partial x^i}\frac{\partial}{\partial
    x^j},
$$
and the same expression results from the right hand side of
(\ref{ChainPr}):
\begin{equation*}
\begin{array}{l}
\displaystyle
(\widehat{\mathbb{G}}_\ast\circ\delta)\left(\frac{\partial}{\partial
    y^i}\right)=\widehat{\mathbb{G}}_\ast\left((-1)^{1+\epsilon_i}c^j_i\frac{\partial}{\partial
    y^j}\right)\\[4mm]
\displaystyle
=\widehat{\mathbb{G}}_\ast((-1)^{1+\epsilon_i}c^j_i)\cdot\widehat{\mathbb{G}}_\ast\left(\frac{\partial}{\partial
y^j}\right)
    =(-1)^{1+\epsilon_i}\frac{\partial Q^j}{\partial x^i}\frac{\partial}{\partial
    x^j}\,.
\end{array}
\end{equation*}
\end{proof}
Of course, the proof above is rather technical and involves a
local coordinate consideration. A more ``conceptual"  proof of
Theorem \ref{Th2} is given in Appendix A.

Given a flat $Q$-manifold $M$, we call (\ref{Charmap}) a
\textit{characteristic map}.

 To simplify our notation we will denote
the restriction of $\widehat{{\mathbb{G}}}_\ast$ to
$\mathcal{T}(\mathbb{E})^{\mathrm{inv}}\subset
\mathcal{T}(\mathbb{E})$ by
\begin{equation}\label{chi}
\chi:\mathcal{T}(\mathbb{E})^{\mathrm{inv}}\rightarrow
\mathcal{T}(M)\,.
\end{equation}
 As a  differential tensor algebra,
$\mathcal{T}(\mathbb{E})^{\mathrm{inv}}$ is generated by the
linear $GL(V)$-invariant tensor fields (\ref{basis}). The
homomorphism $\chi$ assigns to each  generator $C_n\subset
\mathcal{T}^{n,1}(\mathbb{E})^{\mathrm{inv}}$, $n>0$,  the tensor
field
\begin{equation}\label{cov}
\chi (C_n)=\partial^n Q\in \mathcal{T}^{n,1}(M)
\end{equation}
given by the covariant derivatives of the homological vector
field. Supplementing the tensors (\ref{cov}) by the homological
vector field itself, we get the full set of generators for the
differential tensor algebra of local covariants
$\mathcal{A}\subset \mathcal{T}(M)$ associated to a flat
$Q$-manifold (see Sec. 1). The homomorphism (\ref{chi}) induces a
well-defined homomorphism in cohomology
\begin{equation}\label{chii}
    \chi_\ast :
H(\mathbb{E})^{\mathrm{inv}}\rightarrow H_Q(M)\,.
\end{equation}

\begin{thm}
The homomorphism  (\ref{chii}) is independent of a flat connection
and a compatible trivialization.
\end{thm}

\begin{proof} The homomorphism $\chi_\ast$ takes a $\delta$-cocycle $A\in
\mathcal{T}(\mathbb{E})^{\mathrm{inv}}$ to the universal
$\delta$-cocycle $\chi_\ast(A)\in \mathcal{A}$ and the
independence of a flat connection follows from Theorem
\ref{independing thm}.

If now $M\times V$ and $M\times V'$ are two compatible
trivializations of $TM$, then $V'=g(V)$ for some linear
isomorphism $g$. Thus the group $GL(V)$ acts transitively on the
set of all trivializations compatible with a given flat
connection. This action then translates to the natural action of
$GL(V)$ on $\mathbb{E}$ and, by definition, leaves invariant  the
tensor fields of $\mathcal{T}(\mathbb{E})^{\mathrm{inv}}$.
\end{proof}

\begin{rem}
Actually, the homomorphism (\ref{chi})  makes sense for arbitrary
flat $Q$-manifolds, not necessarily simply connected: It just
defines a way of constructing local covariants from the elementary
ones $\{\partial^nQ\}$. The chain property $\chi\circ
\delta=\delta\circ\chi$ is apparently a local condition following
from the identity (\ref{binom}). However, for a simply connected
$Q$-manifold $M$ we have a nice geometric interpretation for
(\ref{chi}), (\ref{chii}) as homomorphisms induced by a morphism
of $Q$-vector bundles.
\end{rem}

To give the final definition of the characteristic classes we need
some facts concerning the structure of the tensor algebra
$\mathcal{A}$. Let $\mathcal{A}'\subset \mathcal{A}$ denote the
subalgebra generated by the homological vector field $Q$ (as a
linear space, $\mathcal{A}'$ is spanned by the tensor powers
$\{Q^{\otimes n}\}$) and let $\mathcal{A}''\subset \mathcal{A}$ be
the subalgebra generated by the elementary covariants $\partial^n
Q$ with $n>0$.

\begin{prop} With the definitions above,

    (a) $\mathcal{A}''=\mathrm{Im}\,\chi$;

    (b) both $\mathcal{A}'$ and $\mathcal{A}''$ are
    differential subalgebras of $\mathcal{A}$;

    (c) $\mathcal{A}=\mathcal{A}'\otimes \mathcal{A}''$, and hence $H(\mathcal{A})= H(\mathcal{A}')\otimes
    H(\mathcal{A}'')$.

\end{prop}

\begin{proof} The equality $(a)$ is obvious and
 the subalgebra $\mathcal{A}''\subset \mathcal{A}$ is closed
under the action of $\delta$ as the homomorphic image of the
differential tensor algebra
$\mathcal{T}(\mathbb{E})^{\mathrm{inv}}$. It is also clear that
$\delta \mathcal{A}'=0$ and $(b)$ is proved.

Next we note that  the elements of $\mathcal{A}$ are obtained by
applying the tensor operations to the generators
$\{\partial^nQ\}$; in so doing, we can ignore the elements
involving contractions of the homological vector field $Q$ with
$\partial^n Q$. Indeed, as is seen from Eq. (\ref{binom}), any
such contraction is equal to an element of $\mathcal{A}''$ and we
are lead to conclude that each element of $\mathcal{A}$ is
uniquely represented by a linear combination of tensor products
$u\otimes w$, where $u\in \mathcal{A}'$ and $w\in \mathcal{A}''$.
Applying the K\"unneth formula to the tensor product of complexes
$\mathcal{A}'\otimes \mathcal{A}''$ completes the proof of $(c)$.
\end{proof}

As is seen, the  contribution of the homological vector field
\textit{per se} to the algebra of local covariants as well as its
cohomology can trivially be factor out, so that  the most
interesting universal cocycles (i.e. ones involving derivatives of
the homological vector field) are centered in $\mathcal{A}''$.
This motivates us to identify the characteristic classes of flat
$Q$-manifolds with the group  $H(\mathcal{A}'')$ rather than the
whole group $H(\mathcal{A})$ as it was done in the Introduction.
Then we have the following

\begin{defn}
The characteristic classes of a flat $Q$-manifold are the
$Q$-cohomology classes belonging to the image of the homomorphism
(\ref{chii}).
\end{defn}

\section{Stable characteristic classes and graph complexes}\label{StGr}

The characteristic classes will carry a valuable piece of
information about the structure of $Q$-manifolds provided that the
$\delta$-cohomology groups $H^{n,m}_r(\mathbb{E})^{\mathrm{inv}}$
of the classifying space are simultaneously wide enough  and
effectively computable. The next logical step is thus  to find an
explicit description for the equivariant cohomology of
$\mathbb{E}$. Unfortunately, the computation of the groups
$H^{n,m}_r(\mathbb{E})^{\mathrm{inv}}$ for arbitrary $m,n,r$ and
$\dim V$ appears to be a hard problem yet to be solved. The
problem, however,  becomes much simpler in the so-called
\textit{stable range of dimensions}, where by stability we mean
$|\dim V|\gg r,n$. For an even vector space  $V$ the corresponding
stable cohomologies were computed by Fuks \cite{Fu2} (see also
\cite{FF}, for evaluation of the lower bound of stable
dimensions). It turns out that the method of \cite{Fu2} applies
well to an arbitrary superspace $V$, not necessarily even. For
completeness sake, below we re-expose  Fuks' results  for general
superspaces in a form convenient for our subsequent discussion.

The first step of our computation consists in the reinterpretation
of the $GL(V)$-equivariant cohomology of $\mathbb{E}$ as the
cohomology of a certain graph complex. We consider the graphs
satisfying the following special properties:
\begin{itemize}
    \item[(a)] each edge is equipped with a direction;
    \item[(b)] each vertex has exactly one outgoing and at least one
    incoming edge;
    \item[(c)]  we admit outgoing and incoming \textit{legs},
    i.e., edges bounded by a vertex from one side and
    having a ``free end'' on the other;
    \item[(d)] the vertices, the incoming
    and outgoing legs are numbered (numbering each of these three
    sets
    we define a \textit{decoration} of a graph).
\end{itemize}
The graphs need not be connected and loops are allowed.

The  \textit{genus} of a graph $\Gamma$ is the first Betty number
of its geometric realization as a one-dimensional cell complex (with
one extra 0-cell added for each leg).  Accordingly, the zero Betty
number counts the number of connected components of $\Gamma$. Two
graphs are considered to be equivalent if they are isomorphic as
cell complexes and the corresponding isomorphism respects both the
decoration and orientation of edges.

It is easy to see that the restriction (b) imposed on the vertices
of our graphs implies that the genus of each \textit{connected}
graph is either 0 or 1. We refer to the graphs of these two groups
as \textit{tree} and \textit{cyclic}, respectively, see
Fig.\ref{graphs}.  Notice that each tree graph has the only
outgoing leg, while a cyclic graph has none.

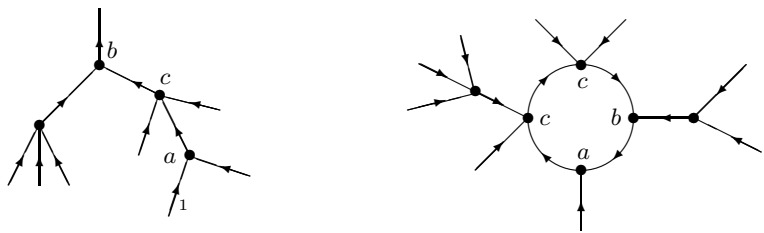
\begin{figure}[t]
\unitlength 1mm 
\begin{picture}(60,31)(15,-14)
\put(4,-8){\line(-1,2){4}}\put(-4.1,-8){\line(1,2){4}}\put(0,-8){\line(0,1){7.5}}
\put(4,-8){\vector(-1,2){2}}\put(-4.1,-8){\vector(1,2){2}}\put(0,-8){\vector(0,1){4}}
\put(0,0){\circle*{1.5}}

\put(8,8){\circle*{1.5}}
\put(8,8){\line(0,1){7.5}}\put(0,0){\line(1,1){7.5}}\put(16,4){\line(-2,1){7.5}}
\put(8,8){\vector(0,1){3.75}}\put(0,0){\vector(1,1){3.75}}\put(16,4){\vector(-2,1){3.75}}

\put(16,4){\circle*{1.5}}
\put(13.2,-4){\line(1,3){2.8}}\put(20,-4){\line(-1,2){3.75}}\put(24,2){\line(-4,1){7.5}}
\put(13.2,-4){\vector(1,3){1.4}}\put(20,-4){\vector(-1,2){2}}\put(24,2){\vector(-4,1){4}}

\put(20,-4){\circle*{1.5}}
\put(17.2,-12){\line(1,3){2.8}}\put(28,-6.8){\line(-3,1){8}}
\put(17.2,-12){\vector(1,3){1.4}}\put(28,-6.8){\vector(-3,1){4}}
\put(18.5,-11){\scriptsize$_1$}

\put(9,9){{\scriptsize$b$}} \put(16,5.5){{\scriptsize$c$}}
\put(16.5,-5){{\scriptsize$a$}}
\put(65,1){\circle*{1.5}} \put(67.25,6.15){\vector(1,1){.5}}
\put(58,-6){\line(1,1){7}}\put(58,4.5){\line(2,-1){7}}
\put(58,-6){\vector(1,1){3.75}}\put(58,4.5){\vector(2,-1){3.75}}
\put(58,4.5){\circle*{1.5}}
\put(50.5,8.25){\line(2,-1){7}}\put(49,2){\line(4,1){9}}\put(56,12){\line(1,-4){2}}
\put(50.5,8.25){\vector(2,-1){3.75}}\put(49,2){\vector(4,1){5}}\put(56,12){\vector(1,-4){1}}

\put(72,-6){\circle*{1.5}} \put(76.85,6,1){\vector(1,-1){.5}}
\put(72,-14){\line(0,1){7.5}} \put(72,-14){\vector(0,1){4}}

\put(72,8){\circle*{1.5}} \put(76.75,-4.15){\vector(-1,-1){.5}}
\put(66,14){\line(1,-1){6.25}}\put(78,14){\line(-1,-1){6.25}}
\put(66,14){\vector(1,-1){3.5}}\put(78,14){\vector(-1,-1){3.5}}

\put(79,1){\circle*{1.5}} \put(67.25,-4.15){\vector(-1,1){.5}}
\put(87,1){\line(-1,0){7.5}} \put(87,1){\vector(-1,0){4.5}}
\put(87,1){\circle*{1.5}}
\put(94,8){\line(-1,-1){7.5}}\put(96,-3.5){\line(-2,1){9.5}}
\put(94,8){\vector(-1,-1){4}}\put(96,-3.5){\vector(-2,1){4}}

\put(72,1){\circle{20}}

\put(66.5,0){{\scriptsize$c$}} \put(71.5,5){{\scriptsize$c$}}
\put(71.5,-4.5){{\scriptsize$a$}} \put(76,0){{\scriptsize$b$}}
\end{picture}
\caption{\protect\small{Some typical examples of tree and cyclic
graphs. The labels on the vertices are from the proof of Theorem
\ref{FT}}.}\label{graphs}
\end{figure}

The algebra $\mathcal{T}(\mathbb{E})^{\mathrm{inv}}$ admits a very
helpful visualization. Namely, to each generator
\begin{equation}\label{gen}
C_n=dy^{i_n}\otimes dy^{i_{n-1}}\otimes\cdots\otimes
dy^{i_1}c_{i_1\cdots i_n}^j\frac{\partial}{\partial y^j}
\end{equation}
of ${\mathcal{T}}(\mathbb{{E}})^{{\mathrm{inv}}}$ we associate a
one-vertex planar graph
\begin{equation}\label{corolla}
\begin{split}
\unitlength 1mm 
\begin{picture}(60,20)(-20,-3)
\put(-10,6){$\gamma_n$ =}
\put(0,0){\line(1,1){7.5}}\put(4,0){\line(1,2){3.75}}\put(16,0){\line(-1,1){7.5}}
\put(0,0){\vector(1,1){3.75}}\put(4,0){\vector(1,2){1.75}}\put(16,0){\vector(-1,1){3.75}}
\put(8,8){\circle*{1.5}} \put(8,8){\line(0,1){7.5}}
\put(8,8){\vector(0,1){3.75}} \put(7,0){{\scriptsize$\ldots$}}
\put(-1.5,-2){{\scriptsize$_1$}} \put(2.5,-2){{\scriptsize$_2$}}
\put(16.3,-2){{\scriptsize$_n$}}
\end{picture}
\end{split}
\end{equation}
called \textit{corolla}, whose incoming legs $1,2,...,n$ symbolize
the covariant indices $i_1,...,i_n$, while  the unique outgoing
leg corresponds to the contravariant index $j$. A linear basis in
$\mathcal{T}(\mathbb{E})^{\mathrm{inv}}$ is obtained from the
generators (\ref{gen}) by means of the tensor operations. These
have obvious graphical counterparts. The tensor product of $k$
generators
\begin{equation}\label{t-prod}
C=C_{n_1}\otimes C_{n_2}\otimes \cdots \otimes C_{n_k}
\end{equation}
is represented by the \textit{ordered} disjoint union
\begin{equation*}
\Gamma=\gamma_{n_1}\sqcup \gamma_{n_2}\sqcup\cdots
\sqcup\gamma_{n_k}
\end{equation*}
 of the corresponding corollas.  The
lexicographical  ordering in writing the disjoint union assigns
the unique vertex of $\gamma_{n_i}\subset \Gamma$ with the number
$i$, and the numberings of the incoming and outgoing legs of
$\gamma_{n_i}$ are shifted by $n_1+\cdots + n_{i-1}$ and $i-1$,
respectively. In such a way  the graph $\Gamma$ gets a natural
decoration.

To describe a contraction of the tensor (\ref{t-prod}), say, one
corresponding to a covariant index $i$ and a contravariant index
$j$, we are just gluing the free end of the outgoing leg $j$ of
$\Gamma$ with the free end of the incoming leg $i$ to produce an
oriented edge joining a pair of vertices; in so doing, the
uncontracted legs are renumbered by consecutive numbers according
to their order in $\Gamma$.

Finally, a permutation of indices in (\ref{t-prod}) results in
exchange of numbers among the corresponding legs.

It is quite clear that any of the graphs can be produced from the
corollas (\ref{corolla}) with the help of elementary operations
above and to any such graph $\Gamma$ one can assign a unique
tensor  $C\in \mathcal{T}(\mathbb{E})^{\mathrm{inv}}$. The
correspondence $\Gamma \mapsto C$ defines a map, denoted by $R$,
from the set of all graphs obeying  ($a$)-($d$) to the set of
$GL(V)$-invariant $\mathbb{E}$-tensors.

\begin{ex}\label{ex} By way of illustration, let us apply the correspondence map $R$ to the graph $\Gamma$ of the form
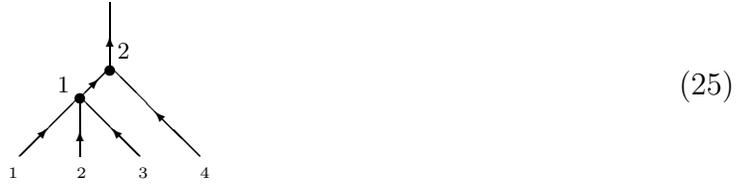
\begin{equation}\label{ex1}
\begin{split}
\unitlength 1mm 
\begin{picture}(60,25)(-20,-2)
\put(0,0.5){\line(1,1){7.5}}\put(8,0.5){\line(0,1){7.5}}\put(16,0.5){\line(-1,1){7.5}}
\put(0,0.5){\vector(1,1){3.75}}\put(8,0.5){\vector(0,1){3.5}}\put(16,0.5){\vector(-1,1){3.75}}
\put(8,8.2){\circle*{1.5}} \put(8.5,8.8){\line(1,1){3.75}}
\put(8.5,8.8){\vector(1,1){2}} \put(24,0.5){\line(-1,1){12}}
\put(24,0.5){\vector(-1,1){6}} \put(12,12){\circle*{1.5}}
\put(12,12){\line(0,1){9}} \put(12,12){\vector(0,1){4.5}}
\put(5,9){{\scriptsize$1$}} \put(13,13.5){{\scriptsize$2$}}
\put(-1.5,-2){{\scriptsize$_1$}} \put(7.6,-2){{\scriptsize$_2$}}
\put(15.8,-2){{\scriptsize$_3$}} \put(24,-2){{\scriptsize$_4$}}
\end{picture}
\end{split}
\end{equation}
To write down the  tensor $R(\Gamma)\in
\mathcal{T}(\mathbb{E})^{\mathrm{inv}}$  we first cut the graph
into the ordered disjoint union of two corollas $\gamma_3\sqcup
\gamma_2$; the order is defined by the numeration of the vertices.
To each corolla $\gamma_n$, $n=2,3$, we associate a generator $C_n
\in \mathcal{T}^{n,1}(\mathbb{E})^{\mathrm{inv}}$ such that the
map $R$ takes $\gamma_3\sqcup\gamma_2$ to the tensor product
$C_3\otimes C_2$. Then we glue the cut edge down to produce the
original graph $\Gamma$, graphically

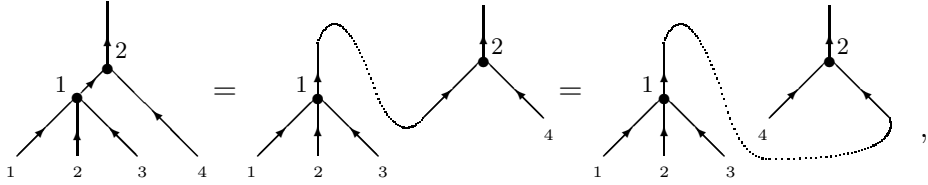
\unitlength 1mm 
\begin{picture}(60,28)(-0,-5)
\put(0,0.5){\line(1,1){7.5}}\put(8,0.5){\line(0,1){7.5}}\put(16,0.5){\line(-1,1){7.5}}
\put(0,0.5){\vector(1,1){3.75}}\put(8,0.5){\vector(0,1){3.5}}\put(16,0.5){\vector(-1,1){3.75}}

\put(8,8.2){\circle*{1.5}}
\put(8.5,8.8){\line(1,1){3.75}}\put(24,0.5){\line(-1,1){12}}
\put(8.5,8.8){\vector(1,1){2}}\put(24,0.5){\vector(-1,1){6}}

\put(12,12){\circle*{1.5}} \put(12,12){\line(0,1){9}}
\put(12,12){\vector(0,1){4.5}}

\put(5,9){{\scriptsize$1$}} \put(13,13.5){{\scriptsize$2$}}

\put(-1.5,-2){{\scriptsize$_1$}}\put(7.4,-2){{\scriptsize$_2$}}\put(16,-2){{\scriptsize$_3$}}
\put(24,-2){{\scriptsize$_4$}}

\put(26,8){{$=$}}
\put(32,0.5){\line(1,1){7.5}}\put(40,0.5){\line(0,1){8}}\put(48,0.5){\line(-1,1){7.5}}
\put(32,0.5){\vector(1,1){3.75}}\put(40,0.5){\vector(0,1){4}}\put(48,0.5){\vector(-1,1){3.75}}
\put(40,8){\circle*{1.5}}\put(37,9){\scriptsize$1$}
\put(40,8){\line(0,1){7.5}} \put(40,8){\vector(0,1){3.75}}

\put(30.5,-2){{\scriptsize$_1$}}
\put(39.4,-2){{\scriptsize$_2$}}\put(48,-2){{\scriptsize$_3$}}
\qbezier[40](39.9,15.5)(43,23)(48,8)
\qbezier[25](48,8)(51,2)(54,5.5)

\put(54,5.5){\line(1,1){7.5}}\put(70,5.5){\line(-1,1){7.5}}
\put(54,5.5){\vector(1,1){3.75}}\put(70,5.5){\vector(-1,1){3.75}}
\put(62,13){\circle*{1.5}}\put(63,14){\scriptsize$2$}
\put(62,13){\line(0,1){7.5}} \put(62,13){\vector(0,1){3.75}}

\put(70,3){{\scriptsize$_4$}}

\put(72,8){{$=$}}
\put(78,0.5){\line(1,1){7.5}}\put(86,0.5){\line(0,1){8}}\put(94,0.5){\line(-1,1){7.5}}
\put(78,0.5){\vector(1,1){3.75}}\put(86,0.5){\vector(0,1){4}}\put(94,0.5){\vector(-1,1){3.75}}
\put(86,8){\circle*{1.5}}\put(83,9){\scriptsize$1$}
\put(86,8){\line(0,1){7.5}} \put(86,8){\vector(0,1){3.75}}

\put(76.5,-2){{\scriptsize$_1$}}
\put(85.4,-2){{\scriptsize$_2$}}\put(94.2,-2){{\scriptsize$_3$}}
\qbezier[40](85.9,15.5)(89,23)(94,8)
\qbezier[17](94,8)(96,0)(100,0)
\qbezier[40](100,0)(117.5,0)(116,5.5)

\put(100,5.5){\line(1,1){7.5}}\put(116,5.5){\line(-1,1){7.5}}
\put(100,5.5){\vector(1,1){3.75}}\put(116,5.5){\vector(-1,1){3.75}}
\put(108,13){\circle*{1.5}}
\put(109,14){\scriptsize$2$}\put(108,13){\line(0,1){7.5}}
\put(108,13){\vector(0,1){3.75}}

\put(98.5,3){{\scriptsize$_4$}}\put(120,3){,}
\end{picture}

\noindent contracting simultaneously the corresponding indices of
$C_3\otimes C_2$. The tensors $\{C_n\}$ being symmetric in lower
indices, it does not matter which incoming leg of $\gamma_2$ is
glued to the outgoing leg of $\gamma_3$.  Finally, we rearrange
the indices according to the order of legs. The result is given by

\begin{equation*}\label{}
\begin{array}{c}
A=(dy^k\otimes dy^j\otimes dy^i c_{ijk}^n
\underset{\rule{0.15mm}{1.5mm}\rule
{13mm}{0.15mm}\rule{0.15mm}{1.5mm}\rule{2mm}{0mm}}{\partial_n)\otimes
(dy^l}\otimes dy^s c_{sl}^m\partial_m
)\\[4mm]
=dy^s\otimes dy^k\otimes dy^j\otimes dy^i a_{ijks}^m
\partial_m\,,
\end{array}
\end{equation*}
where
\begin{equation*}\label{}
    a^m_{ijks}=(-1)^{\epsilon_s}c_{ijk}^nc_{ns}^m\,.
\end{equation*}
Our convention for the pairing of vectors and covectors is
\begin{equation*}\label{}
   \underset{\rule{0.15mm}{1.5mm}\rule
{10mm}{0.15mm}\rule{0.15mm}{1.5mm}\rule{2mm}{0mm}}
{\partial_i\otimes dy^j}:=\langle\partial_i , dy^j
    \rangle=\delta_i^j\,,
\end{equation*}
 and we follow the usual sign rule: if two object of parities
$\epsilon_1$ and $\epsilon_2$ are interchanged then the factor
$(-1)^{\epsilon_1\epsilon_2}$ is inserted.

Observe that applying the correspondence map  $R$ to the graph
(\ref{ex1}) with altered order of vertices (but the same order of
legs) yields
\begin{equation*}\label{}
(dy^l\otimes \underset{\rule{0.15mm}{1.5mm}\rule
{57mm}{0.15mm}\rule{0.15mm}{1.5mm}\rule{2mm}{0mm}}{ dy^s
c_{sl}^m\partial_m)\otimes(dy^k\otimes dy^j\otimes dy^i c_{ijk}^n
\partial_n})=-A\,.
\end{equation*}

Clearly, using the cut-and-paste algorithm above, one can
unambiguously reconstruct a tensor field by a graph with an
arbitrary set of vertices.
\end{ex}

Let us now introduce  the real vector space
$\widetilde{\mathcal{G}}$ spanned by all the graphs satisfying the
conditions ($a$)-($d$) above. The tensor operations and the
correspondence map $R$ extend to the space
$\widetilde{{\mathcal{G}}}$ by linearity. The space $\widetilde{\mathcal{G}}$,
however, is not what we actually need, since the homomorphism $R:
\widetilde{\mathcal{G}}\rightarrow
\mathcal{T}(\mathbb{E})^{\mathrm{inv}}$ is far from being a
bijection. For example, if $\Gamma'$ is  the graph obtained from a
graph  $\Gamma$ by transposition  of two labels on vertices, then
$$
R(\Gamma')= - R(\Gamma)\,.
$$
The minus sign is due to the fact that the elementary covariants
$\{C_n\}$ depicted by the corollas $\{\gamma_n\}$ are Grassman odd
(c.f. Example \ref{ex}). This motivates us to introduce the
quotient space $\mathcal{G}=\widetilde{\mathcal{{G}}}/\sim$ with
respect to the equivalence relation $\Gamma'\sim -\Gamma$. The
tensor operations or, more precisely, their graphical
representation pass through the quotient making $\mathcal{G}$ into
an \textit{abstract tensor algebra}\footnote{Unfortunately, there
is no commonly accepted name for this object.  A more expressive
term would be desirable. As we will see in a moment $\mathcal{G}$
is not just an algebra but a cochain complex, whose coboundary
operator respects the tensor operations. So, an appropriate name
might be the \textit{differential algebra of graphs} (DAG). }
freely generated by the corollas (\ref{corolla}). Concerning the
general concept of an ``abstract tensor calculus'' we refer the
reader to the recent papers \cite{DM}, \cite{Markl}. As a linear
space, $\mathcal{G}$ splits into the direct sum of finite-dimensional subspaces:
\begin{equation*}\label{}
    \mathcal{G}= \bigoplus_{n,m,k}\mathcal{G}^{n,m}_k\,,
\end{equation*}
where the superscripts  $n$ and $m$ refer to the number of
incoming and outgoing legs of a graph, while the subscript $k$
indicates the number of the vertices. In addition to this
trigrading the algebra $\mathcal{G}$ possesses an increasing
filtration
\begin{equation*}\label{}
    0\subset F_1\mathcal{G}\subset F_2\mathcal{G}\subset \cdots
    \subset F_\infty \mathcal{G}=\mathcal{G}
\end{equation*}
mimicking the filtration (\ref{Tfilt}) of
$\mathcal{T}(\mathbb{E})^{\mathrm{inv}}$. By definition, the
subalgebra $F_n\mathcal{G}\subset \mathcal{G}$ is generated by the
corollas (\ref{corolla}) with at most $n$ incoming legs.

The map $R$ induces the map from $\mathcal{G}$ to
$\mathcal{T}(\mathbb{E})^{\mathrm{inv}}$, which we will denote by
the same letter $R$. Passing to the quotient
$\mathcal{G}=\widetilde{\mathcal{G}}/\sim$ does not kill the
kernel of $R$ completely. For example, if $\dim V=1|0$, then

\unitlength 1mm 
\begin{picture}(60,10)(-28,-3)
\put(11,0){$R \Big($}  \put(17,-1){\scriptsize$_1$}
\put(41,-1){\scriptsize$_2$} \put(17,1){\line(1,0){7.5}}
\put(17,1){\vector(1,0){4.5}} \put(25,1){\circle*{1.5}}
\put(23.5,3){{\scriptsize$1$}} \qbezier(25,1)(30,6)(35,1)
\put(30,3.5){\vector(1,0){0.5}} \qbezier(25,1)(30,-4)(35,1)
\put(30,-1.5){\vector(-1,0){0.5}} \put(35,1){\circle*{1.5}}
\put(35,3){{\scriptsize$2$}} \put(35,1){\line(1,0){7.5}}
\put(42,1){\vector(-1,0){4}} \put(43,0){$\Big )= 0\,.$}
\end{picture}

\noindent The last equality, however,  could not take place if the
space $V$ were big enough. More precisely, the second main theorem
of invariant theory \cite{Weyl}, \cite{Fu1} ensures  that the map
\begin{equation}\label{R}
    R: F_r\mathcal{G}^{n,m}\rightarrow
    F_r\mathcal{T}^{n,m}(\mathbb{E})^{\mathrm{inv}}
\end{equation}
is an isomorphism of vector spaces provided that $|\dim V|\gg
r,n$.

Thus, in the stable range of dimensions the homomorphism $R$
allows one to replace the tensor algebra
$\mathcal{T}(\mathbb{E})^{\mathrm{inv}}$ by the graph algebra
$\mathcal{G}$. Then, the pullback of the differential $\delta$ in
$\mathcal{T}(\mathbb{E})^{\mathrm{inv}}$ via the isomorphism
(\ref{R}) gives $\mathcal{G}$ the structure of cochain complex
with respect to the coboundary operator
\begin{equation*}\label{}
\partial=R^{-1}
\delta R: \mathcal{G}^{n,m}_k\rightarrow
\mathcal{G}^{n,m}_{k+1}\,,
\end{equation*} in so doing, the stable cohomology groups of
$\mathcal{T}(\mathbb{E})^{\mathrm{inv}}$ appear to be isomorphic
to the graph cohomology groups $H(\mathcal{G})=\mathrm{Ker}
\partial/ \mathrm{Im} \partial$. Since $\partial$ differentiates the tensor product and
commutes with the contraction it is sufficient to describe its
action on corollas. Translating formulas (\ref{Qmorf}) and
(\ref{Qmorf1}) into the graph language, we get

\begin{equation}\label{diff1}
\begin{split}
\unitlength 1mm 
\begin{picture}(60,10)(-6,-3)
\put(-9,0){$\partial\left(\rule{0mm}{3mm}\right.$}
\put(10,0){$\left.\rule{0mm}{3mm}\right) =$}
\put(3,1){\line(1,0){7.5}} \put(3,1){\vector(1,0){5}}
\put(3,1){\circle*{1.5}} \put(-4,1){\line(1,0){7.5}}
\put(-4,1){\vector(1,0){4}} \put(2,3){{\scriptsize$m$}}
\put(19,1){\line(1,0){7.5}} \put(19,1){\vector(1,0){4.5}}
\put(27,1){\circle*{1.5}} \put(27,1){\line(1,0){7.5}}
\put(27,1){\vector(1,0){5}} \put(25.5,3){{\scriptsize$m$}}
\put(35,1){\line(1,0){7.5}} \put(35,1){\vector(1,0){5}}
\put(35,1){\circle*{1.5}} \put(33,3){{\scriptsize$m\!+\!1$}}
\end{picture}
\end{split}
\end{equation}

\noindent
and

\begin{equation}\label{diff2}
\begin{split}
\unitlength 1mm 
\begin{picture}(60,30)(5,-10)
\put(-9,6){$\partial\left(\rule{0mm}{10mm}\right.$}
\put(17,6){$\left.\rule{0mm}{10mm}\right) = \qquad $
{\large$\sum$}} \put(28,1.75){\scriptsize$J'\sqcup J''=J$}
\put(25,-2.5){\scriptsize$|J'|>0, |J''|>1$}
\put(0,0){\line(1,1){7.5}}\put(4,0){\line(1,2){3.75}}\put(16,0){\line(-1,1){7.5}}
\put(0,0){\vector(1,1){3.75}}\put(4,0){\vector(1,2){1.75}}\put(16,0){\vector(-1,1){3.75}}
\put(8,8){\circle*{1.5}} \put(8,8){\line(0,1){7.5}}
\put(8,8){\vector(0,1){3.75}} \put(7,0){{\scriptsize$\ldots$}}
\put(9,8){{\scriptsize$m$}}
\put(-1,-0.5){$\underbrace{\rule{18mm}{0mm}}_{J}$}
\put(47,5){\line(1,1){7.5}}\put(51,5){\line(1,2){3.75}}\put(63,5){\line(-1,1){7.5}}
\put(47,5){\vector(1,1){3.75}}\put(51,5){\vector(1,2){1.75}}\put(63,5){\vector(-1,1){3.75}}
\put(55,13){\circle*{1.5}} \put(55,13){\line(0,1){7.5}}
\put(55,13){\vector(0,1){3.75}} \put(54,5){{\scriptsize$\ldots$}}
\put(56,13){{\scriptsize$m$}}
\put(46,4.5){$\underbrace{\rule{12mm}{0mm}}_{J'}$}
\put(55,-3){\line(1,1){7.5}}\put(59,-3){\line(1,2){3.75}}\put(71,-3){\line(-1,1){7.5}}
\put(55,-3){\vector(1,1){3.75}}\put(59,-3){\vector(1,2){1.75}}\put(71,-3){\vector(-1,1){3.75}}
\put(63,5){\circle*{1.5}} \put(62,-3){{\scriptsize$\ldots$}}
\put(64,5){{\scriptsize$m+1$}}
\put(54,-3.5){$\underbrace{\rule{18mm}{0mm}}_{J''}$}
\end{picture}
\end{split}
\end{equation}

\noindent for $|J|>1$. A decoration on $\Gamma\in \mathcal{G}$
induces a decoration on each summand of $\partial\Gamma$  as
follows: splitting of the $m$th vertex produces a pair of new
vertices that are labelled by $m$ and $m+1$, the labels less than
$m$ remain intact while the labels greater than $m$ increase by
1.

A remarkable property of the differential $\partial$ is that it
neither  permutes  the connected components of a decorated graph
nor changes their number.  More precisely, the complex
$\mathcal{G}^{n,m}$ is decomposed into a direct sum of its
subcomplexes $\mathcal{G}^{n,m}_{A_1,...,A_k;B_1,...,B_k}$, where
$\{A_1,...,A_k\}$ and $\{B_1,...,B_k\}$ are partitions of the sets
$\{1,...,n\}$ and $\{1,...,m\}$, and
$\mathcal{G}^{n,m}_{A_1,...,A_k;B_1,...,B_k}$ is generated by
graphs with $k$ connected components, the $l$th of which contains
incoming and outgoing legs labelled by the elements of $A_l$ and
$B_l$, respectively. Notice that the sets $\{A_1,...,A_k\}$ and
$\{B_1,...,B_k\}$ are defined up to simultaneous permutations of
$A_i$ with $A_j$ and $B_i$ with $B_j$, and some of the sets
$A_1,...,A_k$, $B_1,...,B_k$ may be empty. Let
$\bar{\mathcal{G}}=\bigoplus\bar{\mathcal{G}}^{n,m}$ denote the
subcomplex of \textit{connected} graphs. Then it is clear that
\begin{equation*}\label{}
\mathcal{G}^{n,m}_{A_1,...,A_k;B_1,...,B_k}\simeq
\bigotimes_{l=1}^k \bar{\mathcal{G}}^{|A_l|,|B_l|}\,,
\end{equation*}
and by the K\"unneth formula the computation of the graph
cohomology boils down  to the computation of the group
$H(\bar{\mathcal{G}})$.

The complex $\bar{\mathcal{G}}$ in its turn splits into the direct
sum of three subcomplexes:
\begin{equation*}\label{dsum}
\bar{\mathcal{G}}=\bar{\mathcal{G}}'\oplus
\bar{\mathcal{G}}''\oplus \bar{\mathcal{G}}'''\,.
\end{equation*}
Here  $\bar{\mathcal{G}}'$ is spanned by the bivalent graphs build
from $\gamma_1$. The subalgebra $\bar{\mathcal{G}}''$ is generated
by the corollas $\gamma_n$ with $n>1$, that is the valency of each
vertex is at least three.  Finally,  $\bar{\mathcal{G}}'''$ spans
the graphs with at least one bivalent vertex and at least one
vertex of valency greater than two.  The invariance  of each of
these three subspaces under the action of the coboundary operator
is obvious.

The following statement can be viewed as a special case of Losik's
lemma \cite{Losik} (see also \cite[Theorem 2.2.8]{Fu1}).

\begin{prop}
$H(\bar{\mathcal{G}}''')=0$,   hence
$H(\bar{\mathcal{G}})=H(\bar{\mathcal{G}}')\oplus
H(\bar{\mathcal{G}}'')$.
\end{prop}

\begin{proof}
Observe that for each graph $\Gamma\in \bar{\mathcal{G}}'''$ there
exists a unique, up to decoration, graph  $ \widetilde{\Gamma}\in
\bar{\mathcal{G}}''$ such that $\Gamma$ is obtained from
$\widetilde{\Gamma}$ by putting bivalent vertices on the edges of
the latter. The edges of $\widetilde{\Gamma}$, equipped with
bivalent vertices, will be called \textit{branches}. The
\textit{length} of a branch is just the number of bivalent
vertices inserted. Denote by $\gamma_n$ the branch of length $n$
with the special decoration
\begin{equation*}\label{}
\begin{split}
\unitlength 1mm 
\begin{picture}(0,10)(10,-3)
\put(-13,0){$\gamma_n=\quad$} \put(-4,1){\line(1,0){17.5}}
\put(-4,1){\vector(1,0){12.5}}\put(-4,1){\vector(1,0){4}}
\put(11,1){\circle*{1.5}} \put(3,1){\circle*{1.5}}
\put(14.5,0.75){$\ldots$} \put(22.5,1){\circle*{1.5}}
\put(2.5,3.5){{\scriptsize$_1$}}
\put(10.5,3.5){{\scriptsize$_{2}$}} \put(20.5,1){\line(1,0){8.5}}
\put(21,1){\vector(1,0){6}} \put(22,3.5){{\scriptsize$_{n}$}}
\end{picture}
\end{split}
\end{equation*}
and consider a decreasing filtration
\begin{equation}\label{filtration}
    \bar{\mathcal{G}}'''= F_1\bar{\mathcal{G}}''' \supset F_2\bar{\mathcal{G}}''' \supset\cdots\supset
    F_\infty\bar{\mathcal{G}}'''=0\,.
\end{equation}
By definition, the graph $\Gamma$ belongs to
$F_n\bar{\mathcal{G}}'''$, if the underlying graph
${\widetilde{\Gamma}}\in \bar{\mathcal{G}}''$ contains at least
$n$ vertices. Clearly, the differential $\partial$ preserves the
filtration. Associated to the filtration (\ref{filtration}) is the
first quadrant spectral sequence $\{E_r,d_r\}$ with
$E_0^{p,q}=F_p\bar{\mathcal{G}}'''_{p+q}$. Geometrically, the
zeroth differential $d_0: E_0\rightarrow E_0$ acts only on the
bivalent vertices of $\Gamma$ according to the general rule
(\ref{diff1}), lengthening the branches of odd length by $1$ and
annihilating the branches of even length. For instance
\begin{equation}\label{dg}
    d_0\gamma_n=\left\{%
\begin{array}{ll}
    \gamma_{n+1}, & \hbox{for $n$ odd;} \\
    0, & \hbox{otherwise.} \\
\end{array}%
\right.
\end{equation}
We claim that $E_1=\mathrm{Ker}d_0/\mathrm{Im}d_0=0$. Indeed,
consider the operator $h: E_0\rightarrow E_0$, successively acting
on $gamma_n$ by the rule
\begin{equation}\label{hg}
    h\gamma_n=\left\{%
\begin{array}{ll}
    \gamma_{n+1}, & \hbox{for $n$ even;} \\
    0, & \hbox{otherwise.} \\
\end{array}%
\right.
\end{equation}
In so doing, the labels on the other vertices that do not belong
to $\gamma_n\subset \Gamma$ increase by 1. The action of $h$ and
$d_0$ on an arbitrary  decorated branch is easily reconstructed
from (\ref{hg}) and (\ref{dg}), if one notes that the change of
decoration can always be  compensated by an overall sign factor.

 Writing $\Delta=h d_0+d_0h$, we see that
$\Delta(\Gamma)=n\Gamma$, where $n>0$ is the number of branches of
nonzero length. The operator $\Delta$ being invertible, we can
take $\Delta^{-1}h: E_0\rightarrow E_0$ to be a contracting
homotopy. Thus the complex $(E_0,d_0)$ is acyclic and so is
$\bar{\mathcal{G}}'''$.
\end{proof}

The complex of bivalent graphs splits as
\begin{equation*}\label{}
    \bar{\mathcal{G}}'=\bar{\mathcal{G}}^{0,0}\oplus
    \bar{\mathcal{G}}^{1,1}\,,
\end{equation*}
where the subcomplexes  $\bar{\mathcal{G}}^{1,1}$ and
$\bar{\mathcal{G}}^{0,0}$ span, respectively,  the straight line
and polygon graphs.

\begin{prop} The complex $\bar{\mathcal{G}}^{1,1}$ is acyclic
and
$$
H_n(\bar{\mathcal{G}}^{0,0})=\left\{%
\begin{array}{ll}
    \mathbb{R}, & \hbox{if $n=2m-1$;} \\[1mm]
    0, & \hbox{if $n=2m$.} \\
\end{array}%
\right.
$$
The corresponding nontrivial cocycles are given by the polygons

\begin{equation}\label{A-graph}
\begin{split}
\unitlength 1mm 
\begin{picture}(100,25)(-45,-12)
    \put(-2,2){\circle*{1.5}}
\put(-5,1){{\scriptsize$4$}} \put(-2,2){\line(2,3){5.5}}
\put(-2,2){\vector(2,3){3}}
    \put(3.75,10.5){\circle*{1.5}}
\put(1,10.5){{\scriptsize$5$}} \put(3.5,10.5){\line(1,0){3.5}}
\put(3.5,10.5){\vector(1,0){3.5}}
\put(7.5,10.5){{\scriptsize$\ldots$}} \put(12,10.5){\line(1,0){3}}
    \put(14.75,10.5){\circle*{1.5}}
\put(16.5,10){{\scriptsize$2m$}}
\put(14.75,10.5){\line(2,-3){5.5}}
\put(14.75,10.5){\vector(2,-3){3}}
    \put(20.25,2){\circle*{1.5}}
\put(22,1){{\scriptsize$2m\!-\!1$}}
\put(20.25,2){\line(-1,-4){2.2}}
\put(20.25,2){\vector(-1,-4){1.3}}
    \put(18,-7){\circle*{1.5}}
\put(19.5,-9){{\scriptsize$1$}} \put(18,-7){\line(-2,-1){9}}
\put(18,-7){\vector(-2,-1){5}}
    \put(9,-11.5){\circle*{1.5}}
\put(10.5,-14){{\scriptsize$2$}} \put(9,-11.5){\line(-2,1){9}}
\put(9,-11.5){\vector(-2,1){5}}
    \put(0,-7){\circle*{1.5}}
\put(-3,-9){{\scriptsize$3$}} \put(0,-7){\line(-1,4){2.2}}
\put(0,-7){\vector(-1,4){1.3}}
\end{picture}
\end{split}
\end{equation}
with an odd number of vertices.
\end{prop}

\begin{proof}
The proof is straightforward.
\end{proof}

We now turn to computation of the group $H(\bar{\mathcal{G}}'')$.
The complex $\bar{\mathcal{G}}''$ breaks up into the direct sum of
various subcomplexes
\begin{equation*}\label{}
    \bar{\mathcal{G}}''=\left({\bigoplus}_{n=2}^\infty
    (\bar{\mathcal{G}}'')^{n,1}\right)\oplus \left({\bigoplus}_{n=1}^\infty
    (\bar{\mathcal{G}}'')^{n,0}\right)\,.
\end{equation*}
Here the spaces $(\bar{\mathcal{G}}'')^{n,1}$ and
$(\bar{\mathcal{G}}'')^{n,0}$ are spanned, respectively, by the
tree and cyclic graphs with $n$ incoming legs and no bivalent
vertices (see Fig. \ref{graphs}). To ease the notation, we will
write
\begin{equation}\label{TC}
\mathcal{T}^n_m=(\bar{\mathcal{G}}'')^{n+1,1}_m\qquad
\mbox{and}\qquad
\mathcal{C}^n_m=(\bar{\mathcal{G}}'')^{n,0}_m\,,
\end{equation}
where the lower index $m$ points to the number of vertices of a
graph, while the upper index $n$ counts the difference between its
incoming and outgoing legs. Since the valency of  each vertex is
assumed to be greater than 2, it is easy to see that the following
inequality holds for any graph from $\mathcal{T}^n_m$ or
$\mathcal{C}^n_m$:
\begin{equation}\label{ineq}
    n =\sum_{v\in V}(n_{v}-2)\geq m\,,
\end{equation}
 $n_v$ being the valency of a vertex $v\in V$. So, for $m>n$ the spaces (\ref{TC})
are assumed to be zero, and for any $n\in \mathbb{N}$ the graph
complexes $\mathcal{T}^n=\{\mathcal{T}^n_m\}$ and
$\mathcal{C}^n=\{\mathcal{C}_m^n\}$ are given by finite sequences
of finite-dimensional vector spaces.

\begin{thm}[D.B. Fuks \cite{Fu2}]\label{FT}

$$\dim H_m(\mathcal{T}^n) =\dim H_m(\mathcal{C}^n)=\left\{%
\begin{array}{ll}
    (n-1)!\,, &if\, \hbox{n=m\,;} \\[1mm]
    0, & \hbox{otherwise\,.} \\
\end{array}%
\right.    $$
\end{thm}

\begin{proof} Here we reproduce the original Fuks' proof.
Another proof can be found in
\cite[Sec.4]{Merkulov}.

Consider first the case of cyclic graphs with $n$ legs. A typical
cyclic graph is depicted in Fig. \ref{graphs}.  It is given by a
collection of legs and trees  attached to the vertices of the
cycle. Following Fuks, we introduce the complex dual to the
complex $\mathcal{C}^n$. Fixing basis in ${\mathcal{C}^n}$ allows
one to identify ${\mathcal{C}^n}$ with its dual space
$(\mathcal{C}^n)^\ast$. Then the differential $\partial^\ast:
{\mathcal{C}}^n_k \rightarrow {\mathcal{C}}^n_{k-1}$ in the dual
complex has an extremely simple description. The graph
$\partial^\ast\Gamma$ is given by a signed sum of decorated graphs
which can be obtained from $\Gamma$ by collapsing a single edge;
in so doing, the loops remain intact. To specify the signs and
decorations we can assume that the initial vertex of a collapsed
edge $e$ is labelled by 1 and the terminal vertex is labelled by
2. (The general case reduces to that by altering the order of
vertices with appropriate sign factors.) Then, the vertex which
results from collapsing $e$ is numbered by 1 and the labels of all
other vertices are reduced by 1. The labels of the legs remain the
same.

Since  the groups $\{{\mathcal{C}}^n_m\}$ are just
finite-dimensional vector spaces,
$$H_m({\mathcal{C}}^{n},
\partial) \simeq H_m({\mathcal{C}}^{n},\partial^\ast )\,,$$
and we may  focus upon  computation of the
$\partial^\ast$-co\-ho\-mo\-logy.

We filter $(\mathcal{C}^n,\partial^\ast)$ by the length of the
cycle, i.e., let $F_p\mathcal{C}^{n}$ be the subcomplex of
$\mathcal{C}^{n}$ spanned by graphs with at most $p$ cyclic
vertices (or edges). In view of (\ref{ineq}) the filtration is
finite:
$$
0 =F_0\mathcal{C}^{n}\subset F_1\mathcal{C}^n\subset\cdots \subset
F_{n} \mathcal{C}^{n}=\mathcal{C}^{n}\,.
$$
Let $\{E^r,d^r\}$ be the spectral sequence associated to the
filtration. Notice that $E_{p,q}^0=F_p\mathcal{C}_{p+q}^n=0$ for
$q<0$ and we have a first quadrant spectral sequence. The
coboundary operator $\partial^\ast: \mathcal{C}^n_{k}\rightarrow
\mathcal{C}^n_{k-1}$ is given by the sum
$\partial^\ast=\partial^\ast_c+\partial^\ast_{nc}$, where the
operator $\partial^\ast_c$ collapses only cyclic edges, and the
operator $\partial^{\ast}_{nc}$ collapses only non-cyclic edges.
By definition $d^0=\partial^\ast_{nc}$.

To compute $E^1=H(E^0,d^0)$, we arrange the cyclic vertices into
three groups labelled by the letters $a$, $b$, and $c$. The
vertices of type $a$ are trivalent vertices incident to one leg
and two edges, the vertices of type $b$ are  trivalent vertices
incident to three edges, and the other cyclic vertices (of valency
greater than 3) belong to the type $c$ (see Fig. \ref{graphs}).
With this partition we define a homotopy
$h:\mathcal{C}^n_k\rightarrow \mathcal{C}^n_{k+1}$. The operator
$h$ acts as a differentiation on the $c$-type vertices.
Graphically,
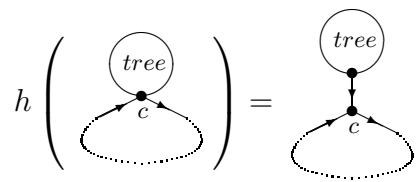
\begin{equation*}
\begin{split}
\unitlength 1mm 
\begin{picture}(60,23)(-10,-3)
\put(-9,6){$h\left(\rule{0mm}{10mm}\right.$}
\put(17,6){$\left.\rule{0mm}{10mm}\right) =$}
\put(2,5){\line(2,1){5.8}}\put(14,5){\line(-2,1){5.8}}
\put(2,5){\vector(2,1){4}}\put(11,6.5){\vector(2,-1){0.5}}
\put(8,8){\circle*{1.5}} \qbezier[25](2,5)(-4,0)(8,-1)
\qbezier[25](8,-1)(20,0)(14,5) \put(8,12.25){\circle{8}}
\put(5.5,11.5){{\scriptsize$tree$}} \put(7.5,5){{\scriptsize$c$}}
\put(30,3){\line(2,1){5.5}}\put(42,3){\line(-2,1){5.5}}\put(36,6){\line(0,1){5}}
\put(30,3){\vector(2,1){4}} \put(39,4.5){\vector(2,-1){0.5}}
\put(36,10.5){\vector(0,-1){3}}
\put(36,6){\circle*{1.5}}\put(36,11){\circle*{1.5}}
\qbezier[25](30,3)(24,-2)(36,-3) \qbezier[25](36,-3)(48,-2)(42,3)
\put(36,15.25){\circle{8}} \put(33.5,14.5){{\scriptsize$tree$}}
\put(35.5,3){{\scriptsize$c$}}
\end{picture}
\end{split}
\end{equation*}
A straightforward consideration shows that the homotopy $h$
connects $0$ with the endomorphism that multiplies each cyclic
graph by the total number of cyclic vertices of types $b$ and $c$.
Thus the subspace of graphs with at least one vertex of  $b$ or
$c$ type is acyclic.

On the other hand, the graphs with only $a$-type vertices have the
form
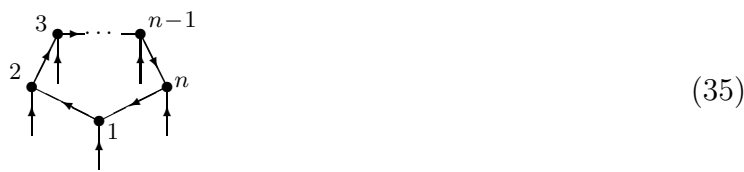
\begin{equation}\label{C-graph}
\begin{split}
\unitlength 1mm 
\begin{picture}(50,22)(-15,-11)
    \put(0,0){\circle*{1.5}}
\put(-3,1){{\scriptsize$2$}} \put(0,0){\line(1,2){3.75}}
\put(0,0){\vector(1,2){2.5}} \put(0,-6.5){\line(0,1){6.5}}
\put(0,-6.5){\vector(0,1){4}}
    \put(3.5,7){\circle*{1.5}}
\put(0.5,7.5){{\scriptsize$3$}} \put(3.5,7){\line(1,0){3.5}}
\put(3.5,7){\vector(1,0){3}} \put(7,7){{\scriptsize$\ldots$}}
\put(12,7){\line(1,0){3}} \put(3.5,0.5){\line(0,1){6.5}}
\put(3.5,0.5){\vector(0,1){4}}
    \put(14.5,7){\circle*{1.5}}
\put(15.5,8){{\scriptsize$n\!-\!1$}}
\put(14.5,7){\line(1,-2){3.75}} \put(14.5,7){\vector(1,-2){2.25}}
\put(14.5,0.5){\line(0,1){6.5}} \put(14.5,0.5){\vector(0,1){4}}
    \put(18,0){\circle*{1.5}}
\put(19,0){{\scriptsize$n$}} \put(18,0){\line(-2,-1){9}}
\put(18,0){\vector(-2,-1){5}} \put(18,-6.5){\line(0,1){6.5}}
\put(18,-6.5){\vector(0,1){4}}
    \put(9,-4.5){\circle*{1.5}}
\put(10,-7){{\scriptsize$1$}} \put(9,-4.5){\line(-2,1){9}}
\put(9,-4.5){\vector(-2,1){5}} \put(9,-11){\line(0,1){6.5}}
\put(9,-11){\vector(0,1){4}}
\end{picture}
\end{split}
\end{equation}
and differ from one another only by decoration. Being free of
non-cyclic edges, the graphs (\ref{C-graph}) are automatically
closed and represent nontrivial classes of $d^0$-cohomology group
$E^1_{n,0}$. The group $S_n$ acts on (\ref{C-graph}) by permuting
incoming legs. Since only the cyclic permutations lead to
isomorphic graphs, the number of different decorations or, what is
the same,  the dimension of the space $E^1_{n,0}$ is equal to
$(n-1)!$.

For reasons of dimension, the spectral sequence degenerates on the
first page and we get   $E^1_{p,0} \simeq H_p(\mathcal{C}^n)$.
This proves the theorem for the case of the cyclic graph complex
$\mathcal{C}^{n}$.

The case of tree graphs may be handled in much the same way, with
the only difference that the filtration of the complex
$\mathcal{T}^n$ is now defined by the length of the unique path
joining the first incoming leg of a connected tree graph to its
outgoing leg. This path plays the role of the cycle in the
previous consideration. The corresponding spectral sequence
degenerates from  $E^1$ yielding the only nontrivial group
$E^1_{n,0}$. The space $E^1_{n,0}$ is spanned by the trivalent
graphs of the form
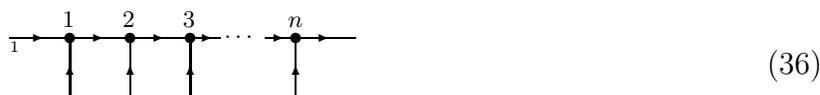
\begin{equation}\label{B-graph}
\begin{split}
\unitlength 1mm 
\begin{picture}(60,12)(-15,-9)
\put(-8,0){\vector(1,0){4.5}} \put(-8,0){\line(1,0){8}}
\put(-8,-1.5){{\scriptsize${}_1$}} \put(0,-8){\line(0,1){8}}
\put(0,-8){\vector(0,1){4.5}} \put(0,0){\circle*{1.5}}
\put(-1,1.5){{\scriptsize$1$}} \put(0,0){\line(1,0){8}}
\put(0,0){\vector(1,0){4.5}} \put(8,-8){\line(0,1){8}}
\put(8,-8){\vector(0,1){4.5}}
    \put(8,0){\circle*{1.5}}
\put(7,1.5){{\scriptsize$2$}} \put(8,0){\line(1,0){8}}
\put(8,0){\vector(1,0){4.5}} \put(16,-8){\line(0,1){8}}
\put(16,-8){\vector(0,1){4.5}}
    \put(16,0){\circle*{1.5}}
\put(15,1.5){{\scriptsize$3$}} \put(16,0){\line(1,0){4}}
\put(16,0){\vector(1,0){3}} \put(20.5,0){{\scriptsize$\ldots$}}
\put(26,0){\line(1,0){4}} \put(26,0){\vector(1,0){2.5}}
    \put(30,0){\circle*{1.5}}
\put(29,1.5){{\scriptsize$n$}} \put(30,-8){\line(0,1){8}}
\put(30,-8){\vector(0,1){4.5}} \put(30,0){\line(1,0){8}}
\put(30,0){\vector(1,0){4.5}}
\end{picture}
\end{split}
\end{equation}
By construction, the leftmost leg is labelled  by 1 and the labels
on the other $n-1$ incoming legs can be prescribed arbitrarily.
Hence $\dim E^1_{n,0}=(n-1)!$.

\end{proof}

\begin{rem}
We can also give an explicit description for basis cocycles whose
classes of $\partial$-cohomology generate the groups
$H_n(\mathcal{C}^n)$ and $H_n(\mathcal{T}^n)$. Relation
(\ref{ineq}) implies that the spaces $\mathcal{T}_n^n$ and
$\mathcal{C}_n^n$ are spanned by the trivalent graphs. Since
$\partial \gamma_2=0$, any trivalent graph is automatically
$\partial$-closed. To extract a basis of nontrivial cocycles
notice that the subspace of $\partial$-coboundaries is generated
by graphs $\partial \Gamma$, where $\Gamma$ has $n-2$ trivalent
vertices and one vertex of valency four. In accordance with
(\ref{diff2}), the action of the differential on the 4-valent
vertex reads
\begin{equation*}
\begin{split}
\unitlength 0.9mm
\begin{picture}(60,30)(18,-7)
\put(-11,6){$\frac12\partial\left(\rule{0mm}{10mm}\right.$}
\put(17,6){$\left.\rule{0mm}{10mm}\right) =$}
\put(0,0){\line(1,1){7.5}}\put(8,0){\line(0,1){7.5}}\put(16,0){\line(-1,1){7.5}}
\put(0,0){\vector(1,1){3.75}}\put(8,0){\vector(0,1){4}}\put(16,0){\vector(-1,1){3.75}}
\put(8,8){\circle*{1.5}} \put(8,8){\line(0,1){7.5}}
\put(8,8){\vector(0,1){3.75}}
\put(-1.5,-2){{\scriptsize$_1$}}\put(7.5,-2){{\scriptsize$_2$}}\put(16,-2){{\scriptsize$_3$}}
\put(30.5,4.75){\line(3,4){4.5}}\put(40,5){\line(-3,4){4.5}}
\put(30.5,4.75){\vector(3,4){3.5}}\put(40,5){\vector(-3,4){3.5}}
\put(35.5,11){\circle*{1.5}} \put(35.5,11){\line(0,1){7.5}}
\put(35.5,11){\vector(0,1){3.75}}
\put(35.25,-1){\line(3,4){4.5}}\put(44.75,-1){\line(-3,4){4.5}}
\put(35.25,-1){\vector(3,4){3.5}}\put(44.75,-1){\vector(-3,4){3.5}}
\put(40,5){\circle*{1.5}}
\put(28.5,3){{\scriptsize$_1$}}\put(33.5,-3){{\scriptsize$_2$}}\put(45,-3){{\scriptsize$_3$}}
\put(48,6){$+$}
\put(56.5,4.75){\line(3,4){4.5}}\put(66,5){\line(-3,4){4.5}}
\put(56.5,4.75){\vector(3,4){3.5}}\put(66,5){\vector(-3,4){3.5}}
\put(61.5,11){\circle*{1.5}} \put(61.5,11){\line(0,1){7.5}}
\put(61.5,11){\vector(0,1){3.75}}
\put(61.25,-1){\line(3,4){4.5}}\put(70.75,-1){\line(-3,4){4.5}}
\put(61.25,-1){\vector(3,4){3.5}}\put(70.75,-1){\vector(-3,4){3.5}}
\put(66,5){\circle*{1.5}}
\put(54.5,3){{\scriptsize$_3$}}\put(59.5,-3){{\scriptsize$_1$}}\put(71,-3){{\scriptsize$_2$}}
\put(74,6){$+$}
\put(82.5,4.75){\line(3,4){4.5}}\put(92,5){\line(-3,4){4.5}}
\put(82.5,4.75){\vector(3,4){3.5}}\put(92,5){\vector(-3,4){3.5}}
\put(87.5,11){\circle*{1.5}} \put(87.5,11){\line(0,1){7.5}}
\put(87.5,11){\vector(0,1){3.75}}
\put(87.25,-1){\line(3,4){4.5}}\put(96.75,-1){\line(-3,4){4.5}}
\put(87.25,-1){\vector(3,4){3.5}}\put(96.75,-1){\vector(-3,4){3.5}}
\put(92,5){\circle*{1.5}}
\put(80.5,3){{\scriptsize$_2$}}\put(85.5,-3){{\scriptsize$_3$}}\put(97,-3){{\scriptsize$_1$}}
\end{picture}
\end{split}
\end{equation*}
Thus adding a coboundary amounts to the following equivalence
transformation for a pair of nearby vertices of a trivalent graph:
\begin{equation}\label{LieOp}
\begin{split}
\unitlength   0.9mm
\begin{picture}(60,30)(35,-10)
\put(30.5,4.75){\line(3,4){4.5}}\put(40,5){\line(-3,4){4.5}}
\put(30.5,4.75){\vector(3,4){3.5}}\put(40,5){\vector(-3,4){3.5}}
\put(35.5,11){\circle*{1.5}} \put(35.5,11){\line(0,1){7.5}}
\put(35.5,11){\vector(0,1){3.75}}
\put(35.25,-1){\line(3,4){4.5}}\put(44.75,-1){\line(-3,4){4.5}}
\put(35.25,-1){\vector(3,4){3.5}}\put(44.75,-1){\vector(-3,4){3.5}}
\put(40,5){\circle*{1.5}}
\put(28.5,3){{\scriptsize$_1$}}\put(33.5,-3){{\scriptsize$_2$}}\put(45,-3){{\scriptsize$_3$}}
\put(48,6){$\sim$}
\put(56.5,4.75){\line(3,4){4.5}}\put(66,5){\line(-3,4){4.5}}
\put(56.5,4.75){\vector(3,4){3.5}}\put(66,5){\vector(-3,4){3}}
\put(61.5,11){\circle*{1.5}} \put(61.5,11){\line(0,1){7.5}}
\put(61.5,11){\vector(0,1){3.75}}
\put(61.25,-1){\line(3,4){4.5}}\put(70.75,-1){\line(-3,4){4.5}}
\put(61.25,-1){\vector(3,4){3.5}}\put(70.75,-1){\vector(-3,4){3.5}}
\put(66,5){\circle*{1.5}}
\put(54.5,3){{\scriptsize$_3$}}\put(59.5,-3){{\scriptsize$_1$}}\put(71,-3){{\scriptsize$_2$}}
\put(74,6){$+$}
\put(82.5,4.75){\line(3,4){4.5}}\put(92,5){\line(-3,4){4.5}}
\put(82.5,4.75){\vector(3,4){3.5}}\put(92,5){\vector(-3,4){3}}
\put(87.5,11){\circle*{1.5}} \put(87.5,11){\line(0,1){7.5}}
\put(87.5,11){\vector(0,1){3.75}}
\put(87.25,-1){\line(3,4){4.5}}\put(96.75,-1){\line(-3,4){4.5}}
\put(87.25,-1){\vector(3,4){3.5}}\put(96.75,-1){\vector(-3,4){3.5}}
\put(92,5){\circle*{1.5}}
\put(80.5,3){{\scriptsize$_2$}}\put(85.5,-3){{\scriptsize$_3$}}\put(97,-3){{\scriptsize$_1$}}
\end{picture}
\end{split}
\end{equation}
It is easy to verify that applying this transformation properly
time and again,  one can bring any connected trivalent graph to a
linear combinations of the graphs (\ref{C-graph}) or
(\ref{B-graph}). For dimensional reasons, the graphs
(\ref{C-graph}) and (\ref{B-graph}) must then span a basis in the
space of nontrivial $\partial$-cocycles.

Interestingly enough that the equivalence  relation (\ref{LieOp})
coincides in form  with the defining relation for the quadratic
operad of Lie algebras  $\mathcal{L}ie$. Our computations then
show that the  graph complex $\mathcal{G}_{tree}=\bigoplus_{n>1}
\mathcal{T}^n$, being an algebra under the grafting of trees,
realizes the minimal resolution of the Lie algebra operad, that
is, $\mathcal{L}ie_{\infty}=({\mathcal{G}}_{tree} ,
\partial)$. For more details on the operadic interpretation of the
graph cohomology, including the case of cyclic graph complex
$\mathcal{G}_{cycle}=\bigoplus_{n>0}\mathcal{C}^n$, we refer the
reader to  \cite{Merkulov}, \cite{MMSh}.
\end{rem}

Let us return to $Q$-manifolds.  We define the stable
characteristic classes of a flat $Q$-manifold as the image of the
stable cohomology classes of
$\mathcal{T}(\mathbb{E})^{\mathrm{inv}}$ under the homomorphism
(\ref{chi}).  Using the explicit formulas (\ref{Gauss-in-coord})
for the characteristic map together with the correspondence map
$R$, one can easily convert the basis cocycles of the graph
complex into the universal cocycles of a flat $Q$-manifold.

Given a homological vector field $Q$ and a flat symmetric
connection $\nabla$, we write  $\Lambda$ for the odd endomorphism
defined  by the rule $\Lambda(X)=\nabla_X Q$ for $\forall X\in
\frak{X}(M)$. The covariant derivative of $\Lambda$ yields a
(2,1)-tensor $\nabla\Lambda$. The assignment $X\mapsto
\nabla_X\Lambda$ allows us to identify $\nabla\Lambda$ with a
right $C^\infty(M)$-module homomorphism from $\frak{X}(M)$ to
$\frak{A}(M)$. Graphically, the tensors $\Lambda=\chi\circ
R(\gamma_1)$ and $\nabla\Lambda=\chi\circ R(\gamma_2)$ are
depicted by the bivalent and trivalent corollas, respectively.
They play the role of building blocks for constructing stable
universal cocycles associated with the graphs (\ref{A-graph}),
(\ref{C-graph}) and (\ref{B-graph}). Namely, the graph cocycles
(\ref{A-graph}), having the form of polygons with $2n-1$ vertices,
give the following sequence of $Q$-invariant functions:
\begin{equation}\label{A}
    A_n=\mathrm{Str}(\Lambda^{2n-1})\,.
\end{equation}
The invariance of $A_n$ follows immediately from  the identity
$\nabla_Q \Lambda=-\Lambda^2$ and the cyclic property of the
supertrace. The cocycles (\ref{B-graph}) give rise to the sequence
of $Q$-invariant tensors $B_n\in \mathcal{T}^{n+1,1}(M)$.
Identifying the $(n+1,1)$-tensors with homomorphisms from
$\frak{X}(M)^{\otimes n}$ to $\frak{A}(M)$, we can write
\begin{equation}\label{B}
    B_n(X_1,X_2,...,X_n)=\nabla_{X_1}\Lambda\nabla_{X_2}\Lambda\cdots
    \nabla_{X_n}\Lambda\,.
\end{equation}
Taking now the trace, we get the universal cocycles corresponding
to the cyclic graphs (\ref{C-graph}):
\begin{equation}\label{C}
    C_n(X_1,X_2,...,X_n)=\mathrm{Str}B_n(X_1,X_2,...,X_n)\,.
\end{equation}
By construction, $C_n\in \mathcal{T}^{n,0}(M)$. The $Q$-invariance
of the tensor fields $\{B_n\}$ and $\{C_n\}$ can be checked
directly. In fact, it is enough to check the invariance of
$B_1=\nabla \Lambda$ as the other cocycles are made of $B_1$ by
means of tensor operations. We have
\begin{equation}\label{LB}
\begin{array}{rcl}
    (-1)^{\epsilon(X)}(L_Q B_1)(X)&=&L_Q B_1(X)-B_1([Q,X])\\[3mm]
    &=&\nabla_QB_1(X)+[\Lambda,B_1(X)]-B_1([Q,X])\\[3mm]
    &=&(-1)^{\epsilon(X)}\nabla_XB_1(Q)+[\Lambda,B_1(X)]\\[3mm]
    &=&-(-1)^{\epsilon(X)}\nabla_X\Lambda^2-[\Lambda,\nabla_X\Lambda]=0\,.
    \end{array}
\end{equation}

We will refer to the $\delta$-cohomology classes of the cocycles
(\ref{A}), (\ref{B}) and (\ref{C}) as the characteristic classes
of $A$-, $B$-, and $C$-series. The results of this section can now
be summarized as follows.

\begin{thm}\label{flatChCl}
In the stable range of dimensions, the characteristic classes of a
flat $Q$-manifold are generated by the characteristic classes of
$A$-, $B$-, and $C$-series by means of tensor products and
permutations of indices.
\end{thm}

\section{Intrinsic characteristic classes}
So far  we have dealt with construction and classification of the
characteristic classes associated to flat $Q$-manifolds. In this
section, we are going to extend the above consideration from the
flat to arbitrary $Q$-manifold. To begin with we note that the
straightforward substitution of an arbitrary, non-flat, symmetric
connection to formulas (\ref{A}-\ref{C}) fails to produce the
universal cocycles. Taking for example $A_1$ we get

\begin{equation*}\label{}
    \delta \mathrm{Str}(\Lambda)=\nabla_Q \mathrm{Str}(\nabla
    Q)=\frac12\mathrm{Str}(R_{QQ})\neq 0\,,
\end{equation*}
where
\begin{equation*}
R_{XY}=[\nabla_X,\nabla_Y]-\nabla_{[X,Y]} \in
\mathfrak{A}(M)\qquad \forall X,Y\in \frak{X}(M)
\end{equation*} is the curvature of
$\nabla$. We can then try to restore the $\delta$-closedness of
the tensors (\ref{A}-\ref{C}), where $\nabla$ is now an arbitrary
symmetric connection, by adding to them appropriate
curvature-dependent terms. Most easily this can be done for the
characteristic cocycles of $B-$ and $C-$series. Define the
$(2,1)$-tensor field $\mathrm{B}_1$ as a $C^\infty(M)$-linear
homomorphism from $\mathfrak{X}(M)$ to $\mathfrak{A}(M)$:
\begin{equation*}\label{}
    X\mapsto\mathrm{B}_1(X)=\nabla_X\Lambda-R_{XQ}\,.
\end{equation*}

A calculation similar to (\ref{LB}) shows that $\mathrm{B}_1$ is a
$Q$-invariant tensor for an arbitrary (i.e., not necessary flat)
connection \footnote{A simple way to convince oneself that $
\mathrm{B}_1$ is $\delta$-closed is to observe that
$\mathrm{B}_1=\delta\Gamma$, with $\Gamma=(\Gamma_{ij}^k)$ being
here the Christoffel symbols of the symmetric connection
$\nabla$.}. Therefore by making replacement $
B_1=\nabla\Lambda\mapsto {\mathrm{B}}_1 $ in (\ref{B}) and
(\ref{C}), we get two infinite series of universal cocycles
$\{\mathrm{B}_n\}$ and $\{\mathrm{C}_n\}$ that generalize $B$- and
$C$-series to arbitrary $Q$-manifolds. Theorem \ref{independing
thm} ensures that the $\delta$-cohomology classes of these
cocycles do not depend on the choice of symmetric connection, so
that we have two well-defined series of characteristic classes
$[\mathrm{B}_n]$, $[\mathrm{C}_n]\in H_Q(M)$. We call these
characteristic classes \textit{intrinsic} to stress the fact that
they may well be nontrivial even for a flat $Q$-manifold. In other
words, the intrinsic characteristic classes are intrinsically
related to the structure of the homological vector field rather
than to the topology of the underlying manifold. The latter may
prevent the existence of a flat connection and give rise to
universal cocycles that essentially involve the curvature and
vanish in the zero-curvature limit. These last cocycles are called
\textit{vanishing} ones. The typical examples of vanishing
universal cocycles are the following $Q$-invariant functions:

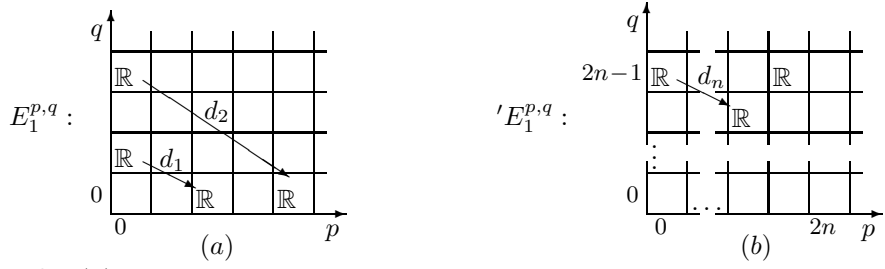
\begin{figure}[t]\label{SpectralSeq1}
\unitlength 0.54mm 

\begin{picture}(60,65)(55,-5)
\multiput(0,0)(0,10){5}{\line(1,0){53}}   
\multiput(0,0)(10,0){6}{\line(0,1){45}}   
\put(-25,22){\footnotesize{$E_{1}^{p, q}:$}}          
\put(1,11.5){\footnotesize{$\mathbb{R}$}}
\put(1,31.5){\footnotesize{$\mathbb{R}$}}
\put(21,1.5){\footnotesize{$\mathbb{R}$}}
\put(41,1.5){\footnotesize{$\mathbb{R}$}}
\put(8,13){\vector(2,-1){13}}  
\put(8,33){\vector(3,-2){36}}
\put(12,11.5){\footnotesize{$d_{1}$}} 
\put(23,23.5){\footnotesize{$d_{2}$}} \put(0,43){\vector(0,1){7}}
\put(-5,44){\footnotesize{$q$}} \put(53,0){\vector(1,0){5}}
\put(53,-5.5){\footnotesize{$p$}} \put(-5,3){\scriptsize{$0$}}
\put(1,-5){\scriptsize{$0$}} \put(22,-10){\footnotesize{$(a)$}}

\put(95,22){\footnotesize{$'E_{1}^{p, q}:$}}            
\multiput(132,0)(0,10){5}{\line(1,0){13}} \multiput(149,0)(0,10){5}{\line(1,0){36}}   
\multiput(132,0)(10,0){6}{\line(0,1){13}} \multiput(132,17)(10,0){6}{\line(0,1){28}}   
\put(133,31.5){\footnotesize{$\mathbb{R}$}} 
\put(116,33){\scriptsize{$2n\!-\!1$}}
\put(127,3){\scriptsize{$0$}} \put(134,-5){\scriptsize{$0$}}
\put(172,-5){\scriptsize{$2n$}}
\put(153,21.5){\footnotesize{$\mathbb{R}$}}
\put(163,31.5){\footnotesize{$\mathbb{R}$}}
\put(143,1){\footnotesize{$\ldots$}}
\put(133,11){\footnotesize{$\vdots$}}
\put(139.5,33){\vector(2,-1){13}} 
\put(144.5,32){\footnotesize{$d_{n}$}} 
\put(132,45){\vector(0,1){7}} \put(127,46){\footnotesize{$q$}}
\put(185,0){\vector(1,0){5}} \put(185,-5.5){\footnotesize{$p$}}
\put(155,-10){\footnotesize{$(b)$}}
\end{picture}

\caption{\protect\small{$(a)$ The transgression from the proof of
Theorem \ref{A-P}: the nonzero fiber and base cells  correspond
to the one-dimensional spaces generated by the universal cocycles
$A_n$ and $P_n$, respectively. $(b)$  The first term of the
spectral sequence from the proof of Theorem \ref{dexrep}.}}
\end{figure}

\begin{equation}\label{Pvan}
    P_n=\mathrm{Str}(\mathrm{R}^n)\,,
\end{equation}
where $\mathrm{R}=R_{QQ}\in \frak{A}(M)$. One can also view the
curvature of $\nabla$ as a two-form $\mathbf{R}$ with values in
$\frak{A}(M)$ and define the function $P_n$ as the value of the
$2n$-form $\mathbf{P}_n=\mathrm{Str}(\mathbf{R}^n)$ on the
homological vector field\footnote{The homological vector field
being odd, the complete contraction of $\mathbf{P}_n$ with
$Q^{\otimes 2n}$ is not equal to zero identically.},
$P_n=\mathbf{P}_n(Q)$. Then the $\delta$-closedness of $P_n$
immediately  follows from the $d$-closedness of $\mathbf{P}_n$:
\begin{equation*}\label{}
    \delta P_n= (d\mathbf{P}_n)(Q)=0\,.
\end{equation*}
The de Rham classes of the $2n$-forms $\mathbf{P}_n$ are known as
the Pontryagin characters of the tangent bundle $TM$. For a deeper
discussion of the Chern-Weyl theory of characteristic classes in
the category of supervector bundles we refer to \cite{Q} and
\cite{KSt}. The latter paper contains also a very natural
generalization of the Chern-Weyl construction to the category of
$Q$-bundles.

 The next
theorem establishes a one-to-one correspondence between the
vanishing cocycles (\ref{Pvan}) and the universal cocycles  of
$A$-series (\ref{A}).

\begin{thm}[\cite{LS}]\label{A-P}
Let $M$  be a $Q$-manifold with symmetric connection. Then for
each $n\in \mathbb{N}$ there exists an invariant matrix polynomial
$\mathrm{A}_n(\Lambda,\mathrm{R})\in C^\infty(M)$ in $\Lambda$ and
$\mathrm{R}$ such that

\begin{equation}\label {A-ser}
\begin{array}{c}
    \mathrm{A}_n(\Lambda,0)=\mathrm{Str}(\Lambda^{2n-1}) \quad \mbox{and}
    \quad \delta \mathrm{A}_n(\Lambda,\mathrm{R})=\binom{2n-1}{n} \mathrm{Str}(\mathrm{R}^{n})\,.
    \end{array}
\end{equation}
\end{thm}

\begin{cor}
The vanishing cocycles $P_n=\mathrm{Str}(\mathrm{R}^n)$ are
trivial.
\end{cor}

\begin{cor}
If $P_n=0$, then the function $\mathrm{A}_n(\Lambda, \mathrm{R})$
is a $\delta$-cocycle.
\end{cor}

\begin{proof}
We start with some auxiliary algebraic constructions, which model
our geometric situation. Namely, consider a DGA-algebra
$W=\bigoplus_{n>0} W_n$ over $\mathbb{R}$ freely generated by one
element $a$ of degree $1$ and one element $b$ of degree $2$. The
elements of $W$ are just linear combinations of words associated
to the binary alphabet $\{a,b\}$ and multiplication is given by
concatenation of words. Denote by $|w|$ the degree of the word
$w$.  The action of the differential $d:W_n\rightarrow W_{n+1}$ on
the generators reads
\begin{equation}\label{diff}
    d a= a^2+b\,,\qquad db=[a,b]\,,
\end{equation}
and extends  to the whole $W$ by the Leibniz rule:
\begin{equation*}
d(w_1w_2)=dw_1w_2+(-1)^{|w_1|}w_1dw_2\,.
\end{equation*}

We define the \textit{cyclic space} $\widehat{W}$ as the quotient
$\widehat{W}=W/[W,W]$, where the subspace $[W,W]\subset W$ is
generated by commutators
\begin{equation*}
w_1w_2-(-1)^{|w_1||w_2|}w_2w_1\,. \end{equation*} The action of the
differential (\ref{diff}) passes trough the quotient making the
cyclic space $\widehat{W}$ into a cochain complex. Explicitly,
\begin{equation*}\label{}
    \hat{d}a=-a^2+b\,,\qquad \hat{d}b=0\,.
\end{equation*}
The complex $(\widehat{{W}}, \hat d)$ is acyclic. This can be
easily seen from the filtration of $\widehat{W}$ by the length of
words:
\begin{equation*}\label{}
      W_n= F'_1\widehat{W}_n\supset F'_2\widehat{W}_n \supset
      \cdots\supset F'_{n+1} \widehat{W}_n=0\,.
\end{equation*}
The zero differential of the corresponding spectral sequence
$\{E'_r, d'_r\}$ acts on the generators as
\begin{equation*}\label{}
    d'_0a=b \,, \qquad d'_0b=0\,.
\end{equation*}
So, the complex $(E'_0,d'_0)$ is obviously acyclic and
$H(\widehat{W})=0$.

Let us now introduce one more grading on the space $\widehat{W}$
by prescribing the following degrees to the generators:
$$\mathrm{deg}\, a=0\,,\qquad \mathrm{deg}\, b =2\,.$$
Associated to this grading is a decreasing filtration
\begin{equation}\label{barF}
    \widehat{W}_n={F}_0\widehat{W}_n\supset
    {F}_1\widehat{W}_n\supset\cdots\supset {F}_{n+1}\widehat{W}_n=0\,.
\end{equation}
Here ${F}_p\widehat{W}=\bigoplus_{n\geq p}\widehat{W}^{(n)}$ and
$\mathrm{deg}\, \widehat{W}^{(n)}=n$. Since the deferential
(\ref{diff}) preserves the filtration (\ref{barF}), we have the
first quadrant spectral sequence $\{{E}_r,{d}_r\}_{r\geq 0}$
converging to $H(\widehat{W})=0$. By definition,
${E}_0^{p,q}={{F}}_p\widehat{W}_{p+q}/{F}_{p+1}\widehat{W}_{p+q}$.
The zero differential of this spectral sequence is completely
defined by its action on the generators of $W$:
\begin{equation}\label{d0}
    {d}_0 a=-a^2\,,\qquad {d}_0b=0\,.
\end{equation}
Notice that ${E}^{p,q}_0=0$, if $p$ is odd. The complex $({E}_0,
{d}_0)$ splits into the direct product ${E}_0=A\oplus B\oplus C$
of three subcomplexes. The space $A=\bigoplus_{q> 0} {E}^{0,q}$ is
generated by the odd powers of the letter $a$,
$B=\bigoplus_{p>0}{E}^{p,0}$ is given by polynomials in  $b$, and
$C=\bigoplus_{p,q>0}{E}^{p,q}$ is spanned by cyclic words
containing at least one syllable $ab$. It is clear from (\ref{d0})
that $H(A)\simeq A$ and $H(B)\simeq B$. The remaining complex $C$
turns out to be acyclic. Indeed, any word $w\in C$ is decomposed
into a product $w=a_{n_1}a_{n_2}\cdots a_{n_k}$ of syllables
$a_n=a^nb$ and this decomposition is unique up to a cyclic
permutation of factors. We have
\begin{equation*}\label{}
    {d}_0 a_n=\left\{%
\begin{array}{ll}
    - a_{n+1}, & \hbox{if $n$ is odd;} \\
    0, & \hbox{otherwise.} \\
\end{array}%
\right.
\end{equation*}
Define a differentiation  $h: C\rightarrow C$ by its action on
syllables:
\begin{equation*}\label{}
    h a_n=\left\{%
\begin{array}{ll}
    - a_{n-1}, & \hbox{if $n>0$ is even;} \\
    0 , & \hbox{otherwise.} \\
\end{array}%
\right.
\end{equation*}
Applying the operator $\Delta=d_0h+h d_0$ to a word $w\in C$, we
get $\Delta w= mw$, where $m>0$ is the number of syllables $a_n$
with $n>0$ the word $w$ consists of. Thus, the operator $\Delta$
is invertible and $\Delta^{-1}h:C\rightarrow C$ is a contracting
homotopy. The nonzero entries of the group ${E}_1$ are depicted in
Fig.$2a$. Since ${E}_r\Rightarrow H(W)=0$, each nonzero element of
the fiber space is transgressive and ``kills'' some element on the
base. In other words, the maps ${d}_{r}:
{{E}}_{r}^{0,r-1}\rightarrow {{E}}_r^{r, 0}$, $r=2,4,...$, are
isomorphisms of one-dimensional vector spaces, so that ${d}_{2n}
a^{2n-1}=\alpha_n b^n$ for some $\alpha_n\neq 0$. The last
equation amounts to the existence of a sequence of elements
$c_m\in {E}_0^{2m,2(n-m)-1}$  such that
\begin{equation}\label{trans}
    d(a^{2n-1}+ c_1+c_2+\cdots + c_{n-1})=\alpha_n b^n\,.
\end{equation}
With a little work one can also find that
$\alpha_{n}=\binom{2n-1}{n}$.

Comparing (\ref{diff}) with the identities
\begin{equation*}\label{}
\delta\Lambda =\Lambda^2 + \frac12\mathrm{R} \,,\qquad \delta
\mathrm{R}=[\Lambda, \mathrm{R}]\,,
\end{equation*}
we see that  the homomorphism $h: W\rightarrow \mathfrak{A}(M)$
given by
\begin{equation}\label{rep}
    a\mapsto \Lambda \,,\qquad b\mapsto \frac12\mathrm{R}
\end{equation}
is a representation  of DGA-algebra $(W,d)$ as a subalgebra of
$(\mathfrak{A}(M),\delta)$. In view of the cyclic property of the
trace, the composition $\mathrm{Str}\circ h: W\rightarrow
C^\infty(M)$ descends to the cyclic space $\widehat{W}$ defining a
homomorphism of complexes. Applying this homomorphism to both
sides of Eq.(\ref{trans}), we conclude that the invariant matrix
polynomial
\begin{equation*}
\mathrm{A}_n(\Lambda,\mathrm{R})=\mathrm{Str}\circ h(a^{2n-1}+
c_1+c_2+\cdots + c_{n-1})\end{equation*} satisfies
Eqs.(\ref{A-ser}). It is not hard to write explicit expressions
for the first polynomials:
\begin{equation*}
\begin{array}{l}
\mathrm{A}_{1} = \mathrm{Str}(\Lambda)\,, \\[3mm]
\mathrm{A}_{2} = \mathrm{Str}(\Lambda^{3} + \frac32
\mathrm{R}\Lambda)
\,,  \\[3mm]
\mathrm{A}_{3} = \mathrm{Str}(\Lambda^{5} + \frac52 \mathrm{R}\Lambda^{3} + \frac{10}{4} \mathrm{R}^{2} \Lambda)\,,  \\[3mm]
\mathrm{A}_{4} = \mathrm{Str}(\Lambda^{7} + \frac72
\mathrm{R}\Lambda^{5} + \frac{14}{4} \mathrm{R}^{2} \Lambda^{3} +
\frac{7}{4} \mathrm{R} \Lambda \mathrm{R} \Lambda^{2} +
\frac{35}{8}\mathrm{R}^{3} \Lambda)\,. \label{series}
\end{array}
\end{equation*}

\end{proof}

In fact, Theorem \ref{A-P} identifies the trivial
$\delta$-cocycles $P_n$ with obstructions to extendability  of
$A$-series' characteristic classes to the non-flat $Q$-manifolds.
To prove this, we only need to show that the $\delta$-cocycles
$P_n$ become nontrivial when considered as elements of the
subcomplex of vanishing covariants $\mathcal{R}\subset
\mathcal{A}$ (see comments after (\ref{ExTr})). The proof of the
last fact is not particularly interesting and we relegate it to
Appendix \ref{A2}. Summarizing, it appears impossible to define
the construction of $A$-series (\ref{A}) in the full category of
$Q$-manifolds, i.e., without any conditions on the geometry of
$M$. Of course, the assumption that $M$ admits a flat connection
is unduly restrictive; instead, one may only suppose the
triviality of the $n$th Pontryagin character of the tangent
bundle $TM$. In the latter case there exists a form
$\mathbf{F}_n\in \Omega^{2n-1}(M)$ such that
$\mathbf{P}_n=d\mathbf{F}_n$. The form $\mathbf{F}_n$ is not
uniquely determined by the connection $\nabla$ as one is free to
add to $\mathbf{F}_n$ any closed $(2n-1)$-form. If
$F_n=\mathbf{F}_n(Q)$, then $\delta F_n=P_n$ and  by Theorem
\ref{A-P} we have the $Q$-invariant function
\begin{equation}\label{A-F}
\begin{array}{c}
    {\mathrm{A}}^F_n=\mathrm{A}_n({\Lambda},
    {R})-\binom{2n-1}{n}F_n\,.
    \end{array}
\end{equation}
If $\mathbf{F}'_n$ is another potential for the Pontryagin form,
i.e., $\mathbf{P}_n=d\mathbf{F}'_n$, then the difference
$\mathbf{K}_n=\binom{2n-1}{n}(\mathbf{F}_n-\mathbf{F}'_n)$ is
closed and defines the de Rham class $[\mathbf{K}_n]\in
H_d^{2n-1}(M)$. We have
\begin{equation}\label{F-F}
\mathrm{A}^{F'}_n-\mathrm{A}^{F}_n=K_n\,,
\end{equation}
where $K_n=\mathbf{K}_n(Q)$. In case $\mathbf{K}_n=d\mathbf{W}_n$,
the r.h.s. of (\ref{F-F}) is given by the coboundary $\delta
\mathbf{W}_n(Q)$, so that the $\delta$-cocycles
$\mathrm{A}^{F'}_n$ and $\mathrm{A}^F_n$ appear to be cohomologous
whenever $[\mathbf{K}_n]=0$. This  motivates us to consider the
homomorphism
 \begin{equation*}\label{}
h_Q: \Omega(M)\rightarrow C^{\infty}(M)
\end{equation*}
that evaluates an exterior form on the homological vector field
$Q$. Since $h_Q d+\delta h_Q=0$, $h_Q$ induces the homomorphism in
cohomology
\begin{equation*}\label{h}
    h^\ast_Q: H_d(M)\rightarrow H_\delta(M)\,.
\end{equation*}
Passing in (\ref{F-F}) to the $\delta$-cohomology classes,  we can
write
\begin{equation*}\label{}
    [\mathrm{A}_n^{F'}]- [\mathrm{A}_n^{F}]\in \mathrm{Im}h_Q^\ast\,.
\end{equation*}
The last formula just amounts to saying that the class
$[\mathrm{A}_n^{F}]+\mathrm{Im} h^\ast_Q $ does not actually
depend on the choice of $F_n$. Thus, we have proved the following

\begin{thm}
Let $M$ be a $Q$-manifold. Assume that the $n$th Pontryagin
character of the tangent bundle $TM$ is trivial. Then the
$Q$-invariant function (\ref{A-F}) represents a well-defined
element of $H_Q(M)/\mathrm{Im}h_Q^\ast$.
\end{thm}

It is well  known that all the Pontryagin characters
$[\mathbf{P}_{2m+1}]$ vanish for the real vector bundles over the
usual (even) manifolds, and the same statement holds true in the
category of supervector bundles. This means the existence of
classes $[\mathrm{A}_{2m+1}^{F}]+\mathrm{Im} h^\ast_Q $ for all
$Q$-manifolds. Moreover, each class
$[\mathrm{A}_{2m+1}^{F}]+\mathrm{Im} h^\ast_Q $ contains a
canonical representative  $[A^0_{2m+1}]\in H_Q(M)$, which can be
viewed as an invariant of the $Q$-manifold itself. The existence
of such a representative is established by the next two
propositions whose proofs can be found in Appendix \ref{A3}.
\begin{prop}\label{6.1}
Every  supermanifold $M$ admits a symmetric affine connection
$\nabla$ with $\mathbf{P}_{2m+1}^\nabla=0$ for all $m\geq 0$.
\end{prop}

Since the construction of $\nabla$ satisfying the property above
utilizes two auxiliary metrics, we will refer to $\nabla$ as the
\textit{metric connection}. Using this connection, we can define
the series of $\delta$-cohomology classes
$[\mathrm{A}^0_{2m+1}]\in H_Q(M)$ represented by the functions
(\ref{A-F}) with $F_{2m+1}=0$. One can view these classes as a
proper generalization for the half of the scalar characteristic
classes of $A$-series (\ref{A}) to the case of arbitrary (i.e.,
non-flat) $Q$-manifolds. The construction of the metric connection
involves a great deal of ambiguity concerning the choice of the
metrics, and it is important to check that this ambiguity does not
affect on the classes $[\mathrm{A}^0_{2m+1}]$.

\begin{prop}\label{6.2}
The class $[\mathrm{A}^0_{2m+1}]\in H_Q(M)$ is independent of the
choice of metric connection.
\end{prop}

In the particular case of homological vector fields coming from
the Lie algebroids (see Example \ref{LA}) the corresponding
characteristic classes $[\mathrm{A}_{2m+1}^0]$ add up to the
secondary characteristic classes of Lie algebroids that were
introduced and studied by Fernandes \cite{F} within the framework
of classical differential geometry. The class $[\mathrm{A}^0_1]$,
called the \textit{modular class of $Q$-manifolds} \cite{LS},
unifies and generalizes the well-known constructions of the
modular classes of (complex) Poisson manifolds \cite{W},
\cite{BZ}, Lie algebroids \cite{ELW}, and the Lie-Rinehart
algebras \cite{H}; hence the name. Let us look at it more closely.

Given a metric connection $\nabla$, the modular class is
represented by the covariant divergence of the homological vector
field
\begin{equation*}
\mathrm{A}^0_1=\mathrm{Str}(\Lambda)=\nabla_iQ^i\,.
\end{equation*}
Equivalently, it can be defined in terms of a nowhere vanishing
density $\rho$, instead of the metric connection:
\begin{equation*}\label{}
    \mathrm{A}_1^0 = \mathrm{div}_\rho Q=\rho^{-1}\partial_i(\rho
    Q^i)\,.
\end{equation*}
For instance, one can take $\rho=(\det g^0)^{\frac12}(\det
g^1)^{\frac12}$, where $g^0$ and $g^1$ are metrics entering the
definition of $\nabla$. If $\rho'$ is another density on $M$, then
${\mathrm{A}'}_1^0$ is obviously cohomologous to $\mathrm{A}_1^0$:
\begin{equation*}\label{}
    {\mathrm{A}'}_1^0-\mathrm{A}_1^0= \delta f\,,\qquad
    f=\mathrm{ln}(\rho'/\rho)\in C^{\infty}(M)\,.
\end{equation*}
Writing $\mathrm{A}_1^0=\rho^{-1}L_Q\rho$ we see that the
vanishing modular class $[A_1^0]\in H_Q(M)$ is the necessary and
sufficient condition for $M$ to admit a $Q$-invariant nowhere
vanishing density.

We conclude this section with the following extension of Theorem
\ref{flatChCl} to non-flat $Q$-manifolds.

\begin{thm}\label{intchar}
Given  a $Q$-manifold with metric connection, the tensor algebra
of intrinsic characteristic classes is  generated by the
$\delta$-cohomology classes of the universal cocycles
$\{\mathrm{A}_{2m+1}^0\}$, $\{\mathrm{B}_n\}$, and
$\{\mathrm{C}_n\}$ with the help of tensor products and
permutations of indices. The intrinsic characteristic classes
depend only on the $Q$-manifold, not on the choice of the metric
connection.
\end{thm}

\section{Characteristic classes with values in forms}

If $M$ is a $Q$-manifold, then the algebra of exterior forms
$\Omega(M)$ carries a pair of commuting differentials: the
exterior differential $d$ and the Lie derivative $\delta=L_Q$. In
this section, we discuss an interesting interplay between both the
differentials  in the context of characteristic classes. For this
purpose, let us consider the bigraded, bidifferential, associative
algebra ${\mathbb{{A}}}=\mathcal{A}\cap\Omega(M)$ whose elements
are form-valued local covariants associated to the homological
vector field $Q$ and (not necessary metric) connection $\nabla$.
The first grading in $\mathbb{A}=\bigoplus \mathbb{A}^{n,m}$ is
just the form degree, while the second one is the degree of
homogeneity of an element $f\in \mathbb{A}$ as a function of $Q$:
\begin{equation*}\label{}
    \mathbb{A}^{\bullet,\, m\mathcal{}}\ni f \quad \Leftrightarrow\quad
    f(tQ)=t^m
    f(Q)\qquad \forall t\in \mathbb{R}\,.
\end{equation*}
The commuting differentials $d$ and $\delta$ increase the
respective degrees by one, making $\mathbb{A}$ into a
multiplicative bicomplex.

The intrinsic part of the $\delta$-cohomology of the differential
algebra $\mathbb{A}$ has been in fact computed in the previous
section. It follows from Theorem \ref{intchar} that the
multiplicative basis of intrinsic $\delta$-cocycles is given by
the forms
\begin{equation}\label{BB}
{\mathbf{C}}_n=\mathrm{Str}({{\mathrm{B}}}_1^n)\in
\mathbb{A}^{n,n}\,,\quad n\in \mathbb{N}\,,
\end{equation}
where
\begin{equation*}\label{FB}
    {\mathrm{B}}_1=\nabla \Lambda - i_Q
    \mathbf{R}\in \Omega^1(M)\otimes\frak{A}(M) \,.
\end{equation*}
Hereafter  we treat the connection as a first-order differential
operator  $\nabla: \Omega^\bullet(M)\otimes \frak{X}(M)\rightarrow
\Omega^{\bullet+1}(M)\otimes \frak{X}(M)$ satisfying the Leibniz
rule:
\begin{equation*}\label{}
\begin{array}{c}
    \nabla (\omega\otimes u)=d\omega\otimes u + (-1)^{n+\epsilon(\omega)}
    \omega\otimes \nabla u\,,\\[3mm]
    \forall\omega\in \Omega^n(M)\,,\quad \forall u\in \mathfrak{X}(M)\,.
\end{array}
\end{equation*}
The action of $\nabla$ is then  canonically extended from
$\Omega(M)\otimes \frak{X}(M)$ to $\Omega(M)\otimes
\mathcal{T}(M)$ by usual formulas of differential geometry. The
curvature of the connection is the matrix-valued two-form
$\nabla^2=\mathbf{R}\in \Omega^2(M)\otimes \frak{A}(M)$. One can
view $\Omega(M)\otimes \frak{A}(M)$ as an associative graded
superalgebra over $\Omega(M)$ with trace. Multiplication in
$\Omega(M)\otimes\mathfrak{A}(M)$ is given by the exterior product
in $\Omega(M)$ and the composition of endomorphisms in
$\mathfrak{A}(M)$, and the supertrace entering the definition
(\ref{BB}) is defined in a natural way as the $\Omega(M)$-linear
map $\mathrm{Str}: \Omega(M)\otimes \frak{A}(M)\rightarrow
\Omega(M)$ vanishing on commutators.

As is seen from the definition, the $Q$-invariant forms
$\{\mathbf{C}_n\}$ are given by the totally antisymmetric part of
the corresponding $\delta$-cocycles (\ref{C}).

The next theorem is the main result of this section.

\begin{thm}\label{dexrep} Let $M$ be a $Q$-manifold with $[\mathbf{P}_n]=0$.
 Then the $\delta$-cohomology class
$[{\mathbf{C}}_n]\in H_Q(M)$ contains  a $d$-exact representative.
\end{thm}

Our proof will be based on comparing two spectral sequences
canonically  associated to some bicomplex $V\subset {\mathbb{A}}$.

Consider first the graded subalgebra $\bar{W}=\langle
a_0,a_1,\cdots, a_6\rangle\subset \Omega(M)\otimes
\mathfrak{A}(M)$ generated by seven tensor fields:
\begin{equation*}\label{}
\begin{array}{llll}
    a_0 = \Lambda\,, &
    a_1=i_Q^2\mathbf{R} \,, &  a_2=i_Q\mathbf{R}\,, \quad& a_3 = \mathbf{R}\,,\\[3mm]
    a_4
    =\nabla a_0 +a_{2}\,,\quad & a_5 =\nabla a_1\,,\quad &  a_6 =\nabla
    a_2\,.&
    \end{array}
\end{equation*}
Note that $\nabla a_3=0$ in virtue of the Bianchi identity.  Let
$W$ denote the image of $\bar{W}$ under  the map
$\mathrm{Str}:\Omega(M)\otimes\mathfrak{A}(M)\longrightarrow
\Omega(M)$. Since $\nabla^2a_i=[a_3, a_i]\in \bar{W}$,  the
subalgebra $\bar{W}\subset \Omega(M)\otimes\mathfrak{A}(M)$ is
invariant under the action of $\nabla$ and we have the commutative
diagram
$$
\xymatrix{{{\bar{W}^p}}\ar[r]^-{\nabla}\ar[d]_-{\mathrm{Str}} &{\bar{W}^{p+1}}\ar[d]^-{\mathrm{Str}} \\
{W^p} \ar[r]^-d     &  {W^{p+1}} }
$$
Therefore  $(W, d)$ is a subcomplex of the de Rham complex of $M$.
Furthermore, $W$ is invariant under the action of the Lie
derivative $\delta=L_Q$.  Using the general relation $\delta
a_i=\nabla_Q a_i + [a_0,a_i]$ and the Bianchi identity for the
curvature tensor one can readily  check that
$$
\begin{array}{ll}\label{delta action}
    \delta a_0 = a_0^2-\frac{1}{2}a_1\,, \;\,\qquad \delta a_1=
    [a_0,a_1]\,, & \delta a_2=
    [a_0,a_2]-\frac12 a_5\,, \\[5mm] \delta a_3=
    [a_0,a_3]-a_6\,, \quad  \delta a_4= 0\,, &\delta a_5=
    [a_0,a_5]-[a_2,a_1]\,, \\[5mm] \delta a_6=
    [a_0,a_6]-\frac12[a_3,a_1]\,.&
\end{array}
$$

Thus, ${W}\subset {\mathbb{A}}$ is a bicomplex. The following
theorem computes the $d$- and $\delta$-cohomology groups of $W$.
\begin{lem}\label{bicomplex cohom}
$$
\begin{array}{lcl}
H_\delta^q ({W}^{p,\,\bullet})&=&\left\{%
\begin{array}{ll}
    \mathbb{R}, & \hbox{if $p=q$;} \\
    0, & \hbox{otherwise.} \\
    \end{array}\right.
\\[5mm]
H_d^p({W}^{\bullet,\, q})&=&\left\{%
\begin{array}{ll}
    \mathbb{R}, & \hbox{if  $p$ is even and $q=0$;} \\
    0, & \hbox{otherwise.} \\
\end{array}%
\right.
\end{array}
$$
The corresponding nontrivial $\delta$- and d-cocycles can be
chosen as
$$
{\mathbf{C}}_n = \mathrm{Str}\, a_4^n\in {W}^{n,\,n}\,,\qquad
\mathbf{P}_n = \mathrm{Str}\,a_3^n\in {W}^{2n,\,0}\,.
$$
\end{lem}
\begin{proof}
The cocycles ${\mathbf{C}}_n$ and $\mathbf{P}_n$ are obviously
nontrivial. To prove that these span the space of all nontrivial
cocycles we will construct two homotopy operators $h_1$ and $h_2$
such that
\begin{equation*}
\begin{array}{cc}
    h_1 \delta +\delta h_1 = \Delta_1\,,&  h_2 d +d h_2
    =\Delta_2\,,\\[3mm]
    \mathrm{Ker}\,\Delta_1 = \mathrm{span} (\mathbf{C}_1,\mathbf{C}_2, ...\,)\,, &\quad
    \mathrm{Ker}\,\Delta_2 =\mathrm{span }(\mathbf{P}_1, \mathbf{P}_2, ...\,)\,.
    \end{array}
\end{equation*}
Define the following pair of odd differentiations of the algebra
$\bar{W}$:
\begin{eqnarray}
\begin{array}{llll}
 \bar{h}_1a_0=0\,, \qquad&\bar{h}_1a_1=-2a_0\,,\qquad& \bar{h}_1a_2=0
\,,\qquad&  \bar{h}_1a_3=0\,, \\[3mm]
\bar{h}_1a_4=0\,, &
\bar{h}_1a_5=-2a_2\,, & \bar{h}_1a_6=-a_3\,;&\\[5mm]
 \bar{h}_2a_0=0\,, &\bar{h}_2a_1=0\,,&
 \bar{h}_2a_2=0\,, &
 \bar{h}_2a_3=0\,, \\[3mm] \bar{h}_2a_4=a_0\,, &
 \bar{h}_2a_5=a_1\,, & \bar{h}_2a_6=a_2\,.&
 \end{array}
\end{eqnarray}
Writing
\begin{equation*}
\bar \Delta_1 =\bar h_1\delta+\delta\bar h_1\,,\qquad \bar
\Delta_2=\bar h_2 \nabla + \nabla\bar h_2\,,
\end{equation*}
we find
\begin{eqnarray}\label{Delta-a}
\begin{array}{llll}
 \bar\Delta_1a_0=a_0\,, \quad &\bar\Delta_1a_1=a_1+2a_0^2\,, \quad &
 \bar\Delta_1a_2=a_2
\,, \qquad\bar\Delta_1a_3=a_3\,,& \\[3mm]
\bar\Delta_1a_4=0\,, &
\bar\Delta_1a_5=a_5+2[a_0,a_2]\,, & \bar\Delta_1a_6=a_6+[a_0,a_3]\,;&\\[5mm]
 \bar\Delta_2a_i=a_i\,, & \forall i\neq3\,, &&\\[3mm]
 \bar\Delta_2a_3=0\,.&&&
 \end{array}
\end{eqnarray}
Now for any $\bar b=a_{i_1}\cdots a_{i_k}\in \bar{W}$ and
$b={\mathrm{Str}}(\bar b)\in W$ we set
\begin{equation*}\label{}
    h_1b= \mathrm{Str}(\bar{h}_{1}\bar b)\,,\qquad h_2b=
    \mathrm{Str}(\bar{h}_{2}\bar b)\,.
\end{equation*}
It is clear that applying the operator $\Delta_2=[d, h_2]$ to $b$
yields
\begin{equation*}\label{}
    \Delta_2 b = \mathrm{Str}(\bar \Delta_2 \bar b)=n b\,,
\end{equation*}
where $n$ is the number of letters $a_i$ in the word $\bar
b=a_{i_1}\cdots a_{i_k}$ which are different from $a_3$. So the
elements $\mathbf{P}_n=\mathrm{Str}\,a_3^n$ do span the kernel of
$\Delta_2$.

To compute the kernel of $\Delta_1=[\delta,h_1]$ consider  the
filtration of the space $W$ by the length of words in $a_i$:
\begin{equation*}\label{}
    W^{(n)} \ni b \quad \Leftrightarrow\quad b= \sum_{k=n}^\infty
    b_k \,, \qquad b_k \in \mathrm{span} (\mathrm{{Str}}(a_{i_1}\cdots a_{i_k}))\,.
\end{equation*}
As is seen from (\ref{Delta-a})  the operator $\Delta_1$ preserves
the filtration and the equality $\Delta_1 b =0$ implies
$b_n=\beta_n \mathbf{C}_n$ for some $\beta_n\in \mathbb{R}$. Since
$\mathbf{C}_n\in \mathrm{Ker}\Delta_1$, we have
$\Delta_1(b-b_n)=0$ and hence
$b_{n+1}=\beta_{n+1}\mathbf{C}_{n+1}$ for some $\beta_{n+1}\in
\mathbb{R}$. Proceeding by induction, we deduce that
$b=\sum_{k=n}^\infty \beta_{k}\mathbf{C}_k$.
\end{proof}

Suppose now that the $n$th Pontryagin character of the tangent
bundle $TM$ is trivial. To model this situation algebraically we
extend ${W}$ by the bigraded space
$$\bar K=\mathrm{span}(\mathbf{F}_n^0,\mathbf{F}_n^1,...,\mathbf{F}_n^{2n-1};
\bar{\mathbf{F}}_n^0,\bar{\mathbf{F}}_n^1,...,\bar{\mathbf{F}}_n^{2n-2})\,,$$
where\footnote{In the previous section, the $(2n-1)$-form
$\mathbf{F}_n^0$ was denoted by $\mathbf{F}_n$, and the function
$\mathbf{F}_n^{2n-1}$ was denoted by $F_n$.}
\begin{equation}\label{dddelta}
    \mathbf{F}_n^k = i_Q^k\mathbf{F}_n^0\,,\qquad \bar{\mathbf{F}}^k_n=\delta\mathbf{F}_n^k\,,\qquad
    d\mathbf{F}_n^0=\mathbf{P}_n=\mathrm{Str}(\mathbf{R}^n)\,.
\end{equation}
It follows immediately from the definition that the extended space
$V=W\oplus \bar K$ is invariant under the action of $d$ and
$\delta$. Indeed, using the Cartan formula $\delta=di_Q+i_Qd$, we
find
$$
d\mathbf{F}_n^k=k\bar{\mathbf{F}}^{k-1}_n+i_Q^k
\mathbf{P}_n\,,\qquad d\bar{\mathbf{F}}^k_n=-\frac{1}{k+1}
di^{k+1}\mathbf{P}_n\,,
$$
\begin{equation}\label{ddelta}
\delta\mathbf{F}_n^{2n-1}=i_Q^{2n}\mathbf{P}_n=P_n\,.
\end{equation}
Thus, $(V, d,\delta)$ is a bicomplex having $(W, d,\delta)$ as a
subcomplex.
\begin{lem}\label{lem2}
$$
\begin{array}{lcl}
H^q_\delta ({V}^{p,\,\bullet})&=&\left\{%
\begin{array}{ll}
    \mathbb{R}, & \hbox{if $p=q$ or $p=0$ and $q=2n-1$;} \\
    0, & \hbox{otherwise.} \\
    \end{array}\right.
\\[5mm]
H^p_d({V}^{\bullet,\, q})&=&\left\{%
\begin{array}{ll}
    \mathbb{R}, & \hbox{if  $q=0$ and $p$ is even and not equal to $2n$;} \\
    0, & \hbox{otherwise.} \\
\end{array}%
\right.
\end{array}
$$
\end{lem}

In plain English, the lemma says that vanishing of the $n$th
Pontryagin character ``kills'' one class of $d$-cohomology in
degree $2n$, giving simultaneously birth to a new
$\delta$-cohomology class in degree $2n-1$.

\begin{proof} Define the quotient complex ${K}={V}/{W}$.  As a linear
space $K\simeq \bar K$.  The short exact sequence
\begin{equation*}
0\longrightarrow
{{W}}\stackrel{i}{\longrightarrow}{{V}}\stackrel{p}{\longrightarrow}{{K}}\longrightarrow
0\,,
\end{equation*}
gives rise to a long exact sequence in cohomology.

For the $d$-cohomology we have
$$
 \cdots\rightarrow
H^{2n-1}_d({K})\stackrel{\partial}{\longrightarrow}H^{2n}_d({W})\stackrel{i_\ast}{\longrightarrow}H_d^{2n}({V})
\stackrel{p_\ast}{\longrightarrow} H^{2n}_d({K})\rightarrow\cdots
$$
It follows from (\ref{dddelta}) that $\mathbf{F}_n^0$ is a
nontrivial $d$-cocycle of ${K}^{2n-1,0}$ and
$\partial[\mathbf{F}_n^0]=[\mathbf{P}_n]$, where $[\mathbf{P}_n]$
spans $H^{2n}_d(W)\simeq \mathbb{R}$. Hence  $\partial$ is epic
and $i_\ast=0$. Since $H^{2n}_d({K})=0$, we conclude that
$H^{2n}_d({V})=0$.

Consider now the long exact sequence for the $\delta$-cohomology
groups:
$$
\begin{array}{rl}
\cdots    \rightarrow
H_\delta^{2n-2}({K})\stackrel{\partial'}{\longrightarrow}
    H_\delta^{2n-1}({W})&\stackrel{i_\ast}{\longrightarrow}H_\delta^{2n-1}({V})\rightarrow\\[3mm]
    &\stackrel{p_\ast}{\longrightarrow}H_\delta^{2n-1}({K})
    \stackrel{\partial}{\longrightarrow}H^{2n}_\delta({W})\rightarrow\cdots
\end{array}
$$
Eq. (\ref{ddelta}) implies that $H^{2n-1}_\delta({K})\simeq
\mathbb{R}$ and $\partial[\mathbf{F}_n^{2n-1}]=[P_n]$. By Theorem
\ref{A-P}, ${P}_n$ is proportional to $\delta
\mathrm{A}_n(\Lambda, \mathrm{R})$. Hence $[P_n]=0$ and
$\partial=0$. Since $H_\delta^{2n-2}({K})=0$ and
$H_\delta^{2n-1}(W)\simeq \mathbb{R}$, we infer  that
$H_\delta^{2n-1}({V})\simeq\mathbb{R}^2$. As a vector space the
group $H_\delta^{2n-1}({V})$ is generated by the
$\delta$-cohomology classes $[\mathbf{C}_n]$ and
$[\mathrm{A}^F_n]$.

Considering the other segments of the long exact sequences above,
one can easily verify that $H^m_d({W})=H^m_d({V})$ and
$H^{m-1}_\delta({W})=H^{m-1}_\delta({V})$ for all $m\neq 2n$. The
details are left to the reader.
\end{proof}

Now we are in position to prove Theorem \ref{dexrep}. Consider the
total complex $\bar{V}=\mathrm{Tot} {V}$ of the bicomplex ${V}$:
\begin{equation*}\label{}
   \bar V^n=\bigoplus_{p+q=n}{V}^{p,q}\,,\qquad D = d+\delta:
    \bar V^n\rightarrow \bar V^{n+1}\,.
\end{equation*}
Let $'\!E=\{'\!E^{p,q}_r,d'_r\}$ and $''\!E=\{'\!E^{p,q}_r,
d''_r\}$ denote two spectral sequences associated to the first and
second filtrations
\begin{equation*}\label{}
F'_p\bar V^n=\bigoplus_{s\geq p} {V}^{s,n-s}\,,\qquad F''_p\bar
V^n=\bigoplus_{s\geq p}{V}^{n-s,s}
\end{equation*}
of the total complex $\bar V$. By definition, $'\!E_1^{p,q}\simeq
H_\delta^q({V}^{p,\bullet})$ and
$''\!E_1^{p,q}=H_d^p({V}^{\bullet,q})$. Both the spectral
sequences lie in the first quadrant and converge to the common
limit $H_D(V)$. By Lemma \ref{lem2}  the term  ${}''\!E_1$ is
supported on the $p$-line and $''\!E_1^{2m-1,0}=0$ for all $m\in
\mathbb{N}$. Hence $''\!E_1\simeq {}''\!E_2 \simeq{}''\!E_\infty$
and $H_D^{2n}(\bar V)\simeq H_d^{2n}({V})=0$. On the other hand,
all but one of the nonzero entries of $\{'\!E_1^{p,q}\}$ are
centered on the diagonal $p=q$, as it is shown in Fig.$2b$. The
only nonzero off-diagonal group is given by $'\!E_1^{0,2n-1}\simeq
\mathbb{R}$. Consequently, $'\!E_1\simeq {}'\!E_{n}$. The
triviality of the group $H_D^{2n}(V)=0$ implies that $d_{n}:
{}'\!E_{n}^{0,2n-1}\rightarrow {}'\!E_{n}^{n,n}$ is an isomorphism
of one-dimensional vector spaces. Explicitly, one may readily see
that the space $'\!E^{0,2n-1}_{n}$ is generated by the
$\delta$-cocycle $\mathrm{A}^F_{n}$ and the map $d_{n}$ takes this
cocycle to the $\delta$-cocycle
$\binom{2n-1}{n-1}{\mathbf{C}}_{n}$, which generates
$'\!E_n^{n,n}$. The last fact amounts to the existence of a
sequence of elements $c_k\in {V}^{k,2n-1-k}$ such that
\begin{equation*}
\begin{array}{rcl}
\delta c_n + d c_{n-1}& =&
\binom{2n-1}{n-1}{\mathbf{C}}_{n}\,,\\[3mm]
\delta c_{n-1} + d c_{n-2} &=& 0\,,\\
&\cdots&\\
\delta c_{1} + d \mathrm{A}^F_{n} &=& 0\,,\\[3mm]
\delta \mathrm{A}^F_{n} &=& 0\,.\\
\end{array}
\end{equation*}
Multiplying the left and right hand sides of the upper line by
$\binom{2n-1}{n-1}^{-1}$, we get the statement of Theorem
\ref{dexrep}.

Here are the explicit expressions for the $d$-exact
representatives of the first three classes $[{\mathbf{C}}_1]$,
$[{\mathbf{C}}_2]$, and $[{\mathbf{C}}_3]$:

\begin{equation*}
\begin{array}{lll}
{\mathbf{C}}_1&-&\frac12 \delta\mathbf{F}^0_1=
d\left[\mathrm{Str}(a_0)+ \mathbf{F}^1_1\right]\,,
\\[5mm]
{\mathbf{C}}_2&+&\delta\left[ \mathrm{Str}(a_3 a_0)+\frac14
\mathbf{F}^1_2\right] =d\left[\mathrm{Str}( a_0 a_4- a_0
a_2)+\frac18  \mathbf{F}^2_2\right]\,,
\\[5mm]
{\mathbf{C}}_3&+&\frac34\delta \left[\mathrm{Str}(a_3 a_0 a_4+a_3
a_4 a_0- a_3 a_0
a_2-a_3 a_2 a_0)-\frac{1}{12} \mathbf{F}^2_3\right] \\[5mm]
&= &d \left[\mathrm{Str}(a_0 a_4^2-\frac{1}{2} a_0 a_4 a_2
-\frac{1}{2} a_0 a_2 a_4 +
a_0 a^2_2 +\frac12 a_3 a_0^3 \right.\\[5mm]
&+&\left. \frac{3}{8}a_3 a_1 a_0 +\frac38 a_1 a_3 a_0)-
\frac{1}{48} \mathbf{F}^3_3\right]\,.
\end{array}
\end{equation*}

Since all the Pontryagin characters $[\mathbf{P}_{2m+1}]$ are
known to be zero, the $d$-exact representatives exist for all
classes  $[\mathbf{C}_{2m+1}]$.

\section{Applications and interpretations}

\subsection{Quantum anomalies}
It is a common knowledge that the anomalies appearing  in quantum
field theory have a topological nature. In a wide sense, the term
\textit{anomaly} refers to breaking of a classical gauge symmetry
upon quantization. As a practical matter, the anomalies manifest
themselves as nontrivial BRST cocycles in ghost number 1 or 2
depending on which formalism, Lagrangian BV or Hamiltonian BFV, is
used. These cocycles represent cohomological obstructions to the
solvability of quantum master equations. Below we interpret the
modular class $[\mathrm{A}^0_1]$ of the BRST differential as the
first obstruction to the existence of quantum master action.
Examining the existence problem for the quantum BRST charge we
encounter the universal cocycle $\mathrm{C}_2$, whose association
with anomalies, however, is more complicated.

Throughout this section,  we assume that the reader is familiar
with basics of the BRST theory. The general reference here is
\cite{HT}. A comprehensive  review of quantum anomalies and
renormalization in BV formalism can be found in \cite{Barnich}.

\subsubsection{One-loop anomalies in the BV formalism}
In the  Bata\-lin-Vil\-ko\-vis\-ky  approach to the quantization
of gauge systems, one usually deals  with the odd cotangent bundle
$\Pi T^\ast M$  of an (infinite-dimensional) supermanifold $M$.
The local coordinates $\{x^i\}$ on $M$ are called fields and the
linear coordinates $\{x^\ast_i\}$ on the fibers of $\Pi T^\ast M$
are called antifields. The  total space of $\Pi T^\ast M$, denoted
below by $\mathcal{M}$, is endowed with the canonical
antibracket\footnote{Another name is the odd Poisson bracket. Upon
identification $C^\infty(\mathcal{M})$ with the space of
polyvector fields on $M$, the antibracket becomes the standard
Schouten-Nijenhuis bracket.}
\begin{equation*}\label{}
    (f,g)=(-1)^{\epsilon_i(\epsilon(f)+1)}\frac{\partial f}{\partial x^i}\frac{\partial g}{\partial
    x_i^\ast}-(-1)^{(\epsilon_i+1)(\epsilon(f)+1)}\frac{\partial f}{\partial
    x_i^\ast}\frac{\partial g}{\partial x^i}
\end{equation*}
for all $f,g\in C^{\infty}(\mathcal{M})$. Besides  the Grassman
parity, the structure sheaf of functions on $\mathcal{M}$ is
endowed with an additional $\mathbb{Z}$-grading, called the ghost
number, so that $\mathrm{gh}(x^\ast_i)=-\mathrm{gh}(x^i)-1$. The
presence of the extra $\mathbb{Z}$-grading is inessential, for our
consideration and we will not mention it below.

The Feynman probability amplitude $e^{\frac i\hbar S}$ on the
space of fields and antifields  is defined by the quantum master
action $S\in C^{\infty}(\mathcal{M})\otimes \mathbb{C}[[\hbar]]$.
The latter is given by a formal power series in $\hbar$ with
smooth even coefficients and obeys the quantum master equation
\begin{equation}\label{MEq}
    (S,S)=2i\hbar \Delta S\qquad \Leftrightarrow \qquad \Delta e^{\frac i\hbar
    S}=0\,.
\end{equation}
Here $\Delta: C^{\infty}(\mathcal{M})\rightarrow
C^\infty(\mathcal{M})$ is a second-order differential operator,
called the odd Laplacian. It is defined in terms of a nowhere zero
density $\rho$ on $\mathcal{M}$ by the rule
\begin{equation*}\label{}
\Delta f=\frac12 (-1)^{\epsilon(f)}\mathrm{div}_\rho X_f\,,
\end{equation*}
where $X_f=(f,\,\cdot\,)$ is the Hamiltonian vector field
associated to $f\in C^{\infty}(\mathcal{M})$.  The density $\rho$
is supposed to be chosen in such a way that $\Delta^2=0$. For
example, if $\sigma$ is a nowhere zero density on $M$, then
$\rho=\sigma^2$ is an appropriate  density on $\mathcal{M}$.

A fundamental property of the odd Laplacian is that it
differentiates  the antibracket:
\begin{equation*}\label{}
    \Delta(f,g)=(\Delta f, g)+(-1)^{\epsilon(f)+1}(f,\Delta
    g)\,.
\end{equation*}
This relation is easily derived from  the following one
identifying the antibracket as the defect of the odd Laplacian to
be a differentiation of the commutative algebra of functions:
\begin{equation*}\label{}
    \Delta (f\cdot g)=\Delta f\cdot g+(-1)^{\epsilon
    (f)}(f,g)+(-1)^{\epsilon(f)}f\cdot \Delta g\,.
\end{equation*}
All the above  properties of the antibracket and the odd Laplace
operator allow one to look upon the quantum master equation
(\ref{MEq}) as a sort of Maurer-Cartan equation.

By definition,  the quantum master action  $S=S_0+\hbar S_1 +
\cdots$ can be regarded as a formal deformation of the classical
one $S_0\in C^\infty(\mathcal{M})$. Expanding (\ref{MEq}) in
powers of $\hbar$ yields the sequence of equations
\begin{equation}\label{ME}
\begin{array}{l}
    (S_0,S_0)=0\,,\\[3mm]
    (S_0, S_1)=i\Delta S_0\,,\\
    \displaystyle (S_0,S_n)=i\Delta S_{n-1}+\sum_{k=1}^{n-1}(S_k,S_{n-k})\,,\qquad
    n\geq 2\,.
    \end{array}
\end{equation}
The first equation is known as a classical master equation. Its
solution $S_0$ is completely determined (under some properness and
regularity  conditions) by the classical gauge theory and can be
systematically constructed by means of the homological
perturbation theory \cite{HT}. Given a classical master action
$S_0$, the actual problem is to iterate the higher-order quantum
corrections $S_n$ satisfying (\ref{ME}).

The classical master equation ensures that the Hamiltonian action
of $S_0$ defines a homological vector field $Q=(S_0,\,\cdot\,)$ on
$\mathcal{M}$, called a classical BRST differential. With the BRST
differential Eqs. (\ref{ME}) take the cohomological form
\begin{equation*}
\delta S_n=B_n(S_0,...,S_{n-1})\,.
\end{equation*}
By induction on $n$, one can see that the function $B_n$ is
$\delta$-closed, provided that $S_0,...,S_{n-1}$ satisfy the first
$n$th equations (\ref{ME}). Thus the existence problem for the
quantum master action appears to be equivalent to the vanishing of
a certain sequence of the $\delta$-cohomology classes $[B_n]$ in
which the $n$th cohomology class is defined provided that all the
previous classes vanish. In particular, the second equation in
(\ref{ME}) expresses the triviality of the modular class
$[\mathrm{A}_1^0]$, where $\mathrm{A}^0_1=\Delta S_0
=\frac12\mathrm{div}_\rho Q$. This allows us to identify the
modular class of the classical BRST differential with the first
cohomological obstruction to the solvability of the quantum master
equation.

\subsubsection{Two-loop anomalies in the BFV formalism}
The Hamiltonian analog of the BV field-antifield formalism is
known as the BFV formalism. This time all the information about
the gauge structure of a theory is encoded by the quantum BRST
charge $\Omega$. In the deformation quantization approach
\cite{Fedosov}, one regards  $\Omega$ as an odd element of the
associative algebra $(C^{\infty}(M)\otimes
\mathbb{C}[[\hbar]],\ast)$, where $M$ is a supermanifold endowed
with a non-degenerate Poisson bracket $\{\cdot,\cdot\}$ and the
$\ast$-product on $C^\infty(M)\otimes \mathbb{C}[[\hbar]]$ is
defined  as a local, $\mathbb{C}[[\hbar]]$-linear deformation in
$\hbar$ of the ordinary function multiplication satisfying the
\textit{correspondence principle}
$$
    f\ast g-(-1)^{(\epsilon(f)+1)(\epsilon(g)+1)}g\ast f
    =i\hbar\{f,g\}+\mathcal{O}(\hbar^2)\quad \forall f,g\in C^\infty(M)\,.
$$
(Again, we leave aside the ghost grading on $M$.) By definition,
the charge $\Omega$ obeys  the quantum master equation
\begin{equation}\label{Om2}
    \Omega\ast \Omega =0\,.
\end{equation}

It is well known \cite{BCG} that all the inequivalent
$\ast$-products on a given symplectic manifold $(M,\omega)$ are in
one-to-one correspondence with the elements of the affine space
${[\omega]}/{\hbar}+ H^2(M)\otimes \mathbb{C}[[\hbar]]$. The
points of this space are called characteristic classes of the
$\ast$-product and are denoted by $\mathrm{cl}(\ast)$. Below we
set for definiteness $\mathrm{cl}(\ast)=[\omega]/\hbar$. The
explicit recurrent formulas for Fedosov's  $\ast$-product on
symplectic supermanifolds can be found in \cite{Ber}, \cite{Bor}.
One more fact about the deformation quantization we need below is
that any $\ast$-product on a symplectic manifold is equivalent to
one with the property
\begin{equation*}\label{}
    f\ast_{n}g=(-1)^{(n+\epsilon(f)\epsilon(g))}g\ast_{n}f\,,
\end{equation*}
where $\ast_n$ stands for the bidifferential operator determining
the $n$-th order in the $\hbar$-expansion of $\ast$.

Substituting the general expansion  $\Omega=\sum_{n\geq
0}\hbar^n\Omega_n$ in (\ref{Om2}), we get a (possibly infinite)
sequence of equations
\begin{equation}\label{dOm}
    \{\Omega_0,\Omega_n\}=-\sum_{\footnotesize\begin{array}{c}_{k+l+m=n+1}\\[-1mm]_{l,m<n}\end{array}}
    \Omega_l\ast_k \Omega_m \,.
\end{equation}
The first equation of this sequence, $\{\Omega_0,\Omega_0\}=0$, is
called the classical master equation for the classical BRST charge
$\Omega_0$. The charge $\Omega_0$ is completely  determined by the
first-class constraints of the original Hamiltonian theory and,
similar to the classical master action, it can always be
constructed by means of the homological perturbation theory. The
classical BRST differential $\delta$ is defined now by the
homological vector field $Q=\{\Omega_0,\cdot\;\}$ on $M$. Then the
next equation in (\ref{dOm}) takes the form
\begin{equation}\label{dOm1}
\delta\Omega_1=0\,.
\end{equation}
It identifies the first quantum correction to the classical BRST
charge with a BRST cocycle. Let $\{f,g\}=\langle\Pi,df\wedge
dg\rangle$, where the triangle brackets denote the natural pairing
between bivectors and two-forms. Then the second-order correction
$\Omega_2$ is defined by the equation
\begin{equation}\label{dOm2}
\begin{array}{c}
    \delta\Omega_2 =\langle \Pi, \frac1{48}\mathbf{C}_2 -\frac12 d\Omega_1\wedge
    d\Omega_1\rangle\,,
    \end{array}
\end{equation}
with  $\mathbf{C}_2$ being here the second universal cocycle of
the $\mathrm{C}$-series associated to the classical BRST
differential $Q$. If $\Gamma$ are the local one-forms determining a
symplectic connection $\nabla$, then  $
    \mathbf{C}_2=\mathrm{Str}(\delta\Gamma\wedge\delta\Gamma)$.

Since the homological vector field $Q$ is Hamiltonian, the Poisson
bivector is $Q$-invariant, $\delta \Pi=0$, and the right hand side
of (\ref{dOm2}) is obviously $\delta$-closed (but not
$\delta$-exact in general). Thus the class $[\mathbf{C}_2]$ can
obstruct the solvability of the quantum master equation. More
precisely, we have the following

\begin{prop}
Let $(M,\omega)$ be a symplectic supermanifold endowed with the
Hamiltonian action of a classical BRST differential $\delta$. If
$[\mathbf{C}_2]=0$, then there exists a $\ast$-product on $M$ such
that the quantum master equation (\ref{Om2}) is solvable up to
order three in $\hbar$.
\end{prop}

In case $[\mathbf{C}_2]=0$ we can  set $\Omega_1=0$ and take
$\Omega_2=\langle \Pi, \delta^{-1}\mathbf{C}_2\rangle$. It should
be noted that the above analysis of low-order anomalies is not
complete as we have restricted ourselves to a particular class of
$\ast$-products. For a general $\ast$-product with
$\mathrm{cl}(\ast)=[\Pi^{-1}]/\hbar+[\omega_0]+\hbar[\omega_1]+\cdots$
the right hand sides of equations (\ref{dOm1}) and (\ref{dOm2})
will involve additional $\delta$-closed terms proportional to
$\omega_0$ and $\omega_1$.

\subsection{Characteristic classes of foliations}
Given a regular foliation $\mathcal{F}$ of an ordinary (even)
manifold $N$, denote by $T\mathcal{F}\subset TN$ the subbundle of
tangent spaces to the leaves of $\mathcal{F}$. Since
$T\mathcal{F}$ is integrable, the inclusion map
$T\mathcal{F}\rightarrow TN$ defines a regular Lie algebroid over
$N$. Thus, there is a one-to-one correspondence between the
categories of regular foliations and injective Lie algebroids. On
the other hand, to any Lie algebroid $E\rightarrow TN$ one can
associate a homological vector field on $\Pi E$ (see Example 1.2)
together with the corresponding characteristic classes. When the
Lie algebroid comes from a regular foliation, these characteristic
classes can be attributed to the foliation itself. Furthermore,
the construction of characteristic  classes, being insensitive to
the regularity of the  Lie algebroid structure, can also be used
to the study of singular foliations. The question of whether every
singular foliation corresponds to the characteristic foliation of
a Lie algebroid remains open. Below we give an example of regular
foliation with nontrivial modular class.

Let $SL(2,\mathbb{R})$ denote the group of $2\times 2$-matrices
with real entries and determinant 1.  It is well known that this
group admits discrete subgroups $\Gamma$ such that the right
quotient space $N=SL(2, \mathbb{R})/\Gamma$ is a compact manifold.
Since $SL(2,\mathbb{R})$ has dimension three, so does $N$.

Let $\{e_a\}$ be Weyl's basis in the space of right invariant
vector fields on $SL(2,\mathbb{R})$:
\begin{equation*}\label{}
    [e_{-1},e_1]=2e_0\,,\qquad [e_0,e_1]=e_1\,,\qquad
    [e_0,e_{-1}]=-e_{-1}\,.
\end{equation*}
The canonical projection $\pi: SL(2,\mathbb{R})\rightarrow N$
takes $\{e_a\}$ to the vector fields $\mathbf{e}_a=\pi_\ast(e_a)$
on $N$ satisfying the same commutation relations. (The vector
fields $e_a$ and $\mathbf{e}_a$ are said to be $\pi$-related.) The
pair $\{\mathbf{e}_0, \mathbf{e}_1\}$ generates the action of the
Borel subalgebra $B\subset sl(2,\mathbb{R})$ on $N$.  Hence, we
have a two-dimensional foliation $\mathcal{F}$ of $M$. The
inclusion map $\rho: T\mathcal{F}\rightarrow TN$ defines the
transformation Lie algebroid $E=N\times B$ over $N$. If $x^i$ are
local coordinates on $N$ and $c^a$ are odd coordinates on $\Pi B$,
then the homological vector field on $N\times\Pi B$ reads
\begin{equation*}\label{}
Q=c^0{\bf e}^{i}_0 \frac{\partial}{\partial x^i} +
c^1{\mathbf{e}}_1^i\frac{\partial}{\partial x^i}
-c^{0}c^{1}\frac{\partial}{\partial c^{1}}\,.
\end{equation*}

Notice that the dual to the three-vector $\mathbf{e}_{-1}\wedge
\mathbf{e}_0\wedge \mathbf{e}_1$ is a $SL(2,\mathbb{R})$-invariant
volume form $\sigma$ on $N$. Therefore, $\mathrm{div}_\sigma
\mathbf{e}_a=0$. Taking $\rho= \sigma dc^0dc^1$ to be a nowhere
zero integration density on $N\times \Pi B$, we get the nontrivial
$\delta$-cocycle
\begin{equation}\label{mod}
    \mathrm{A}_1^0= \mathrm{div}_\rho Q= c^0\,.
\end{equation}
If $\mathrm{A}^0_1$ were trivial there would be a function $f\in
C^{\infty}(N)$ such that $\textbf{e}_0f=1$. But every function on
a compact manifold has a critical point $p$ at which
$(\textbf{e}_0f)(p)=0$. Thus the modular class of the foliation
$\mathcal{F}$  is nontrivial.

It is interesting to note that the foliation $\mathcal{F}$  is the
one that was originally used by Roussarie to demonstrate the
nontriviality of the Godbillon-Vey class \cite{GV}, \cite{Fu1}.
It is not an accident  that both the modular and Godbillon-Vey
classes of $\mathcal{F}$ are nonzero. One can see that the modular
class (\ref{mod}) coincides in fact with the Reeb class of
$\mathcal{F}$. The latter takes value in $H^1_\mathcal{F}(N)$, the
first group of the leafwise cohomology. It is well known that the
vanishing of the Reeb class results in the vanishing of the
Godbillon-Vey class but not vice versa in general. Thus, we have
the implication $GV[\mathcal{F}]\neq 0\Rightarrow
A_0^1[\mathcal{F}]\neq 0$. A detailed discussion of the
relationship between the modular and Reeb classes of regular
Poisson manifolds can found in \cite{AB}.

\subsection{Lie algebras}

Let $\mathcal{G}$ be a Lie algebra with a basis $\{t_a\}$ and the
commutation relations
\begin{equation*}\label{}
[t_a,t_b]=f_{ab}^d t_d\,.
\end{equation*}
Then the homological vector field on $\Pi \mathcal{G}$ is given by
\begin{equation}\label{QLA}
    Q= \frac12c^bc^a f_{ab}^d\frac{\partial}{\partial c^d}\,.
\end{equation}
The linear space of functions on $\Pi \mathcal{G}$ endowed with
the differential $\delta$ gives us a model for the
Chevalley-Eilenberg complex of the Lie algebra $\mathcal{G}$. With
a flat connection on $\Pi \mathcal{G}$ we see that the
characteristic classes of $A$-series (\ref{A}) are nothing but the
primitive elements of the Lie algebra cohomology:
\begin{equation*}\label{Adj}
A_n=\mathrm{tr}(\mathrm{ad}_{a_1}\cdots \mathrm{ad}_{a_{2n-1}})
c^{a_1}\cdots c^{a_{2n-1}}\qquad \forall n\in \mathbb{N}\,,
\end{equation*}
with $\mathrm{ad}_a=(f_{ab}^d)$ being the matrices of the adjoint
representation of $\mathcal{G}$.

The universal cocycles of $B$- and  $C$-series are then identified
with $\mathrm{Ad}$-invariant tensors on $\Pi \mathcal{G}$:
\begin{equation*}\label{}\begin{array}{l}
  \displaystyle   B_n =(\mathrm{ad}_{a_1}\cdots \mathrm{ad}_{a_{n}})_{a_{n+1}}^b dc^{a_1}\otimes \cdots \otimes dc^{a_n}\otimes dc^{a_{n+1}}\otimes \frac{\partial}{\partial
    c^b}\,,\\[5mm]
    C_n=\mathrm{tr}(\mathrm{ad}_{a_1}\cdots \mathrm{ad}_{a_{n}})
dc^{a_1}\otimes \cdots \otimes dc^{a_{n}}\qquad \forall n\in
\mathbb{N}\,.
\end{array}
\end{equation*}
Since any coboundary of (\ref{QLA}) is necessarily proportional to
$c^a$, the tensor cocycles are either zero or nontrivial. In case
$\mathcal{G}$ is semi-simple, for instance, the one-form $C_1$ is
zero, while the two-form $C_2$ is non-degenerate (the Killing
metric).


\appendix
\section{The exponential map}
In this Appendix, we reformulate the chain property
(\ref{ChainPr}) of the characteristic map $\widehat{\mathbb{G}}:
TM\rightarrow \mathbb{E}$ as the integrability condition for some
homological vector field on $TM$. As a by-product, we get an
explicit expression for the action of $\delta$ on the basis
covariants $\{\partial^{n}Q\}$ of a flat $Q$-manifold.

Given a smooth manifold $M$, we can treat the  total space of the
tangent bundle $TM$ as a ``partially formal'' manifold, that is,
formal in the directions of fibers. Concerning the general theory
of partially formal supermanifolds we refer the reader to
\cite{Kon1}, \cite{KS}. Let $\mathcal{F}$ denote the commutative
algebra of formal functions on $TM$; the elements of $\mathcal{F}$
are formal power series in fiber coordinates with smooth
coefficients:
\begin{equation*}\label{}
    \mathcal{{F}}\ni f(x,y)=\sum_{n=0}^\infty y^{i_n}\cdots y^{i_1}f_{i_1\cdots
    i_n}(x)\,.
\end{equation*}
The expansion coefficients $f_{i_1\cdots i_n}(x)$ are covariant
symmetric tensor fields on $M$. We define the formal vector fields
on $TM$ to be the derivations  of the commutative algebra
$\mathcal{F}$ and denote the Lie algebra of all the derivations by
$\mathrm{Der}(\mathcal{F})$.

The natural inclusion $i: C^\infty(M)\rightarrow \mathcal{F}$
identifies $C^\infty(M)$ with the subalgebra of $y$-independent
functions of $\mathcal{F}$. Suppose $M$ admits an affine
connection $\partial$ with zero torsion and curvature. Then we can
define one more homomorphism $\phi : C^{\infty}(M)\rightarrow
\mathcal{F}$ known as the \textit{exponential map}. The
exponential map takes a smooth function $f\in C^{\infty}(M)$ to
the formal function
\begin{equation*}\label{}
    \phi(f)=f(x+y)\equiv
    \sum_{n=0}^\infty\frac1{n!}y^{i_n}\cdots y^{i_1}\partial_{i_1}\cdots\partial_{i_n}f(x)\,.
\end{equation*}
It is clear that $\phi (f\cdot g)=\phi (f)\cdot \phi(g)$.

Associated to the exponential map of functions is the exponential
map of vector fields
\begin{equation*}
\phi' : \mathfrak{X}(M)\rightarrow \mathrm{Der}(\mathcal{F})\,.
\end{equation*}
This is defined as follows. Consider $C^\infty(M)$ and
$\mathcal{F}$ as left modules over the Lie algebras
$\mathfrak{X}(M)$ and $\mathrm{Der}(\mathcal{F})$, respectively.
Then for any vector field $X\in \mathfrak{X}(M)$ there is a unique
formal vector field $\widetilde{X}=\phi'(X)\in
\mathrm{Der}(\mathcal{F})$ such that the following diagrams
commute:
$$
\xymatrix{{\mathcal{F}}\ar[r]^{\widetilde{X}} &{\mathcal{F}} \\
{C^{\infty}(M)}\ar[u]^i \ar[r]^X     &  {C^{\infty}(M)}\ar[u]_i
}\qquad \xymatrix{{\mathcal{F}}\ar[r]^{\widetilde{X}} &{\mathcal{F}} \\
{C^{\infty}(M)}\ar[u]^\phi \ar[r]^X     &
{C^{\infty}(M)}\ar[u]_\phi }
$$
In terms of local coordinates adapted to $\partial$ we have
\begin{equation}\label{exp*}
\begin{array}{lll}
   \widetilde{X} &=& \displaystyle  X^i(x)\left(\frac{\partial}{\partial x^i}-\frac{\partial}{\partial
    y^i}\right)+X^i(x+y)\frac{\partial}{\partial y^i}\\[5mm]
  &=& \displaystyle X^i(x)\frac{\partial}{\partial x^i}+\sum_{n=1}^\infty \frac 1{n!}y^{i_n}\cdots y^{i_1}
    \partial_{i_1}\cdots \partial_{i_n}X^i(x)\frac{\partial}{\partial y^i}\,.
    \end{array}
\end{equation}
The pair $(\phi',\phi)$ defines a homomorphism of the
$\mathfrak{X}(M)$-module $C^\infty(M)$ to the
$\mathrm{Der}(\mathcal{F})$-module $\mathcal{F}$:
\begin{equation*}
\phi'([X,Y])=[\phi' (X),\phi'(Y)]\,,\qquad \phi' (X)
\phi(f)=\phi(Xf)\,.
\end{equation*}

Now let us apply the exponential map (\ref{exp*}) to a homological
vector field $Q$ on $M$. The result is a formal homological vector
field $\widetilde{{Q}}$ on $TM$. The latter can be expanded in the
sum of homogeneous components
\begin{equation*}
\widetilde{Q}=\sum_{n=0}^\infty \widetilde{{Q}}_n\,,\qquad
[N,\widetilde{{Q}}_n]=n\widetilde{{Q}}_n\,,
\end{equation*}
where $N=y^i\partial/\partial y^i$ is a vertical vector field on
$TM$ called sometimes the Euler vector field. The integrability of
$\widetilde{Q}$ implies that
\begin{equation}\label{Q-Q}
2[\widetilde{Q}_0,\widetilde{Q}_m]=-\sum_{k=1}^{m-1}
[\widetilde{Q}_{m-k},\widetilde{Q}_k]\qquad \forall m\in
\mathbb{N}\,.
\end{equation}
In particular,
\begin{equation*}
\widetilde{Q}_0=Q^i\frac{\partial}{\partial
x^i}+y^j\partial_jQ^i\frac{\partial}{\partial y^i}
\end{equation*}
is a smooth homological vector field defining a $Q$-structure on
the tangent bundle $TM$.  The other homogeneous components
$\{\widetilde{Q}_n\}_{n=1}^\infty$ of $\widetilde{Q}$ are vertical
vector fields, which are naturally identified with the elementary
local covariants of the flat $Q$-manifold $M$. Upon this
identification  Eq. (\ref{Q-Q})  compactly expresses the action of
$\delta$ on the elementary local covariants $\{\partial^nQ\}$ with
$n>1$. Namely,
\begin{equation}\label{Qmorf}
\delta \widetilde{Q}_m=\sum_{k=1}^{m-1}
[\widetilde{Q}_{m-k},\widetilde{Q}_k]\qquad \forall m\in
\mathbb{N}\,.
\end{equation}
For  the remaining two generators -- the homological vector field
$Q$ itself and its first covariant derivative $\Lambda: =\partial
{Q}\in \frak{A}(M)$ -- we have
\begin{equation}\label{Qmorf1}
\delta Q=0\,,\qquad \delta \Lambda=\Lambda^2\,.
\end{equation}
As is seen, Eq. (\ref{Qmorf}) coincides exactly with
Eq.(\ref{binom}) characterizing the Gauss map (\ref{Gauss})  as a
morphism of $Q$-manifolds.

\section{Scalar characteristic classes}\label{A2}
Let $M$ be a $Q$-manifold endowed with a symmetric affine
connection $\nabla$ and let $\mathcal{A}$ be the differential
tensor algebra of the local covariants associated to $Q$ and
$\nabla$. As a tensor algebra, $\mathcal{A}$ is generated by the
repeated covariant derivatives of the homological vector field and
the curvature tensor. A suitable generating set of  $\mathcal{A}$
is given by the following elementary covariants:
\begin{equation}\label{Q-type}
\begin{array}{l}
    Q^j\,,\quad Q^j_{i}=\nabla_iQ^j\,,\\[3mm]
    Q^j_{i_1\cdots
    i_{n+2}}=\nabla_{(i_1}\cdots \nabla_{i_{n+2})}Q^j
    -\nabla_{(i_1}\cdots\nabla_{i_n}R^j_{i_{n+1}i_{n+2})k}Q^k\,,\\[3mm]
 R^j_{i_{1}\cdots
 i_{n+3}}=\nabla_{i_1}\cdots\nabla_{i_n}R_{i_{n+1}i_{n+2}i_{n+3}}^j\,.\qquad\qquad\quad\qquad\quad\;\;
 \end{array}
\end{equation}
The round brackets denote symmetrization of the enclosed indices.
Notice that the generators are not free and satisfy an infinite
set of tensor relations coming from the integrability condition
for the homological vector field, the Bianchi identities for the
curvature tensor, and all their differential consequences. Of
particular importance for our analysis will be the following
identities:

\begin{equation}\label{Q-id}
\nabla_QQ=0\,,\qquad \nabla_Q\Lambda =-\Lambda^2+\frac12 R_{QQ}\,,
\end{equation}

\begin{equation}\label{R-id}
\begin{array}{c}
    R_{QX}(Q)=-(-1)^{\epsilon(X)\epsilon(Q)}R_{XQ}(Q)=\frac12R_{QQ}(X)\,,\\[5mm]
    R_{QQ}(Q)=0\,,\qquad \nabla_QR_{QQ}=0\,.
    \end{array}
\end{equation}

The action of the differential $\delta=L_Q$ is given by
\begin{equation*}\label{}
\begin{array}{lcl}
\delta Q^j &=&0\,,\\[3mm]
\delta Q^j_i&=&Q_i^kQ_k^j + \frac12R_{QQ\ i}^j\,,\\[3mm]
\delta Q^j_{i_1\cdots
 i_{n+2}}&=&-\sum_{l=2}^{n+1}\binom{n+2}{l} Q^m_{(i_1\cdots i_l} Q^j_{mi_{l+1}\cdots
 i_{n+2})}\,,\\[3mm]
    \delta R^j_{i_1\cdots i_{n+3}}&=&Q^iR_{ii_1\cdots i_{n+3}}^j - R^i_{i_1\cdots
    i_{n+3}}Q_{i}^{j}\\[3mm]
    &+& \sum_{l=1}^{n+3}(-1)^{(\epsilon_{i_1} +
    \cdots+\epsilon_{i_{l-1}})}Q_{i_l}^{i} R^j_{i_1\cdots
    i\cdots i_{n+3}}\,.
    \end{array}
\end{equation*}
As is seen  the generators in the second and third lines of
\eqref{Q-type}, taken separately, generate two differential ideals
of $\mathcal{A}$, which we denote respectively by $\mathcal{Q}$
and  $\mathcal{R}$. The quotient $\mathcal{A}/\mathcal{Q}$ is
naturally isomorphic to the differential subalgebra
$\mathcal{Q}'\subset \mathcal{A}$ generated by $Q^j$, $Q^j_i$, and
$\{R^i_{i_1\cdots i_{n+3}}\}$, so that the complex $\mathcal{A}$
breaks up into the direct sum of subcomplexes
\begin{equation*}\label{}
    \mathcal{A}=\mathcal{Q}\oplus \mathcal{Q}'\,.
\end{equation*}
Notice that $\mathrm{B}_n,\mathrm{C}_n\in \mathcal{Q}$, while
$P_n\in \mathcal{Q}'$. Yet another decomposition is due to the
tensor type of cochains the complex $\mathcal{A}$ consists of:
\begin{equation*}
\mathcal{A}=\bigoplus_{n,m}\mathcal{A}^{n,m}\,,\qquad
\mathcal{A}^{n,m}=\mathcal{A}\cap \mathcal{T}^{n,m}(M)\,.
\end{equation*}

The elements of $\mathcal{A}$ can also be depicted graphically as
directed, decorated graphs with ``black'' and ``white'' vertices.
The graphs are constructed by gluing together  incoming and
outgoing legs of the corollas
\begin{equation*}\label{}
\begin{split}
\unitlength 1mm 
\begin{picture}(60,20)(0,-3)
\put(-20,7){$Q^j_{i_1\cdots i_n} =$}
\put(0,0){\line(1,1){7.5}}\put(4,0){\line(1,2){3.75}}\put(16,0){\line(-1,1){7.5}}
\put(0,0){\vector(1,1){3.75}}\put(4,0){\vector(1,2){1.75}}\put(16,0){\vector(-1,1){3.75}}
\put(8,8){\circle*{1.5}} \put(8,8){\line(0,1){7.5}}
\put(8,8){\vector(0,1){3.75}} \put(7,0){{\scriptsize$\ldots$}}
\put(-1.5,-2){{\scriptsize${}_{i_1}$}}
\put(2.5,-2){{\scriptsize${}_{i_2}$}}
\put(16,-2){{\scriptsize${}_{i_n}$}}
\put(9,15){{\scriptsize${}_{j}$}}
\put(25,7){and} \put(40,7){$R^j_{i_{1}\cdots i_{n+3}}=$}
\put(60,0){\line(1,1){7.5}}\put(64,0){\line(1,2){3.7}}\put(76,0){\line(-1,1){7.5}}
\put(60,0){\vector(1,1){3.75}}\put(64,0){\vector(1,2){1.7}}\put(76,0){\vector(-1,1){3.75}}
\put(68,8){\circle{1.5}} \put(68,8.75){\line(0,1){7.5}}
\put(68,8.75){\vector(0,1){3.75}}
\put(67,0){{\scriptsize$\ldots$}}
\put(58.5,-2){{\scriptsize${}_{i_1}$}}
\put(62.5,-2){{\scriptsize${}_{i_2}$}}
\put(74.5,-2){{\scriptsize${}_{i_{n+3}}$}}
\put(69,15){{\scriptsize${}_{j}$}}
\end{picture}
\end{split}
\end{equation*}
by the general rules discussed in Sec. \ref{StGr}. The planarity
of the corollas  allows us to order the incoming legs by reading
them anticlockwise with respect to the vertex so that the first
incoming leg appears to be the nearest one to the outgoing leg
from the left. This order, however, is only crucial for the
correspondence between the white corollas and the $R$-generators,
as the $Q$-generators  are fully symmetric in the lower indices.

Let us now clarify the structure of the differential subalgebra
$\mathcal{A}^{0,0}\subset \mathcal{A}$.  Denote by
$\bar{\mathcal{A}}^{0,0}= \mathcal{A}^{0,0}/(\mathcal{A}^{0,0})^2$
the subspace of indecomposable elements of $\mathcal{A}^{0,0}$.
Then  $\mathcal{A}^{0,0}=\bigoplus_{k\in \mathbb{N}}
(\bar{\mathcal{A}}^{0,0})^k$ and we have the following

\begin{prop}
In the stable range of dimensions, the complex
$\bar{\mathcal{A}}^{0,0}$ is isomorphic to the complex
$\widehat{W}$ of cyclic words from the proof of Theorem \ref{A-P}.
\end{prop}

\begin{proof} The elements of $\bar{\mathcal{A}}^{0,0}$ correspond to linear combinations of connected graphs
without legs. If $\Gamma$ is such a graph, then it has the same
number of vertices and edges. Therefore, $\Gamma$ contains exactly
one cycle. The arrows of this cycle are all directed into one
side, and the remaining edges are directed towards the cycle,
forming trees growing from the cyclic vertices (possibly
several from a single vertex). The ``crown'' of each tree is made
of univalent black vertices that may join either a cyclic or a
non-cyclic vertex. In the latter case $\Gamma$ contains one of the
following subgraphs:
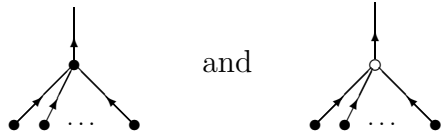
\begin{equation*}\label{}
\begin{split}
\unitlength 1mm 
\begin{picture}(60,20)(0,-3)
\put(0,0){\line(1,1){7.5}}\put(4,0){\line(1,2){3.75}}\put(16,0){\line(-1,1){7.5}}
\put(0,0){\vector(1,1){3.75}}\put(4,0){\vector(1,2){1.75}}\put(16,0){\vector(-1,1){3.75}}
\put(8,8){\circle*{1.5}} \put(8,8){\line(0,1){7.5}}
\put(8,8){\vector(0,1){3.75}} \put(7,0){{\scriptsize$\ldots$}}
\put(0,0){\circle*{1.5}} \put(4,0){\circle*{1.5}}
\put(16,0){\circle*{1.5}} \put(25,7){and}
\put(40,0){\line(1,1){7.5}}\put(44,0){\line(1,2){3.7}}\put(56,0){\line(-1,1){7.5}}
\put(40,0){\vector(1,1){3.75}}\put(44,0){\vector(1,2){1.7}}\put(56,0){\vector(-1,1){3.75}}
\put(48,8){\circle{1.5}} \put(48,8.75){\line(0,1){7.5}}
\put(48,8.75){\vector(0,1){3.75}}
\put(47,0){{\scriptsize$\ldots$}} \put(40,0){\circle*{1.5}}
\put(44,0){\circle*{1.5}} \put(56,0){\circle*{1.5}}
\end{picture}
\end{split}
\end{equation*}
But in view of the symmetry of the $Q$-generators in lower indices
and the deferential consequences of the identities (\ref{Q-id}),
(\ref{R-id}) all these subgraphs correspond to zero elements of
$\mathcal{A}$. Thus the non-vanishing  elements of
$\bar{\mathcal{A}}^{0,0}$ are represented by graphs having no
non-cyclic vertices other than univalent black vertices. The cyclic
vertices can also produce vanishing subgraphs if the algebraic
identities (\ref{Q-id}), (\ref{R-id}) are allowed for. Relations
(\ref{Q-id}) and the symmetry of $Q$-type generators force us to
set
\begin{equation*}\label{}
\begin{split}
\unitlength 0.9mm 
\begin{picture}(60,18)(-15,-3)
\put(-33,8){\line(1,0){8}}\put(-33,2){\line(0,1){6}}\put(-41,8){\line(1,0){8}}
\put(-33,8){\vector(1,0){4}}\put(-33,2){\vector(0,1){4}}\put(-41,8){\vector(1,0){4}}
\put(-33,8){\circle*{1.5}} \put(-33,2){\circle*{1.5}}
\put(-23,7){$=$} \put(-44.5,7){$\frac12$}
\put(-33.5,10){{\scriptsize${}_{2}$}}
\put(-32,0){{\scriptsize${}_{1}$}}
\put(-10,8){\line(1,0){8}}\put(-2,8){\line(1,0){8}}\put(-18,8){\line(1,0){8}}
\put(-10,8){\vector(1,0){4.5}}\put(-2,8){\vector(1,0){5}}\put(-18,8){\vector(1,0){4}}
\put(-10,8){\circle*{1.5}} \put(-2,8){\circle*{1.5}}
\put(-2.5,10){{\scriptsize${}_{1}$}}
\put(-10.5,10){{\scriptsize${}_{2}$}} \put(8,7){$+$}
\put(13,8){\line(1,0){7,25}}\put(21.75,8){\line(1,0){8}}\put(18,2){\line(1,2){2.7}}\put(24,2){\line(-1,2){2.7}}
\put(13,8){\vector(1,0){4}}\put(21.75,8){\vector(1,0){5}}\put(18,2){\vector(1,2){1.7}}\put(24,2){\vector(-1,2){1.7}}
\put(21,8){\circle{1.5}}
\put(18,2){\circle*{1.5}}\put(24,2){\circle*{1.5}}
\put(16,0){{\scriptsize${}_{1}$}}
\put(25,0){{\scriptsize${}_{2}$}} \put(31,7){$\,,$}
\put(42,2){\line(1,1){5.5}}\put(45,2){\line(1,2){3.25}}\put(54,2){\line(-1,1){5.5}}
\put(42,2){\vector(1,1){3.5}}\put(45,2){\vector(1,2){1.6}}\put(54,2){\vector(-1,1){3.5}}
\put(48,8){\circle*{1.5}} \put(48,8){\line(0,1){6.5}}
\put(48,8){\vector(0,1){3.5}} \put(47.5,2){{\scriptsize$\ldots$}}
\put(45,2){\circle*{1.5}} \put(54,2){\circle*{1.5}}
\put(44,1){$\underbrace{\rule{10mm}{0mm}}_{n>1}$}
\put(59,7){$=0\;.$}
\end{picture}
\end{split}
\end{equation*}
Further, taking the iterated covariant derivatives of relations
(\ref{R-id}), we arrive at the following graph equalities:
\begin{equation*}\label{}
\begin{split}
\unitlength 0.9mm 
\begin{picture}(60,18)(33,-2)
\put(12,0){\line(-1,2){3.7}}\put(4,0){\line(1,2){3.7}}\put(15,4.5){\line(-2,1){6.25}}\put(15,8){\line(-1,0){6.25}}\put(15,11.5){\line(-2,-1){6.25}}
\put(12,0){\vector(-1,2){2}}\put(4,0){\vector(1,2){2}}\put(15,4.5){\vector(-2,1){4}}\put(15,8){\vector(-1,0){4}}\put(15,11.5){\vector(-2,-1){4}}
\put(8,8){\circle{1.5}} \put(8,8.75){\line(0,1){7.5}}
\put(8,8.75){\vector(0,1){3.75}} \put(6,0){{\scriptsize$\ldots$}}
\put(4,0){\circle*{1.5}}\put(15,8){\circle*{1.5}}
\put(12,0){\circle*{1.5}}\put(15,4.5){\circle*{1.5}}
\put(20,7){$= \ - $}
\put(39,0){\line(-1,2){3.7}}\put(31,0){\line(1,2){3.7}}\put(42,4.5){\line(-2,1){6.25}}\put(42,8){\line(-1,0){6.25}}\put(42,11.5){\line(-2,-1){6.25}}
\put(39,0){\vector(-1,2){2}}\put(31,0){\vector(1,2){2}}\put(42,4.5){\vector(-2,1){4}}\put(42,8){\vector(-1,0){4}}\put(42,11.5){\vector(-2,-1){4}}
\put(35,8){\circle{1.5}} \put(35,8.75){\line(0,1){7.5}}
\put(35,8.75){\vector(0,1){3.75}}
\put(33,0){{\scriptsize$\ldots$}}
\put(31,0){\circle*{1.5}}\put(42,11.5){\circle*{1.5}}
\put(39,0){\circle*{1.5}}\put(42,4.5){\circle*{1.5}}
\put(47,7){$= \ \frac12 $}
\put(66,0){\line(-1,2){3.7}}\put(58,0){\line(1,2){3.7}}\put(69,4.5){\line(-2,1){6.25}}\put(69,8){\line(-1,0){6.25}}\put(69,11.5){\line(-2,-1){6.25}}
\put(66,0){\vector(-1,2){2}}\put(58,0){\vector(1,2){2}}\put(69,4.5){\vector(-2,1){4}}\put(69,8){\vector(-1,0){4}}\put(69,11.5){\vector(-2,-1){4}}
\put(62,8){\circle{1.5}} \put(62,8.75){\line(0,1){7.5}}
\put(62,8.75){\vector(0,1){3.75}}
\put(60,0){{\scriptsize$\ldots$}}
\put(58,0){\circle*{1.5}}\put(69,11.5){\circle*{1.5}}
\put(66,0){\circle*{1.5}}\put(69,8){\circle*{1.5}}
\put(74,7){$= 0\,, $}
\put(102,0){\line(-1,2){3.7}}\put(94,0){\line(1,2){3.7}}\put(105,4.5){\line(-2,1){6.25}}\put(105,8){\line(-1,0){6.25}}\put(105,11.5){\line(-2,-1){6.25}}
\put(102,0){\vector(-1,2){2}}\put(94,0){\vector(1,2){2}}\put(105,4.5){\vector(-2,1){4}}\put(105,8){\vector(-1,0){4}}\put(105,11.5){\vector(-2,-1){4}}
\put(98,8){\circle{1.5}} \put(98,8.75){\line(0,1){7.5}}
\put(98,8.75){\vector(0,1){3.75}}
\put(96,0){{\scriptsize$\ldots$}} \put(105,4.5){\circle*{1.5}}
\put(105,8){\circle*{1.5}} \put(105,11.5){\circle*{1.5}} \put(108,
7){$=0\,.$}
\end{picture}
\end{split}
\end{equation*}

All the relations above severely restrict the possible form of a
cyclic vertex, leaving in fact only two nontrivial options:

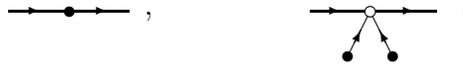
\begin{equation*}\label{ab}
\begin{split}
\unitlength 1mm 
\begin{picture}(60,7)(0,-7)
\put(0,0){\line(1,0){8}}\put(8,0){\line(1,0){8}}
\put(0,0){\vector(1,0){4}}\put(8,0){\vector(1,0){5}}
\put(8,0){\circle*{1.5}}\put(18,-0.5){{,}}
\put(40,0){\line(1,0){7,25}}\put(48.75,0){\line(1,0){8}}\put(45,-6){\line(1,2){2.7}}\put(51,-6){\line(-1,2){2.7}}
\put(40,0){\vector(1,0){4}}\put(48.75,0){\vector(1,0){5}}\put(45,-6){\vector(1,2){1.7}}\put(51,-6){\vector(-1,2){1.7}}
\put(48,0){\circle{1.5}}
\put(45,-6){\circle*{1.5}}\put(51,-6){\circle*{1.5}}
\put(60,-0,5){.}
\end{picture}
\end{split}
\end{equation*}
These vertices correspond to the tensors $\Lambda$ and
$\mathrm{R}=R_{QQ}$, which are already algebraically independent
in the stable range of dimensions. Identifying these tensors with
the generators of the cyclic space (\ref{rep}), we get the desired
isomorphism between the complexes $\bar{\mathcal{A}}^{0,0}$ and
${\widehat{W}}$.
\end{proof}

As immediate corollaries from the proposition above we have

\begin{cor}
The complex of  scalar covariants $\mathcal{A}^{0,0}$ is acyclic.
\end{cor}

Indeed, from the proof of Theorem \ref{A-P} we know that the
complex ${\widehat{W}}\simeq \bar{\mathcal{A}}^{0,0}$ is acyclic.
Applying the K\"unneth formula to
$\mathcal{A}^{0,0}=\bigoplus_{k\in \mathbb{N}}
(\bar{\mathcal{A}}^{0,0})^{\otimes k}$ yields the statement.

\begin{cor}
The functions $P_n$ are nontrivial $\delta$-cocycles of the
subcomplex $\mathcal{R}\subset \mathcal{A}$.
\end{cor}

Suppose the statement were false. Then we could find  a function
$f_n\in \mathcal{R}$ such that $\delta f_n=P_n$. By Theorem
\ref{A-P}, the function
\begin{equation*}\label{}
\begin{array}{c}
    \mathrm{A}_n(\Lambda, \mathrm{R}) -\binom{2n-1}{n}f_n
    \end{array}
\end{equation*}
would be then a  nontrivial intrinsic $\delta$-cocycle, which
contradicts acyclicity of $\bar{\mathcal{A}}^{0,0}$.

\section{The metric connection and $A$-series}\label{A3}

\subsection*{The proof of Proposition \ref{6.1}} Due to the fundamental
classification theorem for smooth supermanifolds \cite{B},  we can
identify $M$ with the total space of an  odd vector bundle $\pi:
E\rightarrow M_0$, where $M_0$ is the body of $M$. Consider the
triple $(g^0,g^1,\nabla^1)$ consisting of a Riemannian metric
$g^0$ on $M_0$, a Euclidean metric $g^1$ on $\Pi E$, and a
connection $\nabla^1$ on $E$. Without loss in generality we can
assume $\nabla^1$ to be compatible with $g^1$ (otherwise replace
$\nabla^1$ by $\nabla^1-\frac12(g^1)^{-1}\nabla^1 g^1$). Notice
also that the metric $g^1$ on $\Pi E$ defines and is defined by  a
fiberwise symplectic structure on $E$.

Let $\{x^i\}$ be a coordinate system in a trivializing chart
$U\subset M_0$ and let $\{\theta^a\}$ be odd coordinates dual to
some frame $\{e_a\}$ in  $E|_{U}$. If
$\nabla^1(e_a)=dx^i\Gamma_{ia}^b e_b$, then the local vector
fields
\begin{equation}\label{vh}
    v_a = \frac{\partial}{\partial
    \theta^a}\quad\mbox{and}\quad
    h_i=\frac{\partial}{\partial x^i}-
    \theta^a\Gamma_{ia}^b\frac{\partial}{\partial \theta^b}
\end{equation}
span, respectively, the subspaces of vertical and horizontal
vector fields of $T(E|_U)$. Define the affine connection
$\widehat{\nabla}$ on the total space of $E|_U$ by setting
\begin{equation}\label{aff-con}
    \widehat{\nabla}_{h_i}h_j=\Gamma_{ij}^k h_k\,,\quad
   \widehat{ \nabla}_{h_i}v_a=\Gamma_{ia}^bv_b\,,\quad
    \widehat{\nabla}_{v_a}h_i=0\,,\quad \widehat{\nabla}_{v_a}v_b=0\,,
\end{equation}
where $\Gamma_{ij}^k$ are the Christoffel symbols of a unique
symmetric connection $\nabla^0$ compatible with the metric $g^0$.
Relations (\ref{aff-con}) are obviously form invariant under the
coordinate changes
\begin{equation*}\label{}
x^{i'}=x^{i'}(x) \,,\qquad \theta^{a'}=\phi^{a'}_b(x)\theta^b
\end{equation*}
on all nonempty intersections $E|_{U}\cap E|_{U'}$; hence
$\widehat{\nabla}$ is a well-defined affine connection on the
whole $M$. The connection $\widehat{\nabla}$ is not symmetric
unless $\nabla^1$ is flat. The nonzero components of the torsion
tensor $T$ are given by
\begin{equation*}\label{}
    T_{h_i
    h_j}=\widehat{\nabla}_{h_i}h_j-\widehat{\nabla}_{h_j}h_i-[h_i,h_j]=\theta^aR_{ija}^bv_b\,,
\end{equation*}
where $\{R_{ija}^b\}$ is the curvature tensor of $\nabla^1$.
Subtracting torsion from $\widehat{\nabla}$ yields a symmetric
connection $\nabla=\widehat{\nabla} -\frac12T$ on $M$.

In the frame (\ref{vh}), the nonzero components of the curvature
tensor of $\nabla$ are collected to the following supermatrices:
\begin{equation}\label{SM}
   R_{_{h_i h_j}}= \left(%
\begin{tabular}{c|c}
  $R_{ijk}^l$&$\frac12\theta^a \bar\nabla_k R_{ija}^b$ \\
  \cline{1-2}
  $0$ & $R_{ijb}^a$ \\
\end{tabular}%
\right) \,,\quad R_{v_ah_i}=\left(%
\begin{tabular}{c|c}
  $0$ & $R_{ija}^b$ \\
  \cline{1-2}
  $0$ & $0$ \\
\end{tabular}%
\right)\,,
\end{equation}
where $\bar\nabla=\nabla^0\oplus\nabla^1$ is the connection on
$TM_0\oplus E$ and $\{R_{ijk}^l\}$ is the curvature tensor of
$\nabla^0$. Since the supermatrices have the block-triangular
form, we readily get
\begin{equation*}\label{}
    \mathbf{P}_{2m+1}^{\nabla}=
    \pi_\ast(\mathbf{P}_{2m+1}^{\nabla^0}-\mathbf{P}_{2m+1}^{\nabla^1})=0\,.
\end{equation*}
Here we used the definition of the supertrace and the fact that
both $\nabla^0$ and $\nabla^1$ are metric connections. Thus,
$\nabla$ is a desired connection.

\subsection*{The proof of Proposition \ref{6.2}}
Let $\nabla $ and $\overline\nabla$ be two metric  connections
associated to the triples $(g^0, g^1,\nabla^1)$ and $(\overline
g^0,\overline g^1,\overline\nabla{}^1)$.  Define the triple
$(g^0_t, g^1_t, \stackrel{_t}{\nabla}{\!\!}^1)$, where
$$g_t^0=(1-t)g^0+t\overline g^0\,,\qquad
g_t^1=(1-t)g^1+t\overline g^1\,,$$ and
$$
\begin{array}{c}
\stackrel{_t}{\nabla}{\!\!}^1=(1-t)\nabla^1+t\overline\nabla^1
-\frac12(g_t^1)^{-1}[(1-t)\nabla^1+t\overline\nabla{}^1]g^1_t
\end{array}
$$
is a one-parameter family of connections interpolating between
$\nabla^1$ and $\overline\nabla{}^1$ as $t$ runs the interval
$[0,1]$. By definition, $\stackrel{_t}{\nabla}{\!\!}^1g_t^1=0$.
Denote by $\stackrel{_t}{\nabla}{\!\!}^0$ the symmetric connection
compatible with $g_t^0$, and let $\stackrel{_t}{\nabla}$ be the
symmetric connection associated to the data
$(g_t^0,g_t^1,\stackrel{_t}{\nabla}{\!\!}^1)$.

The rest of the proof runs in much the same way as the proof of
Theorem \ref{independing thm}. Namely, we introduce the product
manifold $\widetilde{M}=M\times \mathbb{R}^{1|1}$ endowed with the
homological vector field $\widetilde{Q}=Q+\theta\partial_t$ and
the connection $\widetilde{\nabla}=\stackrel{_t}{\nabla}\oplus
\nabla'$, where $\nabla'$ is the standard flat connection on
$\mathbb{R}^{1|1}$. The matrices
$\widetilde{\Lambda}=\widetilde{\nabla}\widetilde{Q}$ and
$\widetilde{\mathrm{R}}=[\widetilde{\nabla}_{\widetilde{Q}},\widetilde{\nabla}_{\widetilde{Q}}]$
have the form
$$
\widetilde{\Lambda}=\left(
\begin{tabular}{c|c}
   $\Lambda_t$ &   0 \\
   \cline{1-2}
   0 &$ \begin{array}{cc}
  0 & 1 \\
  0 & 0
\end{array}$
\end{tabular}
\right)\,,
 \qquad \widetilde{\mathrm{R}}=\left(
\begin{tabular}{c|c}
   $\mathrm{R}_t+\theta \Psi_t$ & 0 \\
  \cline{1-2}
  0 &$
\begin{array}{cc}
  0 & 0 \\
  0 & 0
\end{array}$
\end{tabular}
\right)\,,
$$
where
$$
\begin{array}{c}
\Lambda_t=\stackrel{_t}{\nabla}\!\! Q\,,\qquad
\Psi_t=\partial_t(\stackrel{_t}{\nabla}_{\!Q})\,,\qquad
\mathrm{R}_t=[\stackrel{_t}{\nabla}_{\!Q},\stackrel{_t}{\nabla}_{\!Q}]\,.
\end{array}
$$
Taking into account the identities
\begin{equation*}
\partial_t\mathrm{R}_t=\stackrel{_t}{\nabla}_{\!Q}\Psi_t \,,
\qquad \stackrel{_t}{\nabla}_{\!Q} \mathrm{R}_t=0\,,\qquad
\mathbf{P}_{2m+1}^{\stackrel{_t}{\nabla}}(Q)=0\,,
\end{equation*}
 one
can easily check that
$$
    \mathbf{P}_{2m+1}^{\widetilde{\nabla}}(\widetilde{Q})=\widetilde{Q}F_{2m+1}\,,\qquad
    F_{2m+1}=(2m+1)\int_{0}^t
    \mathrm{Str}((\mathrm{R}_s)^{2m}\Psi_s)ds\,.
$$
Similar to the curvature (\ref{SM}) the supermatrix $\Psi_t$ has
the block-triangular form
\begin{equation*}\label{}
    \Psi_t=\left(%
\begin{array}{c|c}
  \Psi^0_t & \ast \\  \cline{1-2}
  0 & \Psi^1_t \\
\end{array}%
\right)\,, \qquad \Psi_t^a=
\partial_t(\stackrel{_t}{\nabla}{\!}^a_{\!\!Q})\,, \qquad a=0,1\,.
\end{equation*}
This allows us to rewrite  the function $F_{2m+1}$ as
\begin{equation*}\label{}
    F_{2m+1}=(2m+1)\int_{0}^t[\mathrm{tr}((\mathrm{R}^0_s)^{2m}\Psi^0_s) - \mathrm{tr}((\mathrm{R}^1_s)^{2m}\Psi^1_s)]ds\,,
\end{equation*}
where
$\mathrm{R}^a_t=[\stackrel{_t}{\nabla}{\!}^a_{\!\!Q},\stackrel{_t}{\nabla}{\!}^a_{\!\!Q}]$.
Denoting  $I^a_t = (g_t^a)^{-1}\partial_t g_t^a$ and using the
obvious identities
\begin{equation*}\label{}
(g_t^a)^{-1}\Psi^a_t g_t^a =-\stackrel{_t}{\nabla}{\!}^a_{\!\!Q}
I^a_t\,,\qquad (g_t^a)^{-1}\mathrm{R}^a_t
g_t^a=-\mathrm{R}^a_t\,,\qquad
\stackrel{_t}{\nabla}{\!}^a_{\!\!Q}\mathrm{R}^a_t=0
\end{equation*}
(the first one is obtained by differentiating the identity
$\stackrel{_t}{\nabla}{\!}^a_{\!\!Q} g_t^a=0$), we get
\begin{equation}\label{FU}
\begin{array}{c}
\displaystyle
F_{2m+1}=Q \int_0^t U_{2m+1}(s)ds\,, \\[7mm]
U_{2m+1}(t)=(2m+1)\left[
\mathrm{tr}((\mathrm{R}_t^0)^{2m}I^0_t)-\mathrm{tr}((\mathrm{R}_t^1)^{2m}I^1_t)\right]\,.
\end{array}
\end{equation}

Now substituting $\widetilde{Q}$ and $\widetilde{\nabla}$ to  the
general formula (\ref{A-F}), we obtain  the following
$\widetilde{Q}$-invariant function on $\widetilde{M}$:
\begin{equation*}\label{}
\begin{array}{lll}
    \widetilde{\mathrm{A}}^{F}_{2m+1}&=&\mathrm{A}_{2m+1}(\widetilde{\Lambda},
    \widetilde{\mathrm{R}})-\binom{4m+1}{2m+1}F_{2m+1}\\[5mm]
    &=&\mathrm{A}_{2m+1}(\Lambda_t, \mathrm{R}_t+\theta\Psi_t)-
    \binom{4m+1}{2m+1}F_{2m+1}\\[5mm]
    &=&\mathrm{A}_{2m+1}(\Lambda_t, \mathrm{R}_t)-\binom{4m+1}{2m+1}F_{2m+1}
    +\theta W_{2m+1}\,.
    \end{array}
\end{equation*}
In view of equation (\ref{FU}), the identity
$\widetilde{Q}\widetilde{\mathrm{A}}^{F}_{2m+1}=0$ amounts to

\begin{equation*}\label{tr}
\begin{array}{l}
Q\mathrm{A}_{2m+1}(\Lambda_t,\mathrm{R}_t)=0\,,\\[3mm]
\partial_t
\mathrm{A}_{2m+1}(\Lambda_t,
\mathrm{R}_t)=Q\left[W_{2m+1}+\binom{4m+1}{2m+1}U_{2m+1}\right]\,.
\end{array}
\end{equation*}
Integrating the last equality with respect to  $t$ from $0$ to
$1$, we conclude that $\delta$-cocycles
$\mathrm{A}_{2m+1}(\Lambda_0, \mathrm{R}_0)$ and
$\mathrm{A}_{2m+1}(\Lambda_1, \mathrm{R}_1)$ associated to the
metric connections $\nabla$ and $\overline{\nabla}$ are
cohomologous, and the proof is complete.

\end{document}